\documentclass[11pt]{article}

\usepackage[utf8x]{inputenc}
\usepackage[linesnumbered,lined,boxed,commentsnumbered]{algorithm2e}
\usepackage{amsmath}
\usepackage{amssymb}
\usepackage{amsthm}
\usepackage{bm}
\usepackage{enumitem}
\usepackage{etoolbox}
\usepackage{extsizes}
\usepackage{framed}
\usepackage{fancyhdr}
\usepackage[pdftex]{graphicx}
\usepackage[hidelinks]{hyperref}
\usepackage[margin=0.7in]{geometry}
\usepackage{mathtools}
\usepackage{mnsymbol}
\usepackage{pdfpages}
\usepackage{setspace}
\usepackage{thmtools}
\usepackage{thm-restate}
\usepackage{tikz}
\usetikzlibrary{arrows}
\usepackage{verbatim}
\usepackage{xargs}
\usepackage{yfonts}
\usepackage{cleveref}
\usepackage{url}
\usepackage{tocloft}

\newcommand{\ceil}[1]{\left\lceil #1\right\rceil}

\newcommand{\tell}{\tilde{\ell}}
\newcommand{\vv}{\bm{v}}
\newcommand{\vu}{\bm{u}}
\newcommand{\vw}{\bm{w}}

\newcommand{\va}{\bm{a}}
\newcommand{\vb}{\bm{b}}
\newcommand{\vc}{\bm{c}}
\newcommand{\ve}{\bm{e}}

\newcommand{\valpha}{\bm{\alpha}}
\newcommand{\vbeta}{\bm{\beta}}
\newcommand{\vL}{\bm{L}}
\newcommand{\vzero}{\bm{0}}

\newcommand{\xn}{\bm{x}}
\newcommand{\by}{\bm{y}}
\newcommand{\bw}{\bm{w}}
\newcommand{\bz}{\bm{z}}

\newcommand{\n}{\mathbb{N}}

\newcommand{\C}{\mathbb{C}}
\newcommand{\Z}{\mathbb{Z}}

\newcommand{\field}{\mathbb{F}}
\newcommand{\f}{\field}
\newcommand{\F}{\mathbb{F}}
\newcommand{\kk}{\mathbb{K}}
\newcommand{\cC}{\mathcal{C}}
\newcommand{\cH}{\mathcal{H}}
\newcommand{\cF}{\mathcal{F}}
\newcommand{\cU}{\mathcal{U}}
\newcommand{\R}{\mathbb{R}}

\newcommand{\norm}[1]{\left\lVert#1\right\rVert}

\newcommand{\partiald}[2]{\frac{\partial #1}{\partial #2}}

\newcommand{\gen}[1]{\left\langle #1\right\rangle}

\newcommand\restr[2]{{
  \left.\kern-\nulldelimiterspace 
  #1 
  \vphantom{\big|} 
  \right|_{#2} 
  }}
  
 \newcommand{\trace}{\ensuremath{\text{Trace}}}

\newcommand{\szlem}{DeMillo-Lipton-Schwartz-Zippel lemma}

\newcommand{\bra}[1]{\ensuremath{\left( #1 \right) }}
\newcommand{\cbra}[1]{\ensuremath{\left\{ #1 \right\} }}

\newcommand{\matrixgroup}[3]{\ensuremath{\text{#1}_{#2}(#3)}}
\newcommand{\affmatrixgroup}[3]{\ensuremath{\text{#1}^{\text{aff}}_{#2}(#3)}}
\newcommandx{\matrixgroupnofield}[2]{\ensuremath{\text{#1}_{#2}}}
\newcommandx{\GL}[2][1=n,2=\C]{\matrixgroup{GL}{#1}{#2}}
\newcommandx{\GLm}[2]{\matrixgroup{GL}{#1}{#2}}
\newcommandx{\GLplus}[2][1=n,2=\C]{\affmatrixgroup{GL}{#1}{#2}}
\newcommandx{\GLmplus}[2][1=m,2=\C]{\affmatrixgroup{GL}{#1}{#2}}
\newcommandx{\SL}[2][1=n,2=\C]{\matrixgroup{SL}{#1}{#2}}
\newcommandx{\PS}[2][1=n,2=\C]{\matrixgroup{PS}{#1}{#2}}
\newcommandx{\TR}[2][1=n,2=\C]{\matrixgroup{TR}{#1}{#2}}
\newcommandx{\TS}[2][1=n,2=\C]{\matrixgroup{TS}{#1}{#2}}
\newcommandx{\TPS}[2][1=n,2=\C]{\matrixgroup{TPS}{#1}{#2}}
\newcommandx{\GLaff}{\ensuremath{\text{GL}^{\text{aff}}}}

\newcommand{\gl}[2]{\text{GL}_{#1}(#2)}
\newcommand{\gla}[2]{\affmatrixgroup{GL}{#1}{#2}}

\newcommand{\croanf}[1][\Delta]{\ensuremath{\text{ANF}_{#1}}}
\newcommandx{\sdmpol}[2][1=s,2=d]{\ensuremath{\text{T}_{#1,#2}}}

\renewcommand{\det}{\ensuremath{\text{Det}}}

\newcommand{\generator}{\ensuremath{\mathcal{G}}}
\newcommand{\G}{\generator}
\newcommand{\svgen}{\ensuremath{\G^{\text{SV}}}}
\newcommand{\svgenh}{\ensuremath{\G^{\text{SV-hom}}}}

\newcommand{\cont}{\text{C}}
\newcommand{\ccont}[1]{{\text{C}}^{\gla{}{#1}}}
  
\newcommand{\dtwo}{\ensuremath{\Sigma\Pi}}
\newcommand{\dthree}{\ensuremath{\Sigma\Pi\Sigma}}

\newcommand{\VP}{\ensuremath{\text{VP}}}
\newcommand{\VNC}{\ensuremath{\text{VNC}}}
\newcommand{\VPe}{\ensuremath{\text{VP}_{\text{e}}}}

\newcommand{\VNP}{\ensuremath{\text{VNP}}}

\newcommand{\ROP}{\text{ROP}}
\newcommand{\ROF}{\text{ROF}}

\newcommandx{\orbitofgroup}[2]{\ensuremath{#1^{#2}}}

\newcommand{\affineorbitof}[2][]{\orbitofgroup{#2}{\GLplus[#1]}}
\newcommand{\roanf}[1][]{\affineorbitof[#1]{\croanf[]}}

\newcommand{\anf}[2]{\ensuremath{\croanf[#1]^{\GLaff(#2)}}}
\newcommand{\anfn}[2]{\ensuremath{\croanf[#1]^{\gla{n}{#2}}}}
\newcommand{\sdminv}[1][]{\affineorbitof[#1]{\ensuremath{\mathcal T}}}
\newcommand{\taff}[1]{\ensuremath{\mathcal T}^{#1}}
\newcommandx{\sinv}[2][1=,2=]{\affineorbitof[#1]{\dtwo_{#2}}}

\newcommand{\spinv}[1]{\ensuremath{\dtwo^{\gla{}{#1}}}}

\newcommand{\ropinvgen}[1][t]{\ensuremath{#1}-independent polynomial map}
\newcommand{\ropinvgentext}[1][t]{\ensuremath{#1}\text{\normalfont{-independent polynomial map}}}

\newtheorem{theorem}{Theorem}[section]
\newtheorem{lemma}[theorem]{Lemma}

\newtheorem{claim}[theorem]{Claim}

\newtheorem{corollary}[theorem]{Corollary}
\newtheorem{definition}[theorem]{Definition}
\newtheorem{fact}[theorem]{Fact}

\newtheorem{observation}[theorem]{Observation}

\newtheorem{remark}[theorem]{Remark}

\crefname{claim}{Claim}{Claims}
\crefname{corollary}{Corollary}{Corollaries}
\crefname{definition}{Definition}{Definitions}
\crefname{example}{Example}{Examples}
\crefname{fact}{Fact}{Facts}
\crefname{lemma}{Lemma}{Lemmas}
\crefname{observation}{Observation}{Observations}
\crefname{statement}{Statement}{Statements}
\crefname{theorem}{Theorem}{Theorems}
\newtheorem{problem}[theorem]{Problem}

\makeatletter
\def\algbackskip{\hskip-\ALG@thistlm}
\makeatother

\allowdisplaybreaks
\date{}

\title{Hitting Sets and Reconstruction for Dense Orbits in \VPe{} and $\dthree$ Circuits}
\author{Dori Medini\thanks{Department of Computer Science, Tel Aviv University, Tel Aviv, Israel, E-mail: \texttt{dorimedini@gmail.com, shpilka@tauex.tau.ac.il}. The research leading to these results has received funding from the  Israel Science Foundation (grant number 514/20) and from the Len Blavatnik and the Blavatnik Family foundation. 
	}   \and  Amir Shpilka\footnotemark[1]}
\begin{document}
	\maketitle

\begin{abstract}

In this paper we study  polynomials in $\VPe$ (polynomial-sized formulas) and in $\dthree$ (polynomial-size depth-$3$ circuits) whose orbits, under the action of the affine group $\gla{n}{\f}$,\footnote{The action of $(A,\vb)\in\GLplus[n][\f]$ on a polynomial $f\in\f[\xn]$ is defined as $(A,\vb)\circ f=f(A^T\xn+\vb)$.} are \emph{dense} in their ambient class. We construct hitting sets and interpolating sets for these orbits as well as give reconstruction algorithms. Specifically, we obtain the following results:

\begin{enumerate}
	\item For $\cont_n\bra{\ell_1(\xn),\ldots,\ell_n(\xn)}\triangleq \trace \left(
	\begin{pmatrix}
		\ell_1(\xn) & 1 \\
		1 & 0
	\end{pmatrix}
	\cdot 
	\ldots \cdot
	\begin{pmatrix}
		\ell_n(\xn) & 1 \\
		1 & 0
	\end{pmatrix} 
	\right)$, where the $\ell_i$s are linearly independent linear functions, we construct a polynomial-sized interpolating set, and give a polynomial-time reconstruction algorithm. By a result of Bringmann, Ikenmeyer and Zuiddam, the set of all such polynomials is dense in $\VPe$ \cite{bringmann2018algebraic}, thus our construction  gives the first polynomial-size interpolating set for a dense subclass of $\VPe$.

	\item For polynomials of the form $\croanf\bra{\ell_1(\xn),\ldots,\ell_{4^\Delta}(\xn)}$, where $\croanf(\xn)$ is the canonical read-once formula in \emph{alternating normal form}, of depth $2\Delta$, and  the $\ell_i$s are linearly independent linear functions, we provide a quasipolynomial-size interpolating set. We also observe that the reconstruction algorithm of \cite{gupta2014random} works for \emph{all} polynomials in this class. This class is also dense in $\VPe$. 
	
	\item Similarly, we give a quasipolynomial-sized hitting set for read-once formulas (not necessarily in alternating normal form) composed with a set of linearly independent linear functions. This gives another dense class in $\VPe$.
	
	\item We give a quasipolynomial-sized hitting set for polynomials of the form $f\bra{\ell_1(\xn),\ldots,\ell_{m}(\xn)}$, where $f$ is an $m$-variate $s$-sparse polynomial and  the $\ell_i$s are linearly independent linear functions in $n\geq m$ variables.  This class is  dense in $\dthree$.

	\item For polynomials of the form $\sum_{i=1}^{s}\prod_{j=1}^{d}\ell_{i,j}(\xn)$, where the $\ell_{i,j}$s are linearly independent linear functions, we construct a polynomial-sized interpolating set. We also observe that the reconstruction algorithm of \cite{DBLP:journals/cc/KayalNS19} works for \emph{every} polynomial in the class. This class is  dense in $\dthree$. 
		
\end{enumerate}

As $\VP=\VNC^2$, our results for $\VPe$ translate immediately to $\VP$ with a quasipolynomial blow up in parameters. 
 
If any of our hitting or interpolating sets could be made \emph{robust} then this would immediately yield a hitting set for the superclass in which the relevant class is dense, and as a consequence also a lower bound for the superclass. Unfortunately, we also prove that the kind of constructions that we have found (which are defined in terms of $k$-independent polynomial maps) do not necessarily yield robust hitting sets. 

\end{abstract}

\thispagestyle{empty}
\newpage

\tableofcontents
\thispagestyle{empty}
\newpage

\addtocontents{toc}{\protect\thispagestyle{empty}}

\pagenumbering{arabic}

\section{Introduction}

Proving lower bounds on the size of algebraic circuits (also called arithmetic circuits), is an outstanding open problem in algebraic complexity. In spite of much effort, only a handful of lower bounds are known  (a detailed account of most known lower bounds can be found in the excellent survey of Saptharishi \cite{ramprasad}). 
One common theme of most known lower bounds is that they are proved using \emph{algebraic arguments}. That is, a proof of a lower bound for a class of circuits $\cC$, usually has the following structure: one comes up with a set of (nonzero) polynomials $F_1,\ldots,F_m$, in $N={n+d \choose d}$ many variables, such that the coefficient vector of every $n$-variate, degree-$d$ polynomial that can be computed in $\cC$, is a common zero of all the $F_i$s (such $F_i$s are called \emph{separating polynomials}). Then, one exhibits a polynomial $f$ whose coefficient vector is not a common zero, thus proving  $f\not\in \cC$. As an example one can immediately see that the well known partial derivative technique, and its predecessor, shifted partial derivative technique, are algebraic. Grochow  \cite{grochow2015unifying} demonstrated this for most of the known lower bound proofs.
As the set of common zeros of a set of polynomials is closed,\footnote{It is closed in the Zariski topology. Over $\R$ or $\C$ this is the same as being closed in the Euclidean topology.} this immediately implies that if we prove that $f\not\in \cC$ using an algebraic argument, then the same argument also implies that $f\not \in \overline\cC$, the closure of $\cC$. Recall that, in characteristic zero, the closure of a class $\cC$ is the set of all polynomials that are limit points of sequences of polynomials from $\cC$, where convergence is coefficient-wise (see \cref{def:approx} for a general definition over arbitrary characteristic). As most known  techniques are algebraic, we see that for proving a lower bound for a class $\cC$ one actually has to consider the larger, and less structured class, $\overline\cC$.



Geometric Complexity Theory (GCT for short), which was initiated by Mulmuley and Sohoni \cite{DBLP:journals/siamcomp/MulmuleyS01,DBLP:journals/siamcomp/MulmuleyS08}, approaches the lower bound question from a different angle. GCT also looks for an algebraic lower bound proof, but rather than exhibiting an algebraic argument, it aims to prove the existence of a separating polynomial. 
Specifically, GCT attempts to prove Valiant's hypothesis, that \VP$\neq$\VNP, over $\C$, via \emph{representation theory}. Valiant's hypothesis is, more or less, equivalent to showing that the permanent of a symbolic $n\times n$ matrix is not a \emph{projection} of the symbolic $m\times m$ determinant for any $m=m(n)$  polynomial in $n$.\footnote{A super-quasipolynomial lower bound would imply that  \VP$\neq$\VNP{} whereas a super-polynomial lower bound would imply that permanent does not have polynomial-size algebraic formulas or algebraic branching programs.} 
Recall that a projection of a polynomial is a restriction of the polynomial to an affine subspace of its inputs.
Observe that a restriction of an $n$-variate polynomial $f(\xn)$ to a subspace of its inputs, is equivalent to considering the polynomial $f(A\xn+\vb)$, where  $A$ is an $n\times n$ matrix and $\vb\in \C^n$. As any matrix is a limit point of a sequence of invertible matrices, an algebraic proof that the  permanent is not a projection of the $m\times m$ determinant, over $\C$, is equivalent to an algebraic proof showing that the permanent is not in the closure of the set of polynomials $\{\det(AX+\vb) \mid A\in \GLm{m}{\C},\; \vb\in\C^m\}$, where $\GLm{m}{\C}$ is the group of invertible $m\times m$ matrices  (this is true for every field of characteristic $\neq 2$). The set  $\{\det(AX+\vb) \mid A\in \GLm{m}{\C},\; \vb\in\C^m\}$ is called the \emph{orbit} of the determinant under the action of the affine group (we denote the affine group over $\C^m$ with $\gla{m}{\C}$). GCT considers the linear space of polynomials that vanish on every coefficient vector in the orbit of the determinant, and similarly the linear space of polynomials that vanish on every coefficient vector in the orbit of the permanent. There is a natural action of \GLmplus{} on those linear spaces, thus defining two representations of \GLmplus. GCT wishes to find a separating polynomial by showing that some irreducible representation of \GLmplus{} has strictly larger multiplicity when considering the representation corresponding to the determinant. This approach bypasses the barrier given in \cite{DBLP:journals/toc/ForbesSV18,DBLP:journals/corr/Grochow0SS17} as it does not exhibit any efficiently computable separating polynomial but rather just proves the existence of one.  However, the representation theory questions arising in this program are quite difficult, even when considering the analog questions for restricted classes. For an introduction to GCT see the lecture notes of Bl\"aser and Ikenmeyer \cite{BlaserIkenmeyer}.

Another possible approach for proving lower bounds against a class of polynomials $\cC$, is via the construction of a \emph{hitting set} for $\cC$. Recall that a hitting set $\cal H$ for a class $\cC$ is a set of points such that for any nonzero polynomial $f$, that can be computed by a circuit from $\cC$, there is $\vv\in{\cal H}$ such that $f(\vv)\neq 0$. In  \cite{heintz1980testing} Heintz and Schnorr  observed that if we have such a hitting set $\cal H$ then any nonzero polynomial $g$ that vanishes on $\cal H$ cannot be computed in $\cC$. It is also not hard to see that this way of obtaining lower bounds also bypasses the natural proof barrier of  \cite{DBLP:journals/toc/ForbesSV18,DBLP:journals/corr/Grochow0SS17}. The problem is that in most cases we obtained a hitting set for a class only after proving a lower bound for it.  

In \cite{forbes2018pspace} Forbes and Shpilka defined the notion of a \emph{robust} hitting set for a circuit class $\cC$. 
Over fields of characteristic zero, a hitting set $\cH$ for a class $\cC$ is $c$-robust if it also   satisfies that for every $f\in\cC$ there is $\vv\in\cH$ such that $|f(\vv)|\geq c \cdot \norm{f}$, where $\norm{\cdot}$ is some fixed norm on $\C[\xn]$ (see \cref{def:robust} for a definition over arbitrary fields). It is not hard to see that if $\cH$ is a robust hitting set for a class $\cC$ then it also hits the closure of $\cC$. 

In this work we 
focus on depth-$3$ algebraic circuits, known as \dthree{}, and on \VPe, the class of algebraic formulas, two classes for which we lack strong lower bounds, and in particular we do not have hitting sets for them. For \dthree{} circuits the best lower bound is the near cubic lower bound of Kayal, Saha and Tavenas \cite{kayal2016almost}, and for \VPe{} the best lower bound is the 
quadratic lower bound of Kalarkoti \cite{kalorkoti1985lower}. Recall that by the result of Valiant et al. \cite{DBLP:journals/siamcomp/ValiantSBR83}, a super-quasipolynomial lower bound against \VPe{} implies a super-polynomial lower bound against \VP. Similarly, a hitting set for $\VPe$ implies a hitting set for $\VP$. We also note that by a result of Gupta et al. \cite{gupta2013arithmetic}, a strong enough lower bound or a hitting set for $\dthree$ imply both a lower bound for general circuits and a hitting set for them. This result also implies that a polynomial-time reconstruction algorithm for $\dthree$ circuits would give rise to a sub-exponential time \emph{reconstruction algorithm} for general circuits. Recall that a reconstruction algorithm for a class $\cC$ is an algorithm that, given black-box access to a circuit from $\cC$, outputs a circuit in $\cC$ that computes the same polynomial.

Instead of viewing robust hitting sets as a way to obtain  hitting sets for the closure of circuit classes, we suggest to find subclasses of interesting classes, $\tilde{\cC}\subset \cC$, such that $\cC$ is contained in the closure of $\tilde{\cC}$, and aim to construct a robust hitting set for the subclass $\tilde{\cC}$. 
This offers a new approach for constructing hitting sets for known classes and for obtaining lower bounds. 
Specifically, we consider subclasses of $\dthree$ and $\VPe$ that are dense in their superclasses. Each of these subclasses is the orbit of some simple polynomial under the group of invertible affine transformations. 

For $\VPe$, we first consider a subclass that was defined by Bringmann, Ikenmeyer and Zuiddam \cite{bringmann2018algebraic}--the orbit of the so called \emph{continuant} polynomial (see \cref{def:cont}). We give a polynomial-sized interpolating set\footnote{
	Recall that an interpolating set for a class  $\cC$ of polynomials in $n$ variables, over a field $\f$, is a set of points ${\cal H}\subset\f^n$ such that for every  $f\in\cC$, the list of values $f(\cal{H})$ uniquely determines $f$. See \cref{defInterSet}.} for this subclass as well as a  polynomial-time deterministic reconstruction algorithm that uses as oracle a  \emph{root-finding algorithm}.\footnote{A root-finding algorithm, over a field $\f$, when given black-box access to a univariate polynomial, outputs a root of that polynomial in $\f$, if such a root exists.} In particular, this implies a polynomial-time randomized reconstruction algorithm, and, in some cases, a polynomial-time deterministic algorithm.

In addition, we exhibit two other subclasses that are dense in $\VPe$. The first class is defined as the orbit of read-once formulas (ROF for short, see \cref{def:ROF}) and the second as the orbit of read-once formulas in \emph{alternating normal form} (ROANF for short, see \cref{def:ANF}). We obtain hitting sets for both classes and an interpolating set for the second. We also observe that the reconstruction algorithm of \cite{gupta2014random} works for the polynomials in the orbit of ROANFs.
Although the results that we obtain for the subclass defined by the continuant polynomial are stronger, we think that every such dense subclass can shed more light on $\VPe$ and may eventually be used in order to obtain new lower bounds.

For $\dthree$ we consider two subclasses. One is based on orbits of \emph{sparse} polynomials (polynomials having polynomially many monomials) and the other on orbits of \emph{diagonal} tensors (see \cref{defT}). We give a hitting set for the first, an interpolation set for the second, and we also observe that a slight modification of the randomized reconstruction algorithm of \cite{DBLP:journals/cc/KayalNS19} applies for the second class. 

In particular, our results give the first dense subclasses inside $\VPe$ and $\dthree$ for which a polynomial-size interpolating set is known as well as a polynomial-time reconstruction algorithm. By  \cite{DBLP:journals/siamcomp/ValiantSBR83} our result immediately translate to $\VP$, giving a dense subclass of for which a quasipolynomial-sized interpolating set is known as well as a quasipolynomial-time reconstruction algorithm.


 If we could transform the interpolating sets that we have found to \emph{robust hitting sets}  for the orbits, then this will immediately give hitting sets for the closure of the orbits, i.e. for \dthree{} and \VPe, which, by \cite{heintz1980testing} gives a lower bound for the class. 
Thus, our work raises an intriguing problem: 
\begin{problem}\label{prob:robust}
	Given an interpolating set for a class $\cC$ construct a robust hitting set for $\cC$.
\end{problem}
 We stress that by our results, solving this problem would lead to hitting sets, and lower bounds, for  \VPe{} and \VP{}. 

Another advantage for having small interpolating sets for dense subclasses is the following: One approach for searching for separating polynomials for a class, is by considering the map from circuits in the class to the coefficient vectors of the polynomials that they compute. That is, once we fix a computation graph, an assignment to the constants appearing in the circuit determines the output polynomial. Each coefficient is a polynomial in those constants, and as there are ``few'' constants (polynomially many for polynomially sized circuits), and there are exponentially many coefficients, there should be many polynomials vanishing on the closure of the image of this map. If we could get a good understanding of this map then perhaps we could use it to construct a polynomial that vanishes on all such coefficient vectors. This polynomial will vanish on all coefficient vectors of the superclass in which the subclass is dense.
A different approach is to find a coefficient vector that is not in the closure of the image of this map (this is the approach of Raz in \cite{DBLP:journals/toc/Raz10}).
Now, assume that $\cH$ is an interpolating set for a dense subclass $\tilde{\cC}\subset \cC$. We know that the map $f\to \restr{f}{\cH}$ is one-to-one on $\tilde{\cC}$. Thus, the list of values  $\restr{f}{\cH}$ can be viewed as an efficient encoding that is given in terms of values of the computed polynomial. This provides a different encoding of a circuit -- instead of the constants in it, use the evaluations on $\cH$.
Thus, by studying the closure of this map (i.e. the closure of the set of points on $\f^{|\cH|}$ that can be obtained as evaluation vectors of polynomials in the subclass) we may be able to find a separating polynomial, or, as in Raz's approach, find an evaluation vector that is not obtained by any polynomial in the superclass. It is clear that one can also try this approach even if $\cal H$ is not an interpolating set, however, as interpolating sets ``preserve information'' of a dense set, we believe that such sets are better suited for this approach.

To conclude, focusing on dense subclasses and studying their properties could lead to better understanding of their superclasses and perhaps to breakthrough results in algebraic complexity.

To formally state our results we need some definitions that we give next.

\subsection{Basic definitions}


\subsubsection{Circuit classes}



\begin{definition}\label{defFormula}
	An algebraic formula (also called arithmetic formula) over a field $\f$, is a rooted tree whose leaves are labeled with either variable or scalars from $\f$, and whose root and internal nodes (called gates) are labeled with either ``$+$''  (addition) or ``$\times$'' (multiplication). An algebraic formula computes a polynomial in the natural way. Each leaf computes the polynomial that labels it, and each gate computes either the sum or product of its children, depending on its label. The output of the formula is the polynomial computed at its root. The size of a formula is the number of wires in it. The depth of a formula is the length of the longest simple leaf-root path in it. The formula size of a polynomial $f$ is defined as the smallest size of a formula that outputs $f$.
\end{definition}

A sequence $m(n)$ of natural numbers is called polynomially bounded if there exists a univariate polynomial $q$ such that $m(n) \leq q(n)$ for all $n$.

The complexity class \VPe{} is defined as the set of all families of polynomials $(f_n)_n$, with $f_n \in\f[x_1,\ldots,x_n]$, whose formula size is polynomially bounded.

\begin{definition}\label{defDepthTow}
	An arithmetic circuit $\Phi$ is a $\Sigma^{[s]}\Pi^{[d]}$ circuit if it is a layered graph of depth-$2$, has a top gate labeled $+$ with fan-in $\leq s$ and its second layer is comprised entirely of $\times $ gates with fan-in $\leq d$. In other words, $\Sigma^{[s]}\Pi^{[d]}$ compute polynomials of degree $d$ with at most $s$ monomials.
\end{definition}

\begin{definition}\label{defDepthThree}
	An arithmetic circuit $\Phi$ in $n$ variables is a $\Sigma^{[s]}\Pi^{[d]}\Sigma$ circuit if it is a layered graph of depth-$3$, has a top gate labeled $+$ with fan-in $\leq s$, its second layer is comprised entirely of $\times $ gates with fan-in $\leq d$, and its bottom layer is comprised of linear functions in $x_1,\ldots,x_n$. In other words, $\Sigma^{[s]}\Pi^{[d]}\Sigma$ circuit compute polynomials of the form
	$$ f(\xn)=\sum_{i=1}^s\prod_{j=1}^d(\alpha_{i,j,0}+\sum_{k=1}^n\alpha_{i,j,k}x_k) \;. $$
\end{definition}

Given a family of circuits $\cC$, we will sometime denote it as $\cC(\f)$ to stress that we allow coefficients to come from the field $\f$. Observe that the definitions of  the classes above do not depend on the field and so we can define them over any field of our choice.

\subsubsection{Approximate complexity}

The following definition gives sense to the notion of approximation over arbitrary fields. In what follows we let $\varepsilon$ be a new formal variable.\footnote{Intuitively, one should think of $\varepsilon$ as an infinitesimal quantity.} For a field $\f$ we denote with $\f[\varepsilon]$ the ring of polynomial expressions in $\varepsilon$ over $\f$, and with $\f(\varepsilon)$ the fraction field of $\f[\varepsilon]$, i.e. the field of rational expressions in $\varepsilon$.


\begin{definition}\label{def:approx}
	Let $\cC(\f)$ be a circuit class over a field $\f$. The closure of $\cC$, denoted $\overline{\cC(\f)}$, is defined as follows: A family of functions $(f_n)_n$, where $f_n\in \f[x_1,\ldots,x_n]$, is in $\overline{\cC(\f)}$ if there is  a polynomially bounded function $m:\n\to\n$, and a family of functions $(g_{m(n)})_n \in \cC(\f(\varepsilon))$, with $g_{m(n)}\in \f[\varepsilon][x_1,\ldots,x_{m(n)}]$, such that for all $n\in \n$,
	\begin{equation}\label{eq:approx}
	g_{m(n)}(x_1,\ldots,x_{m(n)}) = f_n(x_1,\ldots,x_n) +  \varepsilon \cdot g_{n,0}(x_1,\ldots,x_{m(n)}) \;,		
	\end{equation}
	for some polynomial $g_{n,0}\in \f[\varepsilon][x_1,\ldots,x_{m(n)}]$.
	Whenever an equality as in \eqref{eq:approx} holds we say that 
	$$  g_{m(n)} =f_n  +O(\varepsilon) \quad \text{or} \quad f_n = g_{m(n)} +O(\varepsilon)\;.$$ 
	In that case we think of $g_{m(n)}$ as an ``approximation'' of $f_n$, and we say that the family $(g_{m(n)})_n$ approximates the family $(f_n)_n$.	
\end{definition}

Alder \cite{Alder84}  have shown that over $\C$ it holds that $(f_n) \in \overline{\cC(\C)}$, in the sense of \cref{def:approx}, if and only if it is in the closure of $\cC(\C)$ in the usual sense. That is, if for every $n$ there exists a sequence of polynomials $g_{n,k}\in \cC(\C)$ such that $\lim_{k\to \infty} g_{n,k}  = f_n$, where convergence is taken coefficient wise. This result holds over $\mathbb{R}$ as well, see \cite{DBLP:journals/tcs/LehmkuhlL89,DBLP:journals/focm/Burgisser04}.

Finally, we note that every matrix is approximable (in the sense of \cref{def:approx}) by a non-singular matrix (which is equivalent to being a limit of a sequence of non-singular matrices, in characteristic zero).

\begin{observation}\label{matLimits}
	For every $A\in\f^{n\times n}$ there exists a non-singular matrix $B\in \f(\varepsilon)^{n\times n}$ such that
	$A=B+O(\varepsilon)$. 
\end{observation}

\subsubsection{Hitting and interpolating sets}

\begin{definition}\label{def:hitting}
	A set of points $\mathcal H\subseteq\f^n$ is called a \emph{hitting set} for a circuit class $\mathcal C$ (we also say that $\mathcal H$ \emph{hits} $\mathcal C$) if for every circuit $\Phi\in\mathcal C$, computing a non-zero polynomial, there exists some $\va\in\mathcal H$ such that $\Phi(\va)\neq 0$.
\end{definition}

We next give the definition of a robust hitting set, a notion first defined in \cite{forbes2018pspace}. Here we extend the definition for arbitrary characteristic. We start by giving the definition of \cite{forbes2018pspace}, over characteristic zero (and focus on $\C$) and then the more general definition.

\begin{definition}[Following Definition 5.1 of \cite{forbes2018pspace}]\label{def:robust-0}
	Let $\norm{\cdot}$ be some norm on $\C[\xn]$. A hitting set $\mathcal H$ for a circuit class $\mathcal C\subseteq\C[\xn]$ is called \emph{robust} if there exists some constant $c>0$ such that, for every $0\neq f\in\mathcal C$,\footnote{We abuse notation and write $f\in \cC$ when $f$ is the output of some circuit from $\cC$.} there exists some $\va\in\mathcal H$ such that $|f(\va)|\geq c\cdot \norm{f}$.
\end{definition}

For arbitrary characteristic we use the same approach as in \cref{def:approx}.

\begin{definition}\label{def:robust}
	Let $\f$ be a field of arbitrary characteristic.
	A hitting set $\mathcal H\subset \f^n$ for a circuit class $\cC(\f)$ is called \emph{robust} if for every circuit $\Phi\in\cC(\f(\varepsilon))$ computing a polynomial $f(\xn)= h(\xn)+\varepsilon \cdot g(\xn)$, where $h(\xn)\in\f[\xn]$ and $g(\xn)\in\f[\varepsilon][\xn]$, there exists some $\va\in\mathcal H$ such that $f(\va) \not\in \varepsilon \cdot \f[\varepsilon]$.
\end{definition}

It is not hard to prove using the result of \cite{Alder84}  that for $\f=\C$, \cref{def:robust-0,def:robust} are equivalent. 
 
\begin{observation}\label{tm:robust-hits}
	If $\cH$ is a finite  robust hitting set for $\cC(\f)$, then $\cH$  hits $\overline{\cC(\f)}$ as well.
\end{observation}
\begin{proof}
	Consider $0\neq f \in \overline{\cC(\f)}$. By \cref{def:approx}  there is $g\in \cC(\f(\varepsilon))$, such that $f=g+O(\varepsilon)$. Clearly $g\neq 0$. Let  $\va\in \cH$ be such that $g(\va)\not \in \varepsilon\cdot\f[\varepsilon]$.  It follows that $f(\va)\not \in \varepsilon\cdot\f[\varepsilon]$. In particular, $f(\va)\neq 0$.
\end{proof}

We next define the notion of an interpolating set.

\begin{definition}\label{defInterSet}
	Let $\cC$ be a class of $n$-variate polynomials. A set $\mathcal H\subseteq\f^n$ is called an \emph{interpolating set} for $\cC$ if, for every $f\in\cC$, the evaluations of $f$ on $\mathcal H$ uniquely determine $f$.
\end{definition}

\begin{observation}\label{obs:inter}
	If $\mathcal H$ is a hitting set for $\cC(\f)+\cC(\f)\triangleq \{\alpha f+\beta g:f,g\in\cC,\alpha,\beta\in\f\}$, then $\mathcal H$ is an interpolating set for $\cC$.
\end{observation}

A common method for designing hitting and interpolating sets is via hitting set generators. 

\begin{definition}
	A polynomial mapping $\G:\f^k\to\f^n$ is called a \emph{hitting set generator} (or simply a generator) for a circuit class $\cC(\f)$ if for any non-zero $n$-variate polynomial $f\in\cC$, the $k$-variate polynomial $f\circ\G$ is non-zero. 
	
	Similarly, we call $\G:\f^k\to\f^n$  an \emph{interpolating set generator} for a circuit class $\cC(\f)$ if for any two different  $n$-variate polynomials $f_1,f_2\in\cC$, the $k$-variate polynomial $(f_1-f_2)\circ\G$ is non-zero.
\end{definition}

Generators immediately give rise to hitting sets.

\begin{observation}\label{obsHitSetGen}
	Let $\G:\f^k\to\f^n$ be a {generator} for $\cC(\f)$ such that the individual degree of each coordinate of $\G$ is at most $r$. Let $W\subset \f$ be any set of size $|W|=d\cdot r+1$. Let $\cH = \G\bra{W^k}$. Then $\cH$ hits every $n$-variate polynomial $f\in\cC$ of degree at most $d$.
\end{observation}

\begin{proof}
	As $\G$ is a generator, the $k$-variate polynomial $f\circ \G$ is nonzero. As its individual degrees are bounded by $d\cdot r$ it follows that at least one of the values in $\bra{f\circ \G}\bra{W^k}=f\bra{\cH}$ is not zero.
\end{proof}

\subsubsection{$k$-independent maps}

Our constructions  rely on polynomial mappings $\G_k$, parameterized by some integer $k\leq n$, with the property that the image of $f\circ \G_k$ contains all projections of $f$ to $k$ variables. We call such a map a $k$-independent map.

\begin{definition}\label{def:indPolyMap}
	We call a polynomial mapping $\G(y_1,\ldots,y_t,z_1):\f^{t+1}\to\f^n$ a \ropinvgentext[1]{}  if for every index $i\in[n]$ there exists an assignment $\va_i\in\f^t$ to $y_1,\ldots,y_t$ such that the $i$th coordinate of $\G(\va_i,z_1)$ is $z_1$, and the rest of the coordinates are $0$. For $k>1$, a polynomial mapping $\G(y_1,\ldots,y_{tk},z_1,\ldots,z_k):\f^{k(t+1)}\to\f^n$ is called a \ropinvgentext[k]{} (or a $k$-independent map) if $\G$ is a sum of $k$ variable-disjoint $1$-independent polynomial maps. We denote \ropinvgen[k]{}s as $\G(\bm y,\bm z)$ when $k,t$ are implicit. The $\bm y$ variables are called \emph{control variables}.
	
	A \ropinvgen[k]{} $\G$ is called \emph{uniform} if all $n$ coordinates of $\G$ are homogeneous polynomials of the same degree.
\end{definition}

\subsubsection{The linear and  affine groups and their actions}

Given a matrix $A\in \f^{n\times n}$ and a tuple of variables  $\xn= (x_1,\ldots,x_n)$, we denote  $$A\xn=\bra{\sum_{i=1}^n A_{1,i}x_i,\sum_{i=1}^n A_{2,i}x_i,\ldots,\sum_{i=1}^n A_{n,i}x_i}\;.$$
Let $n\geq m \in \n$. For an $m$-variate polynomial $f(x_1,\ldots,x_m)\in \f[x_1,\ldots,x_m]$, a matrix $A=(A_{i,j})_{i,j=1}^n\in\f^{n\times n}$ and a vector $\vb=(b_1,\ldots,b_n)\in\f^n$, we define the $n$-variate polynomial $f\bra{A\xn+\vb}$ to be
\begin{equation}\label{eq:compose}
	f\bra{A\xn+\vb}\triangleq f\left(\sum_{i=1}^n A_{1,i}x_i+b_1,\sum_{i=1}^n A_{2,i}x_i+b_2,\ldots,\sum_{i=1}^n A_{m,i}x_i+b_m\right)\;.	
\end{equation}
Note that we ignored the last $n-m$ coordinates of $A\xn+\vb$.

We denote with $\GLm{n}{\f}$ the group of invertible $n\times n$ matrices over $\f$, 
and with $\GLplus[n][\f]$ the group of invertible affine transformation, i.e. all the maps $\xn\to A\xn +\vb$, where $A\in \GLm{n}{\f}$ and $\vb\in\f^n$.

For an $m$-variate polynomial $f$ over $\f$, and $n\geq m$ we denote with  $f^{\GLplus[n][\f]}$  the orbit of $f$ under the natural action of $\GLplus[n][\f]$:\footnote{To be precise, the action is $\bra{(A,\vb)\circ f}(\xn)=f(A^T\xn +\vb)$. This is required in order to make the action a homomorphism, however, for the groups that we consider it does not change the orbit.}
$$f^{\GLplus[n][\f]}\triangleq \left\{ f(A\xn+\vb) \mid A\in \GLm{n}{\f},\; \vb\in\f^n\right\}\;.$$ 
We similarly define $f^{\GLm{n}{\f}}$.
More generally, for a class of $m$-variate polynomials ${\mathcal C}(\f)$, we denote the \emph{orbit} of $\mathcal C$ under $\GLplus[n][\f]$ by 
 $${\mathcal C}^{\GLplus[n][\f]}\triangleq\cbra{f(A\xn+\vb)\mid f\in{\mathcal C},\; A\in \GLm{n}{\f},\; \vb\in\f^n}\;.$$ We similarly define ${\mathcal C}^{\GLm{n}{\f}}$.
When we want to speak about orbits of families of polynomials from  $\cC(\f)$, with arbitrary number of variables, we use the notation $\cC^{\GLm{}{\f}}$ or ${\mathcal C}^{\GLplus[][\f]}$.

\subsection{Our results}\label{secResults}


We first give our results for the class $\VPe$ and then for the  class of depth-$3$ circuits, for which it may be easier to obtain a robust hitting set, or prove super-polynomial lower bounds.

\subsubsection{The continuant polynomial}

Bringmann, Ikenmeyer and Zuiddam \cite{bringmann2018algebraic} defined the following polynomial (in Remark 3.14 of their paper), which they called the continuant polynomial:

\begin{definition}\label{def:cont}
	The continuant polynomial on $n$ variables, $\cont_{n}(x_1,\ldots,x_n)$, is defined as the \emph{trace}
	of the following matrix product:
	\begin{equation}
		\cont_n(x_1,\ldots,x_n)\triangleq \trace \left(
		\begin{pmatrix}
			x_1 & 1 \\
			1 & 0
		\end{pmatrix}
	\cdot 
	\begin{pmatrix}
		x_2 & 1 \\
		1 & 0
	\end{pmatrix}
	\cdot \ldots \cdot
	\begin{pmatrix}
		x_n & 1 \\
		1 & 0
	\end{pmatrix} 
	\right) \;.
	\end{equation}
	We denote with $\ccont{\f}$ the class of families of
	polynomials $(f_n)_n$ such that $f_n \in \f[x_1,\ldots,x_n]$ and for some $m\leq n$, $f_n \in \cont_m^{\gla{n}{\f}}$. 
\end{definition}

 A result of Allender and Wang implies that the polynomial $x_1\cdot y_1 + \cdots + x_8 \cdot y_8$ is not in $\ccont{\f}$ \cite{DBLP:journals/cc/AllenderW16}. Thus, as a computational class it is very weak.
However,
Theorem 3.12 of \cite{bringmann2018algebraic} states that for every field $\f$ of characteristic different than $2$, it holds that
\begin{equation}
	\overline{\ccont{\f}} = \overline{\VPe}\;.
\end{equation}

We give a polynomial-size interpolating set for the class $\ccont{\f}$ as well as a  polynomial-time reconstruction algorithm for it. We first state a simple result that gives a hitting set for the class. 

\begin{restatable}{theorem}{thmhitcont}\label{thmhitcont}
	Let $f(x_1,\ldots,x_n) \in \cont_m^{\gla{n}{\f}}$, for $m\leq n$, and arbitrary $\f$. Then, for any uniform \ropinvgen[1]{} $\G$ over $\f$,  $f\circ \G\neq 0$. 
\end{restatable}

As immediate corollary we get a hitting set for the class.

\begin{corollary}\label{cor:hit-cont}
	For every field $\f$, there is an explicit hitting set $\cH\subset\f^n$, of size $|\cH|=O\bra{n^6}$, that hits every $0\neq f\in \cont_m^{\gla{n}{\f}}$. If $|\f|<n^2$ then $\cH$ is defined over  a polynomial-sized extension field of $\f$, $\kk$ such that $|\kk|\geq n^2$.
\end{corollary}

\begin{theorem}\label{thm:cont-int}
	For every field $\f$, there is an explicit interpolating set $\cH\subset\f^n$, of size $|\cH|=O\bra{n^{10}}$, for $\bigcup_{m=1}^{n}\cont_m^{\gla{n}{\f}}$. If $|\f|<n^2$ then $\cH$ is defined over  a polynomial-sized extension field of $\f$, $\kk$ such that $|\kk|\geq n^2$.
\end{theorem}

\begin{theorem}\label{thm:cont-recon}
	There is a deterministic algorithm that given $\f$, an integer $ n$, oracle access to a root-finding algorithm over $\f$, and black-box access to a polynomial  $f(x_1,\ldots,x_n) \in \cont_m^{\gla{n}{\f}}$ (for any $m\leq n$), runs in polynomial-time and outputs linear functions $\left(\ell_1(x_1,\ldots,x_n),\ldots,\ell_m(x_1,\ldots,x_n)\right)$ such that $$f(x_1,\ldots,x_n) = \cont_m\left(\ell_1(\xn),\ldots,\ell_m(\xn)\right)\;.$$
	 If $|\f|<n^{3}$ then the algorithm will make queries from a polynomial-sized extension field of $\f$, $\kk$, such that $|\kk|\geq n^{3}$, and it also requires oracle access to a root-finding algorithm over $\kk$.
\end{theorem}

\subsubsection{Orbits of read-once formulas}
\label{sec:res-rof}

Roughly, a read-once formula (ROF) is a formula in which every variable labels at most one leaf. However, following \cite{shpilka2015read,v010a018} we also allow gates of the formula to pass on their output wire a linear function of their polynomial (see \cref{def:ROF}).  We denote with $\ROF^{\GLm{}{\f}}$ the class of families of polynomials $(f_n)_n$, such that for every $n$ there exists a ROF $\Phi$, on $m\leq n$ variables,  such that  $f_n(x_1,\ldots,x_n) \in \Phi^{\GL[n][\f]}$.

A ROF is in \emph{alternating normal form} (ROANF) if it is a full binary tree of depth $2\Delta$ with alternating layers of addition and multiplication gates. In particular, it is a ROF on $4^\Delta$ many variables (see \cref{def:ANF}). 

We denote with $\croanf$ the canonical ROANF of depth $2\Delta$ in which the leaves are labeled with the variables $x_1,\ldots,x_{4^\Delta}$ according to their order (see \cref{defCroanf}). We denote with $\croanf[]^{\GLaff[\f]}$ the class of families of polynomials $(f_n)_n$, such that for every $n$ there exists $\Delta$ such that $4^\Delta\leq n$ and  $f_n(x_1,\ldots,x_n) \in \croanf^{\gla{n}{\f}}$.

We first make the following simple observation.
\begin{restatable}{theorem}{ROANFisDense}
	\label{ROANFisDense}
For every field $\f$, it holds that 	
\begin{equation}
	\anf{}{\f}\subsetneq\ROF^{\gl{}{\f}}\subsetneq \VPe(\f)\;.
\end{equation}
However, when taking closures we get
\begin{equation}
	\overline{\croanf[]^{\GLaff(\f)}}=\overline{\ROF^{\GLm{}{\f}}}=\overline{\VPe(\f)}\;. 	
\end{equation}
\end{restatable}


Our main results for ROFs and ROANFs are a construction of a hitting set for the orbit of ROFs, and an interpolating set for the orbit of ROANFs. Both constructions are obtained using  independent polynomial maps (\cref{def:indPolyMap}).

\begin{restatable}{theorem}{PITROPINV}
	\label{thm:PITROPINV}
	Let $0\neq f\in \ROF^{\gla{n}{\F}}$ where the underlying \ROF{} depends on $2^t$ variables, for $2^t\leq n$. Then, for any \ropinvgen[(t+1)]{} $\G$, over $\f$, $f\circ \G\neq 0$.
\end{restatable}

\begin{corollary}\label{cor:PITROPINV}
	For every field $\f$, there is a hitting set $\cH\subset \f^n$, of size $|\cH|=n^{O(\log n)}$, that hits every $0\neq f\in \ROF^{\gla{n}{\f}}$. If $|\f|<n^2$ then $\cH$ is defined over  a polynomial-sized extension field of $\f$, $\kk$ such that $|\kk|\geq n^2$.
\end{corollary}

Since a hitting set for all polynomials of the form $g-h$ where $g,h\in\cC$ is the same as an interpolating set for $\cC$, the following theorem gives an interpolating set for the orbit of ROANFs.

\begin{restatable}{theorem}{pitSumOfRoanfThm}\label{thm:pitSumOfRoanfThm}
	Let $f_1=\croanf[\Delta_1](A_1\xn+\vb_1),f_2=\croanf[\Delta_2](A_2\xn+\vb_2)\in\anfn{}{\f}$ and $f=f_1-f_2$. Set $k\triangleq 2\max\{\Delta_1,\Delta_2\}+7$ and let $\G$ be any uniform \ropinvgen[k]{}, over $\f$. If $f\neq 0$ then $f\circ\G\neq 0$.
\end{restatable}

\begin{corollary}\label{cor:pitSumOfRoanf}
	For any field $\f$, the class $\croanf[\Delta]^{\gla{n}{\f}}$, for $4^\Delta\leq n$, admits an interpolating set $\cH\subset \f^n$, of size $|\cH|=n^{O(\Delta)}$. If $|\f|<n^2$ then $\cH$ is defined over  a polynomial-sized extension field of $\f$, $\kk$, such that $|\kk|\geq n^2$.
\end{corollary}

Finally, we observe that the randomized algorithm of Gupta, Kayal And Qiao \cite{gupta2014random}, for reconstructing \emph{random algebraic formula} (for a natural definition of a random formula), yields a randomized reconstruction algorithm for $\roanf$. Naturally, the reconstruction is up to the symmetry group of ROANFs.

\begin{restatable}[A special case of Theorem 1.1 of \cite{gupta2014random}]{theorem}{reconstructROANF}
	\label{thm:reconstructROANF}
	Let $T$ be a finite subset of $\C$. Let $n,\Delta\geq 1$ be integers such that $s\triangleq 4^\Delta\leq n$. Given black-box access to the output $f$ of a circuit $\Phi\in{\croanf[]}^{\gla{n}{\C}}$, with probability at least $1-\frac{n^2s^{O(1)}}{|T|}$ (on internal randomness), Algorithm 6.9 of \cite{gupta2014random} successfully computes a tuple of $s$ linearly independent linear functions $L=(\ell_1,\ldots,\ell_s)\in(\C[\xn])^s$ such that $f=\croanf(\ell_1,\ldots,\ell_s)$, and the $\ell_i$s are identical to the labels of the leaves of $\Phi$ up to $\TS$-equivalence (see \cref{defTS}). Moreover, the running time of the algorithm is $\text{poly}(n,s,\log(|T|))$.
\end{restatable}

\begin{remark}\label{rem:gupta2014-char}
	Theorem 1.1 of \cite{gupta2014random} is stated only for characteristic zero fields. However, in Remark 6.10 they explain how to make the algorithm work over any characteristic, for a large enough field. Thus, \cref{thm:reconstructROANF} also holds over large enough fields in arbitrary characteristic.
\end{remark}

\begin{remark}\label{rem:ze-anf}
		As a direct implication of \cref{thm:pitSumOfRoanfThm}, the reconstruction algorithm of \cref{thm:reconstructROANF} can be converted into a zero-error algorithm, with expected quasipolynomial running time: Given black-box access to some $f_1\in\anf{}{\f}$, we define $f_2$ to be the output of the algorithm of \cref{thm:reconstructROANF} on input $f_1$, and then verify $f_1= f_2$ using \cref{cor:pitSumOfRoanf}.	
\end{remark}

\subsubsection{Dense subclasses of \dthree}

We start by defining the canonical diagonal tensor of degree $d$ and rank $s$, $\sdmpol[s][d]\in\f[x_{1,1},\ldots,x_{s,d}]$, and the resulting class of polynomials $\taff{\gla{}{\f}}$.

\begin{definition}{\label{defT}}
Let $\sdmpol[s][d]\triangleq \sum_{i=1}^{s}\prod_{j=1}^{d}x_{i,j}$. I.e., it is 
a sum of $s$ variable-disjoint monomials.  For $n\geq s\cdot d$, we denote with $\sdmpol[s][d]^{\gla{n}{\f}}$ the orbit of $\sdmpol[s][d]$ over $\f$, under the action of the affine group. Finally, we denote with $\taff{\gla{}{\f}}$ the class of families of polynomials $(f_n)_n$, such that for every $n$ there exist $s$ and $d$ such that $n\geq s\cdot d$ and   $f_n(x_1,\ldots,x_n) \in \sdmpol[s][d]^{\gla{n}{\f}}$.
\end{definition}

Clearly, $\sdmpol[s][d]^{\gla{n}{\f}}\subset \Sigma^{[s]}\Pi^{[d]}\Sigma$.
We next define the class consisting of orbits of sparse polynomials. 

\begin{definition}
Let $\spinv{\f}$ denote the class of families of polynomials that are computed by orbits of depth-$2$ circuits, of polynomially bounded size, over $\f$. I.e., it is all families $(f_n)_n$, of polynomially bounded degree, such that for some polynomially bounded $m(n)$, there exist $\Sigma^{m(n)}\Pi^{\deg(f_n)}$ circuits $\Phi_m$, in $k\leq n$, many variables, such that $f_n\in \Phi_m^{\gla{n}{\f}}$. 	
\end{definition}

As before we first give the basic observation connecting all three classes.

\begin{theorem}\label{thm:denseInDepthThree}
For every field $\f$ it holds that
\begin{equation*}
	\taff{\gla{}{\f}}\subsetneq \spinv{\f} \subseteq\dthree(\f)\;,
\end{equation*}	
and for fields of size $|\f|\geq n+1$ 
$$\spinv{\f} \subsetneq\dthree(\f)\;.$$
In addition,
\begin{equation}\label{eq:Tdense}
	\overline{\taff{\gla{}{\f}}} = \overline{\spinv{\f}}= \overline{\dthree(\f)} \;.
\end{equation}
\end{theorem}

Our main results for this section are a 
quasipolynomial-size hitting set for the class $\spinv{\f}$, and  a polynomial-size interpolating set for $\taff{\gla{}{\f}}$.

\begin{restatable}{theorem}{PITSINV}\label{thm:PITSINV}
	Let $0\neq g\in\f[\xn]$ have sparsity $\leq 2^t$.  Let $(A,\vb)\in\gla{n}{\f}$, and $f(\xn)=g(A\xn+\vb)$. Then, for any \ropinvgen[(t+1)]{} $\G$, $f\circ \G\neq 0$.
\end{restatable}

\begin{corollary}\label{cor:hs-sparse}
	For any integers $s,d,n$, there exists an explicit hitting set ${\cal H}\subset \f^n$, of size $|{\cal H}|= (nd)^{O(\log s)}$, such that $\cal H$ hits every nonzero polynomial $f \in \left(\Sigma^{[s]}\Pi^{[d]}\right)^{\gla{n}{\f}}$. If $|\f|\leq  n\cdot d$ then we let $\cH$ be defined over an extension field $\kk$ of $\f$ of size $|\kk| > n\cdot d$.
\end{corollary}

We next state our result concerning an interpolating set for $\taff{\gla{}{\f}}$.

\begin{restatable}{theorem}{PITsumSDMinv}\label{thm:PITsumSDMinv}
	Let $n,s_1,s_2,d_1,d_2\in\n$ be such that $n\geq s_1\cdot d_1,s_2\cdot d_2$. For $i\in\{1,2\}$ let $f_i\in \sdmpol[s_i][d_i]^{\gl{n}{\f}}$, and let $f=f_1-f_2$. 
	If $f\neq 0$, then any uniform \ropinvgen[6]{} $\G$ satisfies $f\circ \G\neq 0$.
\end{restatable}

Finally we note that the randomized reconstruction algorithm of Kayal and Saha \cite{kayal2019reconstruction}, which works for (as it is termed in their paper) ``non-degenerate'' homogeneous depth-$3$ circuits, works for $\taff{\gla{}{\f}}$. This follows from the observation that $\taff{\gla{}{\f}}$ circuits are always non-degenerate.

\begin{restatable}[special case of Theorem 1 of \cite{kayal2019reconstruction}]{theorem}{reconstructSDM}
	\label{thm:reconstructSDM}
	Let $n,d,s\in\n$, $n\geq(3d)^2$ and $s\leq(\frac{n}{3d})^\frac{d}{3}$. Let $\f$ be a field of characteristic zero or greater than $ds^2$. There is a randomized $\text{poly}(n,d,s)=\text{poly}(n,s)$ time algorithm which takes as input black-box access to a polynomial $f$ that is computable by a $\sdmpol[s][d]^{\gla{n}{\f}}$ circuit, and outputs a $\sdmpol[s][d]^{\gla{n}{\f}}$ circuit $\Phi$ computing $f$ with high probability. Furthermore, $\Phi$ is unique up to $\TPS[s,d][\f]$-equivalence (see \cref{sdmSymGroup}).
\end{restatable}

\begin{remark}\label{rem:zeroErrRecon}
	As in \cref{rem:ze-anf}, \cref{thm:PITsumSDMinv} enables us to convert the reconstruction algorithm of \cref{thm:reconstructSDM} to a zero-error algorithm, with expected polynomial running time. Given black-box access to some $f_1\in\taff{\gla{}{\f}}$, we define $f_2$ to be the output of the algorithm of \cref{thm:reconstructSDM} on input $f_1$, and then verify $f_1\equiv f_2$ by applying \cref{thm:PITsumSDMinv} to $f=f_1-f_2$.
\end{remark}

\subsubsection{Robust hitting sets?}

As we showed in \cref{tm:robust-hits}, if a hitting set $\mathcal H$ for a circuit class $\mathcal C$ is \emph{robust}, then $\mathcal H$ hits $\overline{\mathcal C}$ as well. It is thus natural to ask whether our interpolating sets are already robust. 
Our next result shows that the property of being a $t$-independent map, which was sufficient for the constructions in  \cref{thmhitcont,thm:cont-int,thm:PITROPINV,thm:pitSumOfRoanfThm,thm:PITSINV,thm:PITsumSDMinv} (for the appropriate values of $t$), by itself is not sufficient for obtaining robust hitting sets. We prove this by constructing an independent polynomial map which gives rise to a provably non-robust hitting set. Our construction is the same as the one given by Forbes et al. \cite{DBLP:conf/coco/ForbesSTW16} (Construction 6.3 in the full version).

\begin{theorem}\label{thm:tIndNotDense}
	Let $\f$ be of characteristic zero. For every $t$, there exists a uniform \ropinvgen[t]{} $\G$ and a nonzero polynomial $f$ such that $f\circ \G\equiv 0$, and $f$ can be computed by a $\dthree$ formula of size $t^{O(\sqrt t)}$. If $\f$ has a positive characteristic then $f$ can be computed by a $\dthree$ formula of size $t^{t}$, or by a general formula of size $t^{O(\log t)}$. Furthermore, for a certain arrangement of the variables in a $\sqrt{n}\times \sqrt{n}$ matrix, $f$ can be taken to be the determinant of any $(t+1)\times (t+1)$ minor.
\end{theorem}

\subsection{Polynomial Identity Testing}

So far we discussed our work from the perspective of dense subclasses of classes for which no strong lower bounds are known.  
Here we put our work in the context of the polynomial identity testing problem.

Polynomial Identity Testing (PIT for short) is the problem of designing efficient deterministic algorithms for deciding whether a given arithmetic circuit computes the identically zero polynomial. PIT has many applications, e.g. deciding primality \cite{Agrawal02primesis}, finding a perfect matching in parallel \cite{DBLP:journals/cacm/FennerGT19,DBLP:conf/focs/SvenssonT17} etc., and strong connection to circuit lower bounds  \cite{kabanets2004derandomizing,dvir2010hardness,DBLP:conf/coco/ChouKS18,DBLP:conf/focs/Guo0SS19}. See \cite{DBLP:journals/fttcs/ShpilkaY10,DBLP:journals/eatcs/Saxena09,Saxena-II} for surveys on PIT and \cite{DBLP:journals/eatcs/0001S19} for a survey of algebraic hardness-randomness tradeoffs.

PIT is considered both in the white-box model, in which we get access to the graph of computation of the circuit, and in the black-box model in which we only get query access to the polynomial computed by the circuit. Clearly, a deterministic PIT algorithm in the black-box model is equivalent to  a hitting set for the circuit class. In this work we only focus on the black-box model.

\paragraph{The continuant polynomial and algebraic branching programs:}

The continuant polynomial is trivially computed by width-$2$ \emph{Algebraic Branching Programs} (ABPs). Recall that an ABP of depth-$d$ and width-$w$ computes polynomials of the form $\trace\bra{M_1(\xn)\cdot \ldots\cdot M_d(\xn)}$, where each $M_i$ is a $w\times w$ matrix whose entries contain variables or field elements. Ben-Or and Cleve proved that every polynomial in $\VPe$ can be computed by a width-$3$ ABP of polynomial-size \cite{DBLP:journals/siamcomp/Ben-OrC92}.

Raz and Shpilka gave the first polynomial-time white-box PIT algorithm for read-once ABPs (ABPs in which every variable can appear in at most one matrix) \cite{DBLP:journals/cc/RazS05}.
Forbes, Saptharishi and Shpilka gave the first quasipolynomial-sized hitting set for read-once ABPs (ROABPs) \cite{forbes2013pseudorandomness}. This  result was slightly improved in \cite{DBLP:conf/approx/GuoG20} for the case where the width of the ROABP is small.
Anderson et al. gave a subexponential hitting set for read-$k$ ABPs \cite{DBLP:journals/toct/AndersonFSSV18}.
We note that none of these models is strong enough to contain the orbit $\ccont{\f}$.
For ABPs that are not constant-read we do not have sub-exponential time PIT algorithms. Thus, the following is an interesting open problem (recall that by the result of Ben-Or and Cleve a PIT algorithm for width-$3$ ABPs works for  $\VPe$ as well).

\begin{problem}
	Give a sub-exponential time PIT algorithm for ABPs of width-$2$.
\end{problem}

Although we do not have a PIT algorithm for general branching programs, in \cite{kayal2018reconstruction}  Kayal et al. gave an average-case reconstruction algorithm for low width ABPs.  Kayal, Nair and Saha obtained a significantly better algorithm in  \cite{DBLP:journals/cc/KayalNS19}. Their algorithm succeeds w.h.p, provided the ABP satisfies four non-degeneracy conditions (these conditions are defined in Section 4.3 of \cite{DBLP:journals/cc/KayalNS19}). However, the ABP computing the continuant polynomial does not satisfy the non-degeneracy conditions that are required for their algorithm to work. Thus, \cref{thm:cont-recon} does not follow from  \cite{DBLP:journals/cc/KayalNS19}.

To the best of our knowledge,  $\ccont{\f}$ is the first natural\footnote{It is hard to define what a natural class means, but, for example the set of all polynomials in $\VPe$ with a nonzero free term has a trivial hitting set, but is not a ``computational'' subclass.} computational class that is dense in $\VPe$ for which a polynomial (or even sub-exponential)-sized interpolating set (or a hitting set) is known.

\paragraph{Read-Once formulas:}

Hitting sets for read-once formulas were first constructed by Volkovich and Shpilka \cite{shpilka2015read}, who gave quasipolynomial-sized hitting set for the model, as well as a deterministic reconstruction algorithm of the same running time (earlier randomized reconstruction algorithms were known  \cite{DBLP:journals/siamcomp/BshoutyHH95,DBLP:journals/jcss/BshoutyB98}). Minahan and Volkovich obtained a polynomial-sized hitting set for the class, which led to a similar improvement in the running time of the reconstruction algorithm \cite{ROF}. Anderson, van Melkebeek and Volkovich constructed a hitting set of  size $n^{k^{O(k)}+O(k \log n) }$ for read-$k$ formulas \cite{DBLP:journals/cc/AndersonMV15}.
All these results work in a slightly stronger model in which we allow to label leaves with univariate polynomials, of polynomial degree, such that every variable appears in at most one polynomial, or with sparse polynomials on disjoint sets of variables.

The read-once models that we consider here, $\anf{}{\f}$ and $\ROF^{\gl{}{\f}}$, can be viewed as read-once formulas composed with a layer of addition gates with the restriction that the bottom layer of additions computes linearly independent linear functions. We note that these models do not fall into any of the previously studied models, as a variable can appear in all the linear functions.

As is the case with $\ccont{\f}$, our hitting sets for $\anf{}{\f}$ and $\ROF^{\gl{}{\f}}$  are the first sub-exponential-sized  hitting sets for  natural dense subclasses of $\VPe$.

\paragraph{Small depth circuits:}

The class of $\dtwo$ circuits was considered in many works, see e.g. \cite{DBLP:conf/stoc/Ben-OrT88,klivans2001randomness} and polynomial-sized hitting sets were constructed. 
The class of $\dthree$ circuits also received a lot of attention but with lesser success. Dvir and Shpilka \cite{dvir2007locally} and Karnin and Shpilka \cite{karnin2008black} gave the first quasipolynomial-time white-box and black-box PIT algorithms for $\Sigma^{[k]}\Pi^{[d]}\Sigma$ circuits, respectively. Currently, the best result is by Saxena and Seshadhri who gave a hitting set of size $(nd)^{O(k)}$ for such circuits  \cite{saxena2012blackbox}. In  \cite{DBLP:journals/cc/OliveiraSV16}
a subexponential-size hitting set for \emph{multilinear} $\dthree$ circuits was given.
In \cite{DBLP:journals/siamcomp/AgrawalSSS16}, Agrawal et al. gave a hitting set of size 
$n^{O(1)}\cdot (kd)^{O(r)}$ for
$\Sigma^{[k]}\Pi^{[d]}\Sigma$ circuits, where $r$ is an upper bound on the \emph{algebraic rank} of the multiplication gates in the circuit. Thus, known quasipolynomial-size hitting sets for subclasses of $\dthree$ circuits are known when the fan-in of the top gate is poly-logarithmic, or when the algebraic rank of the set of multiplication gates is poly-logarithmic. In contrast,  polynomials in $\taff{\gla{n}{\f}}$ and $\spinv{\f}$, when viewed as $\dthree$ circuits, 
can have polynomially many multiplication gates and their algebraic rank can be $n$. On the other hand, the corresponding $\dthree$ circuits are such that  the \emph{different} linear functions that are computed at their bottom layer  are linearly independent (when we view linear functions that are a constant multiple of  each other as the same function). Thus, our \cref{cor:hs-sparse} provides a hitting set for a new  subclass of $\dthree$ circuits.

To the best of our knowledge, our results for $\taff{\gla{}{\f}}$ and $\spinv{\f}$ give the first sub-exponential size hitting sets for natural subclasses that are dense in $\dthree$.

\subsection{More related work}

Approximations in algebraic complexity were first studied by Bini et al. in the context of algorithms for matrix multiplication \cite{BINI1979234}. For more on the history of border rank in the context of matrix multiplication see notes of chapter 15 in \cite{burgisser2013algebraic}. More recently, influenced by the GCT program, a lot of research was invested in trying to find polynomials characterizing tensors of small rank. See \cite{landsberg_2017} for a discussion on this approach. More recently, Kumar proved that \emph{every} polynomial over $\C$ can be approximated by a $\Sigma^{[2]}\Pi\Sigma$ circuit (of exponential degree) \cite{DBLP:journals/toct/Kumar20}.

Very little is known about the closure of circuit classes. Forbes observed that the class of ROABPs is closed \cite{ForbesWACT}. I.e. $\text{ROABP}=\overline{\text{ROABP}}$. We are not aware of other collapses or separation between general ``natural'' classes and their closures.

Beside the reconstruction algorithms mentioned earlier,  reconstruction algorithms are known for $\Sigma\Pi$ circuits \cite{DBLP:conf/stoc/Ben-OrT88,klivans2001randomness};  for random depth three \emph{powering} circuits \cite{AffProj};
for set-multilinear $\dthree$ and ROABPs \cite{beimel2000learning,DBLP:journals/toc/KlivansS06}; for $\dthree$ circuits with bounded top fan-in   \cite{shpilka2009interpolation,karnin2009reconstruction,DBLP:conf/coco/Sinha16}; and for multilinear depth-$4$ circuits with a constant top fan-in
\cite{gupta2012reconstruction,DBLP:conf/soda/BhargavaSV20}. 

In general, we do not expect the reconstruction problem to be solvable efficiently, as the problem of finding the minimal circuit computing a given polynomial  is a notoriously hard problem. A  detailed discussion on the hardness of reconstruction can be found in \cite{DBLP:journals/cc/KayalNS19}.

\subsection{Proof technique}

Our proofs are based on the following simple yet important, and as far as we know novel, observations concerning $k$-independent polynomial maps. Specifically, our proofs are based on the following two claims:

\begin{enumerate}
	\item If we have a hitting-set generator $H$ for nonzero polynomials of the form $\partiald{f}{x_1}$, for $f\in \cC$, and if $\G$ is a $1$-independent map then $H+\G$ hits every nonzero $f\in \cC$. This is proved in \cref{lem:indDerivLinear}.
	
	\item Similarly, we prove that if we have a hitting-set generator $H$ for nonzero polynomials of the form $\restr{f}{\ell=0}(A\xn+\vb)$, for $f\in \cC$, a linear function $\ell$, and an invertible affine transformation $(A,\vb)$, and if $\G$ is a $1$-independent map then $H+\G$ hits every nonzero $f\in \cC$. This follows from \cref{lem:indProjectZero}.
\end{enumerate}
	
By applying these claims $k+r$ times we get that composition with a $(k+r)$-independent map allows to reduce the problem of hitting a class $\cC$ to hitting polynomials of the form $\restr{\partiald{^k f}{x_{i_1}\partial x_{i_2}\cdots\partial x_{i_k}}}{\ell_1=\ldots=\ell_r=0}$. Thus, if we could prove that for a class $\cC$, there is such a sequence of derivatives and restrictions that simplifies the polynomials in it to a degree that they can be easily hit by some map $H$, then we conclude that $H+\G_{k+r}$, for a $(k+r)$-independent map $\G_{k+r}$, is a hitting set generator for $\cC$. 

It seems that all that is left to do is prove that for each of the orbits that we consider in Section~\ref{secResults} that is such small $k$ and $r$. However, a potential problem is that a partial derivative of the polynomial $g(\xn)=f(A\xn +\vb)$ gives $\partiald{g}{x_1}=\sum_{i=1}^{n} \partiald{f}{y_i}\cdot \partiald{\ell_i}{x_1}$, where $\ell_i$ is the $i$th coordinate of $A\xn+\vb$. Thus, it is no longer a derivative composed with an affine transformation but rather a sum of such derivatives, which could lead to polynomials outside of our class. For example, it is not hard ot prove that if we compose the ROF $y_1\cdot y_2 \cdot y_3$ with $(x_1,x_1+x_2,x_1+x_3)$ and then take a derivative according to $x_1$, then the resulting polynomial, $\partiald{\bra{x_1\cdot (x_1+x_2)\cdot (x_1+x_3)}}{x_1}=3x_1^2 + 2x_1\cdot (x_2+x_3) + x_2\cdot x_3$, is not in the orbit of any ROF. The solution to this problem is to take a \emph{directional derivative} in a direction coming from a \emph{dual basis}. For example if $\ell_i(\vv_j)=\delta_{i,j}$ then $\partiald{g}{\vv_1}=\partiald{f}{x_1}\bra{A\xn+\vb}$ (see \cref{lem:dualder}). Now, comes another important observation: If $H$ is a hitting-set generator  for nonzero polynomials of the form $\partiald{f}{\vv}$, for $f\in \cC$ and a direction $\vv$, and if $\G$ is a $1$-independent map then $H+\G$ hits every nonzero $f\in \cC$. The point is that if $\partiald{f}{\vv}\circ H \neq 0$ then for some $i$, $\partiald{f}{x_i}\circ H \neq 0$ and the claim follows from the first claim above. Thus, composition with $(k+r)$-independent maps allows us to reduce the problem of  hitting a class $\cC$ to finding a generator for polynomials that are obtained as a restriction to a subspace of co-dimension $r$ of a directional partial derivative of  order $k$ of polynomials in $\cC$.

Let us demonstrate this idea for the case of orbits of sparse polynomials. I.e. to polynomials of the form $g(\xn)=f(A\xn+\vb)$, where the number of monomials in $f$ is at most $2^t$. It is not hard to see that there is a variable $x_i$ such that if we consider $\restr{f}{x_i=0}$ and $\partiald{f}{x_i}$ then one of these polynomials has at most $2^{t-1}$ monomials.\footnote{This is not exactly accurate -- it only holds if $f$ is not divisible by some variable $x_i$. However, the case where there is a monomial dividing $f$ is also quite easy to handle as it is enough to hit the polynomial obtained after dividing by that monomial (since a composition with a $1$-independent map keeps any nonzero linear function nonzero).} Thus, after a a sequence of at most $t$ partial derivatives and restrictions, we get to a polynomial with only one monomial that we can easily hit.
Hence after at most $t$ directional derivatives and restrictions to a subspace, we get that $g$ is a product of linear forms, which we can easily hit. This proves that any $(t+1)$-independent map hits such nonzero polynomials $g$.

To obtain interpolating sets for our classes (and also a reconstruction algorithm for the orbit of the continuant polynomial), we prove that if two polynomials in the orbit, of any of the classes that we consider, are different, then there is a sequence of a few (directional) partial derivatives and restrictions that makes one of them zero while keeping the other nonzero. Using this and the ideas from above we construct our interpolating sets.  

\subsection{Discussion}
\label{sec:discuss}


As \cref{thm:tIndNotDense} shows, our hitting sets are not necessarily robust. It is thus an outstanding open problem to find a way to convert a hitting set to a robust one (recall Problem~\ref{prob:robust}).

The following toy example demonstrates that converting a hitting set for a class $\cC$ to a robust hitting set for $\cC$, cannot be done in a black-box manner and one has to use information about $\cC$ for that: let $\cC(\f)$ be the class of all polynomials with non-zero free term. A trivial hitting set for $\cC$ would simply be the singleton set $\mathcal H=\{\vzero\}$. On the other hand, it is clear that $\overline{\mathcal C}=\f[\xn]$, so making $\cH$ robust would yield a hitting set for \emph{all} polynomials. Note, however, that this is not a ``computational class.''

Another potential approach for obtaining robust hitting sets follows from the observation that the set of queries made by a non-adaptive 
deterministic black-box reconstruction algorithm, $\cal A$, for $\cC$, which is \emph{continuous} at $0$ (i.e. at the identically zero polynomial) is a robust hitting set for $\cC$. The reason is, that if $0\neq f\in\overline{\mathcal C}$ and $\{f_k\}_{k=1}^\infty\subseteq\mathcal C$ converges to $f$, then for large enough $k$: $\norm{f_k}_2\geq\frac{1}{2}\norm{f}_2>0$. As the $f_k$ sequence converges and polynomial evaluation is continuous (and their evaluation vectors are bounded), the sequence $\vv_k=\restr{f_k}{\mathcal H}\subseteq\C^{|\mathcal H|}$ must also converge to some vector $\vv=\restr{f}{\mathcal H}\in \C^{|\mathcal H|}$. If $\vv=\vzero$ then the continuity of $\mathcal A$ at $\vzero$ implies the coefficients of  the polynomials $f_k(\xn)$ must also converge to zero, as $\mathcal A(\vzero)= 0$. This would contradict $\norm{f_k}_2\geq\frac{1}{2}\norm{f}_2>0$ for large enough $k$, so $\vv\neq \vzero$ and thus $\mathcal H$ hits $\overline{\mathcal C}$.

Thus, an interesting challenge is to derandomize the reconstruction algorithms given in \cref{thm:cont-recon,thm:reconstructROANF,thm:reconstructSDM}, hoping that the resulting algorithms are continuous at $\vzero$.
We note however, that currently we do not even have efficient deterministic root-finding algorithms over $\C$. It is also known that in general, finding the minimal circuit for a polynomial can be very difficult. E.g., in \cite{DBLP:journals/jal/Hastad90,DBLP:conf/approx/Swernofsky18} it was shown that the question of computing, or even approximating, tensor rank, for degree $3$ tensors, is NP hard, over any field.


\begin{remark}
	In \cref{thm:PITsumSDMinv}, we have seen that any uniform \ropinvgen[O(\log(sn))]{} $\G$ is an interpolating set generator for $\sdminv$; i.e, $\G$ induces an interpolating set $\mathcal H$ for $\sdminv$. On the other hand, in \cref{thm:tIndNotDense}, we constructed such a map $\G$, with the additional property that $\G$ is \emph{not} a hitting set generator for \dthree{} circuits. In particular, this implies that the induced (non-efficient) reconstruction map $\mathcal A$ (that takes $f(\cH)$ and returns a circuit computing $f$) is not continuous at $\vzero$.
\end{remark}

%
%
%

We conclude this section with a somewhat vague question. 

\begin{problem}
	Find a ``computational'' class of polynomials $\mathcal C$ with a known hitting set $\mathcal H$, such that $\overline{\mathcal C}\neq\mathcal C$, and convert $\mathcal H$ to a robust hitting set. 
\end{problem}

We note that the closure of $\Sigma\bigwedge\Sigma$ circuits (i.e. circuits computing polynomials of the form $\sum_i \ell_i(\xn)^d$, for linear functions $\ell_i$) is contained in the class of commutative read-once algebraic branching programs (see \cite{forbes2013pseudorandomness}). Thus, the hitting set for the latter class gives a robust hitting set for the former \cite{forbes2013pseudorandomness}. However, we seek an example in which there is an ``interesting'' conversion of a hitting set to a robust one.

\subsection{Organization} The paper is organized as follows.  Section~\ref{sec:prelim} contains some more basic notations and definitions as well as characterization of the groups of symmetries of $\croanf$ and of $\sdmpol$. In Section~\ref{sec:k-ind} we give properties and constructions of $k$-independent polynomial maps and prove \cref{thm:tIndNotDense}. In Section~\ref{sec:cont} we study the continuant polynomial and prove \cref{thmhitcont,thm:cont-int,thm:cont-recon}. In Section~\ref{sec:rof} we study orbits of ROFs and ROANFs and prove \cref{ROANFisDense,thm:PITROPINV,thm:pitSumOfRoanfThm,thm:reconstructROANF}. Section~\ref{sec:sps} contains our results for subclasses of $\dthree$ circuits (\cref{thm:denseInDepthThree,thm:PITsumSDMinv,thm:PITSINV,thm:reconstructSDM}).
The appendix contains missing definitions that are required for explaining the reconstruction algorithm of 
\cite{gupta2014random}.

\section{Preliminaries}\label{sec:prelim}

\subsection{Notation}\label{sec:not}

For $k\in\n$, we denote $[k]\triangleq \{1,2,3,\ldots,k\}$ and $[k]_0\triangleq \{0,1,2,\ldots,k-1\}$. We use boldface lowercase letters to denote tuples of variables or vectors, as in $\xn=(x_1,\ldots,x_n)$, $\va =(a_1,\ldots,a_m)$, when the dimension is clear from the context. For any two elements $i,j$ coming from some set $S$ (usually $i$ and $j$ will be numbers), $\delta_{i,j}$ equals $1$ when $i=j$ and $0$ otherwise. For every $m\in \n$ we denote with $I_m$ the $m\times m$ identity matrix. When we wish to treat the entries of a matrix $A$ as formal variables, we use boldface $\bm A$. We will note use capital bold face letters other than to denote such matrices.

For an exponent vector $\va=(a_1,\ldots,a_n)\in\n^n$, we denote  $\xn^{\va}\triangleq \prod_{i=1}^n x_i^{a_i}$. In some cases we shall consider ``monomials'' with respect to set of linear functions $\cbra{\ell_i}_{i=1}^{m}$: for an exponent vector $\ve=(e_1,\ldots,e_m)\in {\n}^m$ we denote  ${\bm \ell}^{\ve} = \prod_{i=1}^m\ell_i^{e_i}$ and refer to it as  an \emph{$\cbra{\ell_i}$-monomial}.
For a polynomial $f(\xn)$ we define the \emph{monomial support} of $f$, denoted $\text{\normalfont mon}(f)$, as the set of monomials with non-zero coefficient in $f$. The \emph{variable set} of $f$, denoted $\text{var}(f)$, is the set of variables that $f$ depends on. I.e., all variables that appear in $\text{\normalfont mon}(f)$.
The individual degree of a variable $x_i$ in $f(\xn)$ is the degree of $f$ as a polynomial in $x_i$.
A polynomial $f\in\f[\xn]$ of $\deg(f)\leq 1$ is called a linear function, and if $f$ is homogeneous then it is called a \emph{linear form}. 
For a polynomial $f\in\f[\xn]$ and an integer $k\in\n$ we denote by $f^{[k]}$ the degree-$k$ homogeneous part of $f(\xn)$,i.e. the sum of all monomials of $f$ of degree exactly $k$. In particular,
$$ f(\xn)=f^{[0]}(\xn)+f^{[1]}(\xn)+\ldots+f^{[\deg(f)]}(\xn) \;.$$
Note that for a linear function $f$, $f^{[1]}$ is a linear form. We say that a polynomial $f$ is homogeneous of degree $k$ or that $f$ is $k$-homogeneous if $f=f^{[k]}$.
We say a set of linear functions $\cbra{\ell_1(\xn),\ldots,\ell_n(\xn)} \subset \f[\xn]$ is \emph{linearly independent} if the set $\cbra{\ell_i^{[1]}}$ is linearly independent.\footnote{Note that by our definition, $x$ and $x+1$ are linearly dependent.}
Given a polynomial $f(\xn)$, a subset of variables $\bm y\subseteq\{x_1,\ldots,x_n\}$ and an assignment to those variables $\va\in\f^{|\bm y|}$, we denote by $\restr{f}{\bm y=\va}\in\f[\xn\setminus\bm y]$ the polynomial resulting from assigning the values of $\va$ to the variables of $\bm y$ in $f(\xn)$. We sometimes abuse notation and write $\bm y\subseteq[n]$ to indicate the indices of the assigned variables instead of the variables themselves.

Given an arithmetic circuit $\Phi$, we frequently denote by $\Phi(\xn)$ or, abusing notation, by $\Phi$, the polynomial computed at the output node of $\Phi$. 
Given a class of arithmetic circuits $\cC$ and a polynomial $f\in\f[\xn]$, we say $f\in\cC$ if $f$ can be computed by some circuit from $\cC$. For a circuit class $\cC(\f)$ we denote by $\overline{\mathcal C}(\f)$ the \emph{closure} of $\cC(\f)$, as in \cref{def:approx}.

\subsection{Groups of matrices and their action}

We first list some simple properties of composition with a linear (or affine) transformation that we shall use implicitly.


\begin{observation}\label{obsTrivProps}
	For any $m$ variate polynomial $f(x_1,\ldots,x_m)$ and $n\geq m$:
	\begin{itemize}
		\item For any $A\in \gl{n}{\f}$ and $d\in\n$, $f^{[d]}(A\xn)$ is the $d$-homogeneous part of $f(A\xn)$.
		\item For any $A\in \gla{n}{\f}$, $f(\xn)$ is irreducible if and only if $f\bra{A\xn}$ is irreducible.
		\item The set of matrices $A$ for which $f(\xn)=f(A\xn)$ forms a multiplicative subgroup of $\gl{n}{\f}$ and a similar claim holds for  $\gla{n}{\f}$.
	\end{itemize}
\end{observation}

We next define some special groups that serve as group of symmetries of some of the models that we consider. We first define the group of symmetries of $\croanf(\xn)$.

\begin{definition}\label{defTR}
	For $m,\Delta\in\n$ such that  $m=2^\Delta$, the \emph{tree-symmetry group} $\TR[m][\f]$ denotes the automorphisms of a rooted complete binary tree of depth $\Delta$. It is defined recursively as follows.
	\begin{itemize}
		\item For $m=1$, $\TR[1][\f]$ consists only of the identity matrix.
		\item For $m>0$, $\TR[m][\f]$ is generated by matrices of the form
		$$ \begin{pmatrix}
			A&0\\0&B
		\end{pmatrix}\;\;\;and\;\;\;\begin{pmatrix}
			0&I_{\frac{m}{2}}\\I_{\frac{m}{2}}&0
		\end{pmatrix} $$
		where $A,B\in \TR[\frac{m}{2}][\f]$.
	\end{itemize}
\end{definition}

\begin{definition}\label{defTS}
	For any $m=4^\Delta$, the \emph{tree-scale group} $\TS[m][\f]$ is the group generated by elements of $\TR[m][\f]$ and matrices of the form
	$$ \begin{pmatrix}
		\alpha I_{\frac{m}{4}}&0&0&0\\
		0&\alpha^{-1}I_{\frac{m}{4}}&0&0\\
		0&0&\beta I_{\frac{m}{4}}&0\\
		0&0&0&\beta^{-1}I_{\frac{m}{4}}
	\end{pmatrix} $$
	where $0\neq \alpha,\beta\in\f$.
\end{definition}

The importance of the group  $\TS[m][\f]$ stems from the fact that it is the symmetry group of $\croanf$. To intuitively see why this is the case, notice that in any representation of an ANF one may swap children of any node without changing the output polynomial. We call such symmetries ``tree-symmetries'' and they are captured by the group $\TR[n][\f]$. A second source of ambiguity comes from the fact that we can rescale the formula. Recall that the output polynomial is of the form $f_1\cdot f_2 + f_3 \cdot f_4$ (\cref{def:ANF}). Clearly, the output does not change if we replace $f_1$ by, say, $2f_1$ and $f_2$ by $f_2/2$. Such rescaling symmetries are captured by the group  $\TS[n][\f]$. Finally, another source for ambiguity comes from the fact that the quadratic polynomials computed at the bottom two layers of the ANF may have different representations. For example, 
$$4xy+4wz = (x+y+w-z)\cdot(x+y-w+z) + (w+z+x-y)\cdot (w+z-x+y)\;.$$ As there is an infinite number of representations for each quadratic polynomial (over infinite fields), we can expect to characterize the symmetries in term of the quadratics computed at the bottom two layers of the ANF. 

\begin{fact}[Special case of Theorem 5.43(iii) of \cite{gupta2014random}]\label{fact:anf-sym}
	Let $m,\Delta,n\in \n$ such that  $m=4^{\Delta-1}\leq n/4$. Let	$f=\croanf[\Delta](\ell_1,\ldots,\ell_{4m})\in \croanf^{\gla{n}{\f}}$. Let $Q=(q_1,\ldots,q_m)$ be the list of quadratic polynomials that are computed at the bottom two layers of the formula $\croanf[\Delta](\ell_1,\ldots,\ell_{4m})$. In particular, $f=\croanf[\Delta-1](q_1,\ldots,q_m)$. 
If $Q'=(q'_1,\ldots, q'_m)$ is any other $m$-tuple of quadratic polynomials for which $f=\croanf[\Delta-1](q'_1,\ldots,q'_m)$
then $Q$ is $\TS[m][\f]$-equivalent to $Q'$.	
\end{fact}

Next, we define the group of symmetries of $\sdmpol(\xn)$.

\begin{definition}
	For any $n\in\n$ the \emph{permutation-scale group}, denoted $\PS[n][\f]$, is the set of all matrices $A\in\GL[n][\f]$ which are row-permutations of 
	non-singular diagonal matrices with determinant one.
\end{definition}

For example,  $\begin{pmatrix}
	0 & -2 & 0 \\
	0 & 0 & -1 \\
	1/2 & 0 & 0 \end{pmatrix}\in \PS[3][\C]$.

\begin{definition}\label{sdmSymGroup}
	Let $s,d,n\in\n$ such that $n= s\cdot d$. A matrix $A\in\GL[n][\f]$ is a member of the \emph{tensor permutation-scale group}, denoted $\TPS[s,d][\f]$, if $A= ( P\otimes I_d) \cdot B$, where $P$ is an $s\times s$ permutation matrix and $B=\begin{pmatrix}
		B_1 & 0 & \ldots & 0 \\
		0 & B_2 & \ldots & 0 \\
		\vdots &  & \ddots & \vdots \\
		0 & \ldots & 0 & B_d
	\end{pmatrix}$ is a block diagonal matrix such that  each   block $B_i$ of $B$ satisfies  $B_i\in\PS[d][\f]$. 
\end{definition}
For example, for $s=d=2$ the matrix $A= \begin{pmatrix}
	0 & 0 & 0 & 2 \\
	0 & 0 & 1/2 & 0 \\
	-1 & 0 & 0 & 0 \\
	0 & -1 & 0 & 0
\end{pmatrix}$
is in $\TPS[2,2][\C]$, as for $P= \begin{pmatrix}
	0 & 1\\
	1 & 0
\end{pmatrix}$ and $B=\begin{pmatrix}
	-1 & 0 & 0 & 0 \\
	0 & -1 & 0 & 0\\
	0 & 0 & 0 & 2 \\
	0 & 0 & 1/2 & 0
\end{pmatrix}$, we have $A=\bra{P\otimes I_2}\cdot B$, and clearly each block of $B$ is in $\PS[2][\C]$.\\

 Another way of defining the group is as follows: index rows and columns of $A$ with pairs $(i,j)\in [s]\times [d]$. Then, $A \in \TPS[s,d][\f]$ if and only if 
 there exists a permutation $\pi:[s]\to[s]$, and for all $i\in[s]$ permutations $\theta_i:[d]\to[d]$ and constants $\alpha_{i,j}$ satisfying $\prod_{j=1}^d\alpha_{i,j}=1$, such that $A_{(i,j),(i',j')}= \delta_{\pi(i),i'}\cdot \delta_{\theta_i(j),j'}\cdot \alpha_{i,j}$ for all $i,j$. 

We next prove that $\TPS[s,d][\f]$ is the group of symmetries of $\sdmpol(\xn)$. In other words, we show that  $\sdmpol(\xn)=\sdmpol(A\xn)$ if and only if $A\in \TPS[s,d][\f]$.
Intuitively, $\sdmpol$ admits no symmetries other than the trivial ones: permutations on the product gates, and internal permutation-scale of each product gate such that the product of the scale coefficients is $1$. This is exactly captured by the group $\TPS[s,d][\f]$, which is therefore contained in  the group of symmetries of $\sdmpol(\xn)$.

\begin{restatable}{lemma}{sdmSymmetry}\label{sdmSymmetry}
	Let $s,d,n\in\n$, such that $d>2$ and $n=s\cdot d$. If $A\in\gl{n}{\f}$ satisfies $\sdmpol(\xn)=\sdmpol(A\xn)$, then $A\in\TPS[s,d][\f]$.
\end{restatable}

\begin{proof}[Proof of \cref{sdmSymmetry}]


%
%
%

\sloppy
Fix linear forms $\ell_{1,1},\ldots,\ell_{s,d}$ such that the $(i,j)$th coordinate of $A\xn$ (using the indexing $[n]=[s]\times [d]$) is $\ell_{i,j}(\xn)$, and $\sdmpol(A\xn)=\sum_{i=1}^s\prod_{j=1}^d\ell_{i,j}(\xn)$. By the discussion above, our goal is to prove that there exists a permutation $\pi:[s]\to[s]$, and  for all $i\in[s]$ permutations $\theta_i:[d]\to[d]$ and constants $\alpha_{i,j}$ satisfying $\prod_{j=1}^d\alpha_{i,j}=1$, such that $\ell_{i,j}(\xn)=\alpha_{i,j}\cdot x_{\pi(i),\theta_i(j)}$ for all $i,j$. 
Fix some $i\in[s]$ and take a derivative of the equation $\sdmpol(\xn)=\sdmpol(A\xn)$ by $x_{i,1}$:
\begin{equation}\label{eq:der-tsd}
	\prod_{j\in\{2,\ldots,d\}}x_{i,j}=\partiald{\sdmpol(\xn)}{x_{i,1}}=\partiald{\sdmpol(A\xn)}{x_{i,1}}=\sum_{r=1}^s\partiald{}{x_{i,1}}\left(\prod_{j=1}^d\ell_{r,j}(\xn)\right)\;.
\end{equation}
For $r\in[s]$, denote $h_{i,r}(\xn)\triangleq \partiald{}{x_{i,1}}\left(\prod_{j=1}^d\ell_{r,j}(\xn)\right)$. As $d>2$, the LHS of Equation~\eqref{eq:der-tsd} is a reducible polynomial, so $\sum_{r=1}^sh_{i,r}(\xn)$ is also reducible. Composition with a non-singular matrix preserves reducibility, so $\sum_{r=1}^sh_{i,r}(A^{-1}\xn)$ is also reducible. However, $h_{i,1}(A^{-1}\xn),\ldots,h_{i,s}(A^{-1}\xn)$ are $s$ variable-disjoint, multilinear polynomials, each of which is either $(d-1)$-homogeneous or zero. Thus, by \cref{irreducibleMultilinear} below, at most one $h_{i,r}(A^{-1}x)$ can be non-zero. Accordingly, for every variable $x_{i,j}$ there exists a unique $i'$ such that $x_{i,j}\in\text{var}\left(\prod_{j'=1}^d\ell_{i',j'}(\xn)\right)$. Thus, for some $i'$ we have
\begin{align}
	\prod_{j\in\{2,\ldots,d\}}x_{i,j}=\partiald{}{x_{i,1}}\left(\prod_{j=1}^d\ell_{i',j}(\xn)\right)\;.\label{eqRefOneDeriv}
\end{align}
For any $j>1$, if we take a derivative of \eqref{eqRefOneDeriv} by $x_{i,j}$ then the LHS is clearly non-zero. Thus, both $x_{i,1}$ and $x_{i,j}$ exist in $\text{var}\left(\prod_{j'=1}^d\ell_{i',j'}(\xn)\right)$, proving variables in the same product gate of $\sdmpol(\xn)$ are mapped to the same product gate of $\sdmpol(A\xn)$. A similar argument  shows that variables from distinct product gates of $\sdmpol(\xn)$ are mapped to  different product gates of $\sdmpol(A\xn)$. It follows that  product gates of $\sdmpol(A\xn)$ are variable-disjoint and that there exists a permutation $\pi:[s]\to[s]$ satisfying
$$ \forall i\in[s]:\qquad\text{var}\left(\prod_{j=1}^d\ell_{i,j}(\xn)\right)=\{x_{\pi(i),1},\ldots,x_{\pi(i),d}\} \;.$$
In particular, there can be no cancellations between different product gates of $\sdmpol(A\xn)$. Therefore, by multilinearity, for every $i\in[s]$, the linear forms $\ell_{i,1}(\xn),\ldots,\ell_{i,d}(\xn)$ must be variable-disjoint.
Exactly $d$ variables appear in $\prod_{j=1}^d\ell_{i,j}(\xn)$, so for every $i\in[s]$ and $j\in[d]$ there exists a permutation $\theta_i:[d]\to[d]$ and a non-zero constant $\alpha_{i,j}\in\f$ such that $\ell_{i,j}(\xn)=\alpha_{i,j}x_{\pi(i),\theta_i(j)}$. As $\prod_{j=1}^d\alpha_{i,j}$ is the coefficient of $\prod_{j=1}^dx_{\pi(i),j}$ in $\sdmpol(A\xn)$, this product must be $1$, which completes the proof.
\end{proof}

\begin{observation}\label{irreducibleMultilinear}
	If $f,g$ are non-constant, variable-disjoint, multilinear polynomials, then for every $c\in\f$ the polynomial $f(\xn)+g(\xn)+c$ is irreducible.
\end{observation}

\section{$k$-independent polynomial maps and their properties}\label{sec:k-ind}


%

All the hitting and interpolating sets that we construct are based on $k$-independent polynomial maps (\cref{def:indPolyMap}). We next give some simple properties of independent polynomial maps, that follow immediately from the definition.

\begin{observation}\label{obs:kwise}
	It holds that
\begin{enumerate}
	\item \label{reduceIndPolyMap}
	If $\G(\bm y,\bm z)$ is a \ropinvgen[(k+1)]{}, then there exists a subset of variables $S$ and an assignment $\valpha\in\f^{|S|}$ such that $\restr{\G}{S=\valpha}$ is a \ropinvgen[k]{}.
	\item \label{item:coordsGenInd}
	For any $k\geq 1$, the $n$ coordinates of any \ropinvgen[k]{} are $\f$-linearly independent.	
	\item \label{item:ind-lin} Let $\ell_1(\xn)$ and $\ell_2(\xn)$ be linearly independent linear functions in $\f[\xn]$. Let $\G({\bf y},z_1,z_2)$ be any \ropinvgen[2]{}. Consider $\ell_1 \circ \G$ and $\ell_2 \circ \G$ as polynomials in $z_1,z_2$ over $\f({\bf y})$. Then, $\bra{\ell_1 \circ \G}^{[1]}$ and $\bra{\ell_2 \circ \G}^{[1]}$ are linearly independent, as linear forms in $z_1,z_2$ over $\f({\bf y})$. 
\end{enumerate}
\end{observation}

We next give the construction of \cite{shpilka2015read} of  a \ropinvgen[k]{}
(denoted $G_k$ in \cite{shpilka2015read}).

\begin{definition}\label{defSVGEN}
Fix $n$ and a set of $n$ distinct field elements $\mathcal A=\{\alpha_1,\ldots,\alpha_n\}\subseteq\f$.\footnote{If $|\f|<n$ then we take these elements from an appropriate extension field of $\f$.} For every $i\in[n]$ let $L_i(w):\f\to\f$ be the $i$th Lagrange Interpolation polynomial for the set $\mathcal A$. That is, each $L_i(w)$ is polynomial of degree $n-1$ that satisfies $L_i(\alpha_j)=\delta_{i,j}$.
We define $\svgen_1(y_1,z_1):\f^2\to\f^n$ as:
$$ \svgen_1(y_1,z_1)\triangleq\left(L_1(y_1)\cdot z_1,L_2(y_1)\cdot z_1,\ldots,L_n(y_1)\cdot z_1\right), $$
and for any $k\geq 1$, we define $\svgen_k:\f^{2k}\to\f^n$ as:
$$ \svgen_k(\bm y,\bm z)\triangleq\svgen_1(y_1,z_1)+\svgen_1(y_2,z_2)+\ldots+\svgen_1(y_k,z_k)=\left(\sum_{j=1}^kL_1(y_j)\cdot z_j,\sum_{j=1}^kL_2(y_j)\cdot z_j,\ldots,\sum_{j=1}^kL_n(y_j)\cdot z_j\right). $$
\end{definition}

\begin{observation}\label{obs:SVgen}
	$\svgen_k$ is a $k$-independent polynomial map, in which each variable has degree at most $n-1$.
\end{observation}

%

The generator $\svgen_k$ can be converted to a uniform \ropinvgen[k]{} by adding another $k$ control variables $y_{k+1},\ldots,y_{2k}$, and swapping out the $L_i(y_j)$s for their homogenizations $y_{j+k}^{n-1}L_i\bra{\frac{y_j}{y_{j+k}}}$:

\begin{definition}\label{defsvgenhom}
With the notation used in \cref{defSVGEN}, define the \emph{uniform SV-generator} with $k$ independence $\svgenh_k:\f^{3k}\to\f^n$ as:
\begin{align*}
\svgenh_k\bra{y_1,\ldots,y_{2k},z_1,\ldots,z_k}&\triangleq y_{1+k}^{n-1}\cdot\svgen_1\bra{\frac{y_1}{y_{1+k}},z_1}+y_{2+k}^{n-1}\cdot\svgen_1\bra{\frac{y_2}{y_{2+k}},z_2}+\ldots +y_{2k}^{n-1}\cdot\svgen_1\bra{\frac{y_k}{y_{2k}},z_k}\\
&=\left(\sum_{j=1}^ky_{j+k}^{n-1}L_1\bra{\frac{y_j}{y_{j+k}}}\cdot z_j,\sum_{j=1}^ky_{j+k}^{n-1}L_2\bra{\frac{y_j}{y_{j+k}}}\cdot z_j,\ldots ,\sum_{j=1}^ky_{j+k}^{n-1}L_n\bra{\frac{y_j}{y_{j+k}}}\cdot z_j\right).
\end{align*}
\end{definition}

\begin{observation}
$\svgenh_k$ is a uniform \ropinvgen[k]{}, with  individual degrees at most $n-1$.
\end{observation}

We next show how we can use $k$-independent polynomial maps in order to, roughly, simulate a $k$th order directional derivative or, project a polynomial to a subspace of co-dimension $k$. We first need to define the notion of a directional derivative.

\begin{definition}\label{def:DeriveLinear}
	For an $n$-variate polynomial $f\in\f[\xn]$ and  $\vv=(v_1,\ldots,v_n)\in\f^n$, the \emph{derivative of $f(\xn)$ in the direction $\vv$} is defined as:
	$$ \partiald{f}{\vv}=\sum_{i=1}^n v_i\cdot\partiald{f}{x_i}. $$
\end{definition}

If $\f$ has positive characteristic then by $\frac{\partial F}{\partial x_i}$ we refer to the formal derivative (which in the case of fields of characteristic zero is equal to the analytical definition). Observe that we still have that $$\frac{\partial^2 f}{\partial y \partial x} = \frac{\partial^2 f}{\partial x \partial y}  \text{,} \quad 
\frac{\partial (fg)}{\partial x} = \frac{\partial f}{\partial x} \cdot g+ \frac{\partial g}{\partial x} \cdot f  \quad \text{and} \quad \frac{\partial f\bra{g_1(\xn),\ldots,g_m(\xn)}}{\partial x_k} =\sum_{i=1}^{m}\frac{\partial f}{\partial y_i}\bra{g_1(\xn),\ldots,g_m(\xn)} \cdot \frac{\partial g_i}{\partial x_k}\;, $$ 
where in the last expression $f$ is an  $m$ variate polynomial, and  $g_1,\ldots,g_m$ are $n$ variate polynomials.

We shall often take derivatives according to a \emph{dual set} to a set of linearly independent linear functions:
\begin{definition}\label{def:dual}
A dual set for $m$ linearly independent linear functions (recall that we say that linear functions are linearly independent if and only if their degree-$1$ homogeneous parts are linearly independent)  in $n\geq m$ variables,   $\ell_1(\xn),\ldots,\ell_m(\xn)$ is a set of $m$ vectors $\cbra{{\vv_i}}\subset \f^n$ such that $\ell_i^{[1]}({\vv_j})=\delta_{i,j}$. 	
\end{definition}

\begin{lemma}\label{lem:dualder}
	Let $\ell_1,\ldots,\ell_m \in \f[x_1,\ldots,x_n]$, for $n\geq m$, be linearly independent linear functions. Let $\cbra{{\vv_i}}\subset \f^n$ be a dual set. Let $g\in\f[y_1,\ldots,y_m]$ be a polynomial. 
	Then, for $f(\xn)=g\bra{\ell_1(\xn),\ldots,\ell_m(\xn)}$
	it holds that
	$$ \partiald{f}{\vv_i}(\xn)=\partiald{g}{y_i}\bra{\ell_1(\xn),\ldots,\ell_m(\xn)}\;. $$
\end{lemma}
\begin{proof}
	\begin{align*}
			\partiald{f}{\vv_i}(\xn) &=\sum_j v_{i,j}\cdot\partiald{f}{x_j}(\xn)=\sum_{j,k}v_{i,j}\cdot
			\partiald{\ell_k}{x_j}\cdot\partiald{g}{y_k}\bra{\ell_1(\xn),\ldots,\ell_m(\xn)}\\
			&=\sum_k\ell_k^{[1]}(\vv_i)\cdot\partiald{g}{y_k}\bra{\ell_1(\xn),\ldots,\ell_m(\xn)}=\partiald{g}{y_i}\bra{\ell_1(\xn),\ldots,\ell_m(\xn)}\;. \qedhere
	\end{align*}
\end{proof}

\begin{lemma}\label{lem:indDerivLinear}
	Let $f\in\f[\xn]$ where $\xn=(x_1,\ldots,x_n)$. Let $H(\bm w):\f^t\to\f^n$ be a polynomial map in variables $\bm w$, and let $\G(\bm y,\bm z)$ be a \ropinvgen[k]{} such that $\text{\normalfont var}(H)\cap\text{\normalfont var}(\G)=\emptyset$. Then, for any  $\vv_1,\ldots,\vv_k\in\f^n$:
	$$ \partiald{^k f}{\vv_1\partial\vv_2\cdots\partial\vv_k}\circ H\neq 0\quad \Rightarrow \quad f\circ(\G+H)\neq 0\;. $$
\end{lemma}
\begin{proof}
	By definition of \ropinvgen[k]{}s, $\G=\G_1(\bm{y_1},z_1)+\ldots+\G_k(\bm{y_k},z_k)$ for some variable-disjoint \ropinvgen[1]{}s $\G_1,\ldots,\G_k$. It is therefore enough to prove the lemma for $k=1$, as we can replace $f$ with $\partiald{^{k-1} f}{\vv_2\cdots\partial\vv_k}$, 
	$H$ with $H+\G_2+\ldots+\G_k$ and $\G$ with $\G_1$; by iterative application of the result for $k=1$, we will get the general result for an arbitrary $k\in\n$.
	
	Denote $H=(H_1,H_2,\ldots,H_n)$. By \cref{def:DeriveLinear},  the condition $\partiald{f}{\vv}\circ H\neq 0$ implies that there exists some $i\in[n]$ such that $\partiald{f}{x_i}\circ H\neq 0$. Assume, WLOG, $\partiald{f}{x_1}\circ H\neq 0$. As $\G$ is a \ropinvgen[1]{}, there exists some $\valpha\in\f^{|\bm{y_1}|}$ such that $\restr{f\circ(\G+H)}{\bm{y_1}=\valpha}=f(z_1+H_1,H_2,\ldots,H_n)$; denote $g\triangleq \restr{f\circ(\G+H)}{\bm{y_1}=\valpha}$. As no coordinate of $H$ depends on $z_1$:
	$$ \partiald{g}{z_1}= \partiald{\bra{z_1+H_1}}{z_1}\cdot \partiald{f}{x_1}(z_1+H_1,H_2,\ldots,H_n)= 1\cdot\left(\partiald{f}{x_1}\right)(z_1+H_1,H_2,\ldots,H_n) $$
	and therefore:
	$$ \restr{\partiald{g}{z_1}}{z_1=0}=1\cdot\left(\partiald{f}{x_1}\right)(0+H_1,H_2,\ldots,H_n)=\left(\partiald{f}{x_1}\right)\circ H\neq 0 \;.$$
	As $g$ is a projection of $f\circ(\G+H)$, it follows that $f\circ (\G+H)\neq 0$.
\end{proof}

The next lemma shows how to use $k$-independent maps in order to project a polynomial to a subset of its coordinates.

\begin{lemma}\label{lem:indProjectZero}
	Let $m\leq n\in \n$ and $g(\bw)\in\f[w_1,\ldots,w_m]$. Let $f(\xn)=g(\ell_1(\xn),\ldots,\ell_m(\xn))$ for linearly independent linear functions $\ell_1(\xn),\ldots,\ell_m(\xn)$.  Let $\G(\bm y,\bm z)$ be a \ropinvgen[k]{}. For a set $S\subseteq[n]$ of size $k$ denote by $\tilde g(x_i:i\in[m]\setminus S)=\restr{g}{S=0}$ the projection of $g$ to the variables outside of $S$. Then, there exist linearly independent linear functions $\{\tilde\ell_i(\xn):i\in[m]\setminus S\}$, additional linear functions $\vL(\xn)=(L_1(\xn),\ldots,L_k(\xn))$ and an assignment $\valpha\in\f^{|\bm y|}$ such that:
	$$f(\xn+\G(\valpha,\vL(\xn))) =  \tilde g(\tilde\ell_i(\xn):i\in[m]\setminus S)\;.$$
\end{lemma}
\begin{proof}
	It is enough to prove the lemma for the case $k=1$, as we may then define $\tilde f(\xn)\triangleq f(\xn+\G(\valpha,L_1(\xn)))=\tilde g(\tilde\ell_1(\xn),\ldots,\tilde\ell_{m-1}(\xn))$ and apply the result iteratively. Thus, assume $k=1$, and WLOG assume $S=\{x_1\}$ (thus, $\tilde g(w_2,\ldots,w_m)=g(0,w_2,\ldots,w_m)$).
	
Let $x_i$ be some variable with a non-zero coefficient in $\ell_1(\xn)$. Such a variable exists as the $\ell_j$s are linearly independent. For $j\in [m]$, denote $\beta_j= \partiald{\ell_j}{x_i}$, i.e. $\beta_j$ is the coefficient of $x_i$ in $\ell_j$. By our choice of $i$, $\beta_1 \neq 0$. Choose some $\valpha\in\f^{|\bm y|}$ such that $\G(\valpha,z_1)$ has $z_1$ in the $i$th coordinate, and $0$ in all other coordinates. Define $L(\xn)\triangleq -\frac{\ell_1(\xn)}{\beta_1}$, so we get:
	$$ f(\xn+\G(\valpha,L(\xn))=f\bra{x_1,x_2,\ldots,x_{i-1},x_i-\frac{\ell_1(\xn)}{\beta_1},x_{i+1},\ldots,x_n}.$$
	Observe that for every $i$,
	$$\ell_i\bra{\xn+\G(\valpha,L(\xn)} = \ell_i\bra{x_1,x_2,\ldots,x_{i-1},x_i-\frac{\ell_1(\xn)}{\beta_1},x_{i+1},\ldots,x_n} = \ell_i(\xn)-\frac{\beta_i}{\beta_1}\cdot \ell_1(\xn) \;.$$
	In particular, $\ell_1\bra{\xn+\G(\valpha,L(\xn)} =0$.
	For $i=2,\ldots,m$, define:
	$$\tilde\ell_i(\xn)\triangleq \ell_i(\xn)-\frac{\beta_i}{\beta_1}\cdot \ell_1(\xn) \;.$$
	As $\ell_1,\ldots,\ell_m$ are linearly independent, it follows that $\tilde\ell_2,\ldots,\tilde\ell_m$ are also linearly independent. We get that
	\begin{equation*}
		f(\xn+\G(\valpha,L(\xn)))=g(0,\tilde\ell_2(\xn),\ldots,\tilde\ell_m(\xn))=\tilde g(\tilde\ell_2(\xn),\ldots,\tilde\ell_m(\xn))\;. \qedhere
	\end{equation*}
\end{proof}

\subsection{Proof of \cref{thm:tIndNotDense}}

	We next prove that there are $k$-independent maps that are provably not robust. The proof is by giving a different construction of such maps that, for an appropriate arrangement of the $n$ variables in a matrix, is guaranteed to output matrices of rank at most $k$. Thus, a determinant of any $(k+1) \times (k+1)$ minor, a polynomial that has small formulas for small values of $k$, vanishes on the output of any such map.	
	
	The fact that such a construction exists 
	was already noticed in \cite{DBLP:conf/coco/ForbesSTW16} (Construction 6.3 of the full version of the paper). For completeness we repeat the construction here. 
	
	\begin{proof} \emph{(of \cref{thm:tIndNotDense})}
		Fix the number of variables $n$ and assume WLOG $n$ is a perfect square, i.e., $n=m^2$. We index the variables as $x_{i,j}$ for $i,j\in[m]$. 
		We let $f=\det_{t+1}$. 
		By \cite{gupta2013arithmetic}, over fields of characteristic zero, $f$ has a $t^{O(\sqrt t)}=O(n)$ sized $\dthree$ formula, which is polynomial in $n$ for $t = O\left( \left(\log n / \log \log n \right)^2\right)$. Over fields of positive characteristic the formula size is quasipolynomial in $t$, and the $\dthree$ complexity is at most $t!$, which is polynomial in $n$ for $t = O\left(\log n / \log \log n \right)$. 
		
		Denote by $\bm{M}$ the $(t+1)\times(t+1)$ symbolic matrix of variables $\bm{M}_{i,j}=x_{i,j}$. We first construct a uniform \ropinvgen[1]{} $\G_1$ such that $\bm{M}\circ \G_1$ is of rank $1$, and define $\G$ to be a sum of $t$ variable-disjoint copies of $\G_1$. As $rank(\bm{M}\circ \G_1)=1$, we have $rank(\bm{M}\circ \G)\leq t$ so $\det_{t+1}(\bm{M}\circ \G)=0$, as required. We now focus on $\G_1$.
		
		Fix $n$ distinct field elements $\{\alpha_{i,j}\}_{i,j=1}^m\subseteq\f$ and let $w,y,z$ be new variables. Define two vectors of polynomials of degree $n-1$,   $R=(R_1,\ldots,R_m),C=(C_1,\ldots,C_m)\in\f[y]^m$, such that for every $k\in[m]$ $R_k$ and $C_k$ satisfy
		\begin{align*}
			R_k(\alpha_{i,j})&=\delta_{i,k}\quad \text{and}\quad C_k(\alpha_{i,j})=\delta_{j,k}.
		\end{align*}
		Define $\G_1(w,y,z)$ as the $m\times m$ matrix $z\cdot (w^{2n-2}R(\frac{y}{w})\cdot C(\frac{y}{w})^T)$ (the $(i,j)$ entry of $\G_1$ is $z\cdot w^{2n-2}\cdot R_i(\frac{y}{w})\cdot C_j(\frac{y}{w})$). As every coordinate of $\G_1$ is a homogeneous polynomial of degree $2n-1$, $\G_1$ is a uniform polynomial map. For any $i,j\in[m]$ we have that
		$$\G_1(1,\alpha_{i,j},z)=z\cdot(R_{i'}(\alpha_{i,j})\cdot C_{j'}(\alpha_{i,j}))_{i',j'\in[m]}=z\cdot(\delta_{i,i'}\delta_{j,j'})_{i',j'\in[m]}\;.$$
		The above matrix has $z$ in entry $(i,j)$ and $0$ everywhere else, so $\G_1$ is a uniform \ropinvgen[1]{}. The resulting matrix $\bm{M}\circ \G_1$ is of rank $1$ since it is a product of vectors $R\cdot C^T$, so the variable-disjoint sum $\G=\sum_1^t \G_1(w_i,y_i,z_i)$ is a uniform \ropinvgen{} satisfying $f\circ \G=0$.
	\end{proof}

\section{Interpolation and reconstruction for orbits of the continuant polynomial}\label{sec:cont}

We start by proving that any uniform $1$-independent map hits $\ccont{\f}$ (\cref{thmhitcont}).

\begin{proof}[Proof of \cref{thmhitcont}]
	Let $f(x_1,\ldots,x_n)=\cont_{m}\bra{\ell_1(\xn)+b_1,\ldots,\ell_m(\xn)+b_m}$, where the $\ell_i$s are linear forms.
	Observe that $\cont_{m}(y_1,\ldots,y_m)$ is a multilinear polynomial that has a unique  monomial of degree $m$ and all other monomials are of smaller degree. Thus,  
	$$\cont_{m}(y_1,\ldots,y_m)= \prod_{i=1}^{m}y_i + \tilde{\cont}_{m-1}(y_1,\ldots,y_m)\;,$$
	where $\deg\bra{\tilde{\cont}_{m-1}}\leq m-1$.
	Hence,
	$$f(\xn)=\cont_{m}\bra{\ell_1(\xn)+b_1,\ldots,\ell_m(\xn)+b_m} = \prod_{i=1}^{m}\ell_i + \tilde{f}\bra{\ell_1,\ldots,\ell_m}\;,$$
	where $\deg\bra{\tilde{f}}\leq m-1$.
	
	Let $\G_1$ be a uniform $1$-independent polynomial map into $\f^n$. Let $d$ be the degree of the different components of $\G_1$. 
	\cref{obs:kwise}(\ref{item:coordsGenInd}) implies that  $\bra{\prod_{i=1}^{m}\ell_i} \circ \G_1\neq 0$ and hence it is a nonzero homogeneous polynomial of degree $m\cdot d$. As $\deg\bra{\tilde{f}\circ\G_1} \leq (m-1)\cdot d< \deg\bra{\bra{\prod_{i=1}^{m}\ell_i} \circ \G_1}$, we have that	
	$$f \circ \G_1 = \bra{\prod_{i=1}^{m}\ell_i} \circ \G_1 + \tilde{f}\circ\G_1\neq 0$$
	and the claim follows.
\end{proof}

\cref{cor:hit-cont} follows immediately from \cref{thmhitcont}, \cref{obsHitSetGen} and the construction of a uniform generator in \cref{defsvgenhom}.

\begin{remark}
	A similar argument would show that   $\G(y,z)\triangleq \bra{y^{n-1},y^{n-2}z,\ldots,z^{n-1}}$ is a hitting set generator for $\cont_{m}^{\gla{n}{\f}}$, which leads to a hitting set of size $n^4$. 
\end{remark}
We now turn to giving a reconstruction algorithm for $\ccont{\f}$. We start by proving some simple lemmas that will be used for constructing an interpolating set.

\begin{definition}\label{def:triplet}
		We call an ordered triplet $(i,j,k) \in \Z_m^3$ a \emph{consecutive triplet} 
		if $j=i+1$ and $k=i+2$, or $j=k+1$ and $i=k+2$, where all equalities are taken modulo $m$.
\end{definition}

\begin{lemma}\label{lem:triplet}
	Let $m\geq 3$. Then $(i,j,k)$ is a consecutive triplet if and only if every monomial in $\cont_{m}(x_0,\ldots,x_{m-1})$ that contains both $x_i$ and $x_k$, also contains $x_j$.
\end{lemma}

\begin{proof}
	Observe that a polynomial $f(\xn)$ has a monomial containing $x_i$ and $x_k$ but not $x_j$, if and only if this is also the case when we set $x_j=0$. Assume that $(i,j,k)$ is a consecutive triplet. Then,
	\begin{align*}
		\cont_{m}(x_0,\ldots,x_i,0,x_{i+2},\ldots,x_{m-1}) = \trace \bra{\begin{pmatrix}
				x_0 & 1 \\
				1 & 0
			\end{pmatrix}  
		 \cdot \ldots \cdot 
		 \begin{pmatrix}
		 	x_{i} & 1 \\
		 	1 & 0
		 \end{pmatrix}
		 \cdot 
		 \begin{pmatrix}
		 0 & 1 \\
 		1 & 0
		\end{pmatrix}
		\cdot\begin{pmatrix}
			x_{i+2} & 1 \\
			1 & 0
		\end{pmatrix}  
		\cdot \ldots \cdot \begin{pmatrix}
			x_{m-1} & 1 \\
			1 & 0
		\end{pmatrix}  
	 } \\
 = \trace \bra{\begin{pmatrix}
 			x_0 & 1 \\
 			1 & 0
 		\end{pmatrix}  
 		\cdot \ldots \cdot 
 		\begin{pmatrix}
 			x_{i-1} & 1 \\
 			1 & 0
 		\end{pmatrix}
 		\cdot  \begin{pmatrix}
 		x_i+x_{i+2} & 1 \\
 		1 & 0
 	\end{pmatrix}  
 	\cdot \begin{pmatrix}
 		x_{i+3} & 1 \\
 		1 & 0
 	\end{pmatrix}  
 	\cdot \ldots \cdot \begin{pmatrix}
 		x_{m-1} & 1 \\
 		1 & 0
 	\end{pmatrix}  
 }\;.
	\end{align*}
	It immediately follows that no monomial of $\cont_{m}(x_0,\ldots,x_i,0,x_{i+2},\ldots,x_{m-1})$ contains both $x_i$ and $x_{i+2}$.
	
	We now prove the second direction in the claim. Since $\cont_{m}$ is a trace of a matrix product, by properties of trace we can assume WLOG that  $i<j<k$, by first  rotating the order of the matrices until we have $i<j<k$ or $k<j<i$ (where $a<b$ means that the matrix corresponding to $a$ comes before that of $b$). As both cases are equivalent we can assume that $i<j<k$. We next handle this case.
	Assume WLOG that $j-i>1$. Set $x_r=0$  for every $i+2 \leq r < k$, to $0$.  We get that the new polynomial has the form 
		\begin{align*}
			& \trace \bra{\begin{pmatrix}
					x_{0} & 1 \\
					1 & 0
				\end{pmatrix}
				\cdot \ldots \cdot \begin{pmatrix}
					x_{i} & 1 \\
					1 & 0
				\end{pmatrix}
				\cdot \begin{pmatrix}
					x_{i+1} & 1 \\
					1 & 0
				\end{pmatrix}
			\cdot \begin{pmatrix}
					0&1\\1&0
					\end{pmatrix}^{k-i-2}
				\cdot  \begin{pmatrix}
					x_{k} & 1 \\
					1 & 0
				\end{pmatrix}
				\cdot \ldots \begin{pmatrix}
					x_{m-1} & 1 \\
					1 & 0
				\end{pmatrix}
				 }\\ &=\left\{\begin{array}{lr}
				 \cont_{m-k+i+2}\bra{x_0,\ldots,x_i,x_{i+1},x_k,\ldots,x_{m-1}}, & \text{for }k-i \text{ even}\\
				  \cont_{m-k+i+1}\bra{x_0,\ldots x_{i-1},x_i,x_{i+1}+x_k,\ldots,x_{m-1}}, & \text{for }k-i \text{ odd}
			 \end{array}\right. \;,  
		\end{align*} 
	and a monomial of maximal degree in  this polynomial contains both $x_i$ and $x_k$ (when $k-i$ is even there is a unique monomial of maximal degree, and when $k-i$ is odd there are two such monomials). 
	\end{proof}

\begin{corollary}\label{cor:triplet}
	Let $m\geq 3$. Then $(i,j,k)$ is a consecutive triplet if and only if 
	$\restr{\frac{\partial^2 \cont_{m}}{\partial x_i \partial x_k}}{x_j=0}=0$.
\end{corollary}

For every list of three distinct indices $(i,j,k)\in [m]_0^3$ denote  $$\cont_{m}^{(i,j,k)}(\xn) \triangleq \restr{\frac{\partial^2 \cont_{m}}{\partial x_i \partial x_k}}{x_j=0}\;.$$

\begin{lemma}\label{lem:cont-hit-deriv}
	Let $n\geq m\geq 3$ and $t$ be integers. Assume  $H(\bw):\f^t\to \f^n$ is a hitting-set generator for ${\cont_{m}^{(i,j,k)}}^{\gla{n}{\f}}$, for every list of three distinct indices $(i,j,k)\in  [m]_0^3$.
	Let $\G_3(\by,\bz)$ be a $3$-independent polynomial map (into $\f^n$) that each of its coordinates is a homogeneous linear function in $\bz$, over $\f(\by)$ (for example, $\svgen_k$ has this property, for every $k$). 
	Then, for every $m_1,m_2$ and $ n$ and every two polynomials $f_1 \in \cont_{m_1}^{\gla{n}{\f}}$ and $f_2 \in \cont_{m_2}^{\gla{n}{\f}}$   it holds that $f_1=f_2$ if and only if $f_1\circ(H+\G_3)=f_2\circ(H+\G_3)$.
\end{lemma}

Roughly, what the lemma claims is that if $\G_3$ is a $3$-independent map and $H$ hits ${\cont_{m}^{(i,j,k)}}^{\gla{n}{\f}}$, then $H+\G_3$ is
an interpolating-set generator.

\begin{proof}
	Denote $f_1 = \cont_{m_1}\bra{\ell_{1,0},\ldots,\ell_{1,m_1-1}}$ and $f_2 = \cont_{m_2}\bra{\ell_{2,0},\ldots,\ell_{2,m_2-1}}$.
	The proof has three steps. We first prove that if $f_1\circ(H+\G_3)=f_2\circ(H+\G_3)$ then $m_1=m_2$ and there exists a permutation $\pi:[m]_0\to [m]_0$, and constants 
	$\alpha_j$, such that for every $j$ it holds that $\ell_{1,j}= \alpha_{j} \cdot \ell_{2,\pi(j)}$. We then show that, possibly after rotating the order and taking a transpose, we can assume WLOG that $\pi$ is the identity permutation. 
	At the last step we prove that either $\alpha_j=1$ for every $j$, or that $m$ is even, $\alpha_0\cdot \alpha_1=1$ and for every $j$, $\alpha_{2j}=\alpha_0$ and $\alpha_{2j+1}=\alpha_1$.
	
	{\bf Step 1:} As in the proof of \cref{thmhitcont}, $\deg(f_i)=m_i$ and the homogeneous part of degree $m_i$ in $f_i$ is given by $$f_i^{[m_i]}=\prod_{j=0}^{m_i-1}\ell_{i,j}^{[1]}\;.$$ Observe that since $f_i^{[m_i]}\circ \bra{H+\G_3}$ is nonzero (e.g. by \cref{obs:kwise}(\ref{item:coordsGenInd})), and its degree, as a polynomial in $\bz$, is exactly $m_i$ (and every other term in $f_i\circ(H+\G_3)$ has degree strictly smaller as a polynomial in $\bz$), 
	it must hold that $m_1=m_2$. To simplify the notation let $m=m_1=m_2$. Again by comparing terms of maximal degree in $\bz$ we see that 
	\begin{equation}\label{eq:high-term}
		\bra{\prod_{j=0}^{m-1}\ell_{1,j}^{[1]}} \circ \G_3 = \bra{\prod_{j=0}^{m-1}\ell_{2,j}^{[1]}} \circ \G_3\;.
	\end{equation} 
	As both $\cbra{\ell_{1,i}}$ and $\cbra{\ell_{2,i}}$ are linearly independent sets, we get from unique factorization and from
	\cref{obs:kwise}(\ref{item:ind-lin}), that there exists a permutation $\pi:[m]_0\to [m]_0$ and constants $\cbra{\alpha_j}$ so that $\ell_{1,j}=\alpha_j \ell_{2,\pi(j)}$, for every $j$.  This completes the first step.
	
	{\bf Step 2:} We wish to show that the permutation $\pi$ is an ``ordered'' cycle of length $m$. That is, that  it either has the form $(i,i+1,\ldots ,m-1,0,\ldots,i-1)$, or $(i,i-1,\ldots,0,m-1,\ldots,i+1)$, for some $i$. Indeed, assume for a contradiction that this is not the case. Then, there must be an index $i$ such that $(\pi(i),\pi(i+1),\pi(i+2))$ is not a consecutive triplet. Let $\cbra{{\vv_j}}_j$ be a dual set to $\cbra{\ell_{2,j}}_j$. \cref{cor:triplet} and \cref{lem:dualder} imply that 
	$$ \restr{\frac{\partial^2 f_1}{\partial \vv_i\partial \vv_{i+2}}}{\ell_{2,{i+1}}(\xn)=0} = 0 \quad \text{and } \quad \restr{\frac{\partial^2 f_2}{\partial \vv_i\partial \vv_{i+2}}}{\ell_{2,{i+1}}(\xn)=0} \neq 0 \;.$$ In particular
	$$ -\cont_{m}^{(\pi(i),\pi(i+1),\pi(i+2))}\bra{\ell_{2,0},\ldots,\ell_{2,m-1}}=\restr{\frac{\partial^2 (f_1-f_2)}{\partial \vv_i\partial \vv_{i+2}}}{\ell_{2,{i+1}}(\xn)=0} \neq 0 \;.$$ 
	By the assumption on $H$ we get that $$-\cont_{m}^{(\pi(i),\pi(i+1),\pi(i+2))}\circ H= \bra{\restr{\frac{\partial^2 (f_1-f_2)}{\partial \vv_i\partial \vv_{i+2}}}{\ell_{2,{i+1}}(\xn)=0}} \circ H=-\bra{\restr{\frac{\partial^2 f_2}{\partial \vv_i\partial \vv_{i+2}}}{\ell_{2,{i+1}}(\xn)=0}} \circ H\neq 0 \;.$$ 
	Applying \cref{lem:indDerivLinear} for $k=2$ and \cref{lem:indProjectZero} for $k=1$ we get that  $(f_1-f_2)\circ(H+\G_3)\neq 0$, in contradiction.
	
	{\bf Step 3:}  To simplify notation, assume, WLOG, that $\pi$ is the identity permutation. Observe that  $\prod_{i=0}^{m-1}\ell_{1,i}^{[1]} \circ \G_3 = \prod_{i=0}^{m-1}\alpha_i \cdot  \prod_{i=0}^{m-1}\ell_{2,i}^{[1]} \circ \G_3$. Hence, Equation~\eqref{eq:high-term} implies that $\prod_{i=0}^{m-1}\alpha_i=1$. If there is $i$ such that $\alpha_i \cdot \alpha_{i+1}\neq 1$ then use $\G_3$ to restrict to the subspace $\ell_{1,i}=\ell_{1,i+1}=0$ (using \cref{lem:indProjectZero}). 
	Denote with $\G_3'$, the map $\G_3$  after we used two of the $z_i$s for the restriction ($\G_3'$ is a $1$-independent map).
	As $\cont_m(x_0,\ldots,x_{i-1},0,0,x_{i+2},\ldots,x_{m-1}) = \cont_{m-2}(x_0,\ldots,x_{i-1},x_{i+2},\ldots,x_{m-1})$, we get a contradiction by considering the terms of maximal degrees (as polynomials in the remaining $z$) in $f_1 \circ \G_3$ and  $f_2 \circ \G_3$  as follows:
	\begin{align*}
	\bra{\prod_{j\in [m]_0\setminus \cbra{i,i+1}}\ell_{1,j}^{[1]}} \circ \G_3' &= \bra{\prod_{j\in [m]_0\setminus \cbra{i,i+1}}\ell_{2,j}^{[1]}} \circ \G_3' = \bra{ \prod_{j\in [m]_0\setminus \cbra{i,i+1}} \alpha_j \cdot\ell_{1,j}^{[1]} }\circ \G_3' \\
	&= \bra{ \prod_{j\in [m]_0\setminus \cbra{i,i+1}} \alpha_j} \cdot \bra{ \prod_{j\in [m]_0\setminus \cbra{i,i+1}}  \cdot\ell_{1,j}^{[1]} }\circ \G_3' \\ &= \frac{1}{\alpha_i \cdot \alpha_{i+1}} \cdot \bra{ \prod_{j\in [m]_0\setminus \cbra{i,i+1}}  \cdot\ell_{1,j}^{[1]} }\circ \G_3' \neq  \bra{\prod_{j\in [m]_0\setminus \cbra{i,i+1}}{\ell^{[1]}_{1,j}}}\circ \G_3'\;,
	\end{align*}
	where the first equality follows from the assumption that $f_1 \circ \G_3 = f_2 \circ \G_3$ and the last inequality uses the assumption  $\alpha_i \cdot \alpha_{i+1}\neq 1$. Consequently, either for every  $i$, $\alpha_i=1$, which means that $f_1=f_2$, as we wanted to prove, or  $m$ is even and  for every $i$, $\alpha_{2,i}=\alpha_0$ and $\alpha_{2,i+1}=\alpha_1$, and that $\alpha_0\cdot \alpha_1=1$. We next show that in this case as well the polynomials are equal. Indeed, observe that
	$\begin{pmatrix}
		1 & 0 \\
		0 & \alpha_0
	\end{pmatrix} \cdot \begin{pmatrix}
	1 & 0 \\
	0 & \alpha_1
\end{pmatrix}  = \begin{pmatrix}
1 & 0 \\
0 & 1
\end{pmatrix} $. Hence, 
	\begin{align}
		\nonumber
		f_1 &= \trace \bra{\begin{pmatrix}
				\ell_{1,0} & 1 \\
				1 & 0
			\end{pmatrix}  
			\cdot  
			\begin{pmatrix}
				\ell_{1,1} & 1 \\
				1 & 0
			\end{pmatrix}
			\cdot\ldots 			\cdot\begin{pmatrix}
				\ell_{1,m-1} & 1 \\
				1 & 0
			\end{pmatrix}  
		} \\ \nonumber
		&=\trace \left( \begin{pmatrix}
				1 & 0 \\
				0 & \alpha_0
			\end{pmatrix} \cdot\begin{pmatrix}
				\ell_{1,0} & 1 \\
				1 & 0
			\end{pmatrix}  
			\cdot  
			\begin{pmatrix}
				1 & 0 \\
				0 & \alpha_0
			\end{pmatrix} \cdot \begin{pmatrix}
				1 & 0 \\
				0 & \alpha_1
			\end{pmatrix} 
		\cdot 
			\begin{pmatrix}
				\ell_{1,1} & 1 \\
				1 & 0
			\end{pmatrix}
			\cdot
			\begin{pmatrix}
				1 & 0 \\
				0 & \alpha_1
			\end{pmatrix} \cdot \begin{pmatrix}
				1 & 0 \\
				0 & \alpha_0
			\end{pmatrix}  \right. \\ \nonumber
			&\qquad\qquad\quad\cdot \left. \begin{pmatrix}
				\ell_{1,2} & 1 \\
				1 & 0
			\end{pmatrix}  
			\cdot  
			\begin{pmatrix}
				1 & 0 \\
				0 & \alpha_0
			\end{pmatrix} \cdot \begin{pmatrix}
				1 & 0 \\
				0 & \alpha_1
			\end{pmatrix} 
			\cdot 
			\ldots 			\cdot \begin{pmatrix}
				1 & 0 \\
				0 & \alpha_0
			\end{pmatrix} \cdot \begin{pmatrix}
				1 & 0 \\
				0 & \alpha_1
			\end{pmatrix} \cdot\begin{pmatrix}
				\ell_{1,m-1} & 1 \\
				1 & 0
			\end{pmatrix}  \cdot \begin{pmatrix}
		1 & 0 \\
		0 & \alpha_1
	\end{pmatrix} 
		\right)\\ \label{eq:rescale}
	&= \trace \bra{\begin{pmatrix}
			\ell_{1,0} & \alpha_0 \\
			\alpha_0 & 0
		\end{pmatrix}  
		\cdot  
		\begin{pmatrix}
			\ell_{1,1} & \alpha_1 \\
			\alpha_1 & 0
		\end{pmatrix}
		\cdot\ldots 			\cdot\begin{pmatrix}
			\ell_{1,m-1} & \alpha_1 \\
			\alpha_1 & 0
		\end{pmatrix}  
	} \\ \nonumber
&= \trace \bra{\begin{pmatrix}
	\alpha_0\cdot\ell_{2,0} & \alpha_0 \\
	\alpha_0 & 0
\end{pmatrix}  
\cdot  
\begin{pmatrix}
	\alpha_1 \cdot\ell_{2,1} & \alpha_1 \\
	\alpha_1 & 0
\end{pmatrix}
\cdot\ldots 			\cdot\begin{pmatrix}
	\alpha_1\cdot \ell_{2,m-1} & \alpha_1 \\
	\alpha_1 & 0
\end{pmatrix}}\\ \nonumber
 &= \bra{\alpha_0 \cdot \alpha_1}^{m/2}\cdot \trace \bra{\begin{pmatrix}
		\ell_{2,0} & 1 \\
		1 & 0
	\end{pmatrix}  
	\cdot  
	\begin{pmatrix}
		\ell_{2,1} & 1 \\
		1 & 0
	\end{pmatrix}
	\cdot\ldots 			\cdot\begin{pmatrix}
		\ell_{2,m-1} & 1 \\
		1 & 0
	\end{pmatrix}  
} = 1\cdot f_2 \;.
	\end{align}

 This concludes the proof of the lemma.
\end{proof}

From \cref{lem:cont-hit-deriv} we see that all that we have to do in order to construct an interpolating set for $\ccont{\f}$, is to find a map $H$ as in the statement of the lemma.

\begin{lemma}\label{lem:2hitder}
	Let $n\geq m$ be integers. Let $\G_2(\by,\bz)$ be a $2$-independent polynomial map into $\f^n$, that is linear in $\bz$. Then, For every list of three distinct indices $(i,j,k)\in [m]_0^3$ and for every $m$ $n$-variate linearly independent linear functions $\ell_0(\xn),\ldots,\ell_{m-1}(\xn)\in \f[\xn]$  it holds that if $
	\cont_{m}^{(i,j,k)}\bra{\ell_0,\ldots,\ell_{m-1}}\neq 0$ then $\cont_{m}^{(i,j,k)}\bra{\ell_0,\ldots,\ell_{m-1}} \circ \G_2 \neq 0$. 
\end{lemma}

\begin{proof}
	As $\cont_{m}^{(i,j,k)}\bra{\ell_0,\ldots,\ell_{m-1}}\neq 0$ it follows that $(i,j,k)$ is not a consecutive triplet. Assume WLOG that $i<j-1<j<k$. Use $\G_2$ to further restrict the polynomial to the subspace $\ell_{j-1}=0$ (using \cref{lem:indProjectZero}). Let $\G_2'$ denote $\G_2$ after the restriction. \cref{lem:indProjectZero} guarantees that $\G_2'$  is $1$-independent. Observe that the homogeneous term of maximal degree in $\restr{\cont_{m}^{(i,j,k)}\bra{\ell_0,\ldots,\ell_{m-1}}}{\ell_{j-1}(\xn)=0}$ is equal to $\prod_{t\in [m]_0\setminus \cbra{i,j-1,j,k}}\ell_{t}^{[1]}$. It follows that the term of maximal degree, as a polynomial in $\bz$, in
$\restr{\cont_{m}^{(i,j,k)}\bra{\ell_0,\ldots,\ell_{m-1}}}{\ell_{j-1}(\xn)=0}\circ \G_2'$ is  	
	 	$\bra{\prod_{t\in [m]_0\setminus \cbra{i,j-1,j,k}} \ell_{t}^{[1]}}\circ \G_2'$, which is nonzero by \cref{obs:kwise}(\ref{item:coordsGenInd}).  
\end{proof}

Combining \cref{lem:cont-hit-deriv,lem:2hitder} we get the following corollary:

\begin{corollary}\label{cor:5intscont}
	Let $\G_5(\by,\bz):\f^t\to\f^n$ be a $5$-independent polynomial map  that is linear in $\bz$. Then, for every $m_1,m_2 \leq n$ and every two polynomials $f_1 \in \cont_{m_1}^{\gla{n}{\f}}$ and $f_2 \in \cont_{m_2}^{\gla{n}{\f}}$,   it holds that $f_1=f_2$ if and only if $f_1\circ \G_5 =f_2\circ\G_5$.	
\end{corollary}

\cref{thm:cont-int} follows immediately from \cref{cor:5intscont} and \cref{obsHitSetGen}.

\subsection{Reconstruction algorithm for $\ccont{\f}$}\label{sec:cont-recon}

The reconstruction algorithm is given in Page~\pageref{alg:ccont}.

\begin{algorithm}
	\SetKwData{Left}{left}\SetKwData{This}{this}\SetKwData{Up}{up}	\SetKwFunction{Union}{Union}\SetKwFunction{FindCompress}{FindCompress}
	\SetKwInOut{Input}{input}\SetKwInOut{Output}{output}
	\Input{Integer $n$, black-box access to $f= \cont_m^{\gla{n}{\f}}$}
	\Output{Linear functions $\tilde\ell_0,\ldots,\tilde\ell_{m-1}\in\f[\xn]$ such that $f= \cont_m\bra{\tilde\ell_0,\ldots,\tilde\ell_{m-1}}$}
	\BlankLine
	Compute $m$ using interpolation and the hitting set constructed in \cref{thmhitcont} \label{alg:degstep}\;
	Factor $f^{[m]}$ \label{alg:factorstep} \tcc*{Using univariate root-finding}
	\tcc{We found linear functions $L_0^{[1]},\ldots,L_{m-1}^{[1]}$, such that for some permutation $\pi$ and scalars $\alpha_i$, $\alpha_i\cdot L_i^{[1]} =  \ell_{\pi(i)}^{[1]}$}
	Compute a dual set $\cbra{\vv_i}_i$ to  $\cbra{L_i^{[1]}}_i$\;
	\tcc{Next we compute the free terms}
	\For{$i=0$ \KwTo $m-1$ }{\label{alg:freetermloop}
		Define $f_i'(\xn) \triangleq \frac{\partial f}{\partial {\vv_i}}(\xn)$\;
		Set $g_i(\xn) = f(\xn) - L_i^{[1]}(\xn) \cdot f'_i(\xn)$ \tcc*{We can simulate queries to $g_i$  }
		Compute $\deg\bra{g_i}$\;
		\eIf{$\deg\bra{g_i}=m-2$}{
		set $\lambda_i=0$}{
		Find $\vu\in\f^n$ such that $f^{[m]}(\vu) \neq 0$\;
		Set $\lambda_i = \bra{L_i^{[1]}(\vu)\cdot{g_i}^{[m-1]}(\vu)}/f^{[m]}(\vu)$\;
		}
	Set $L_i = L_i^{[1]} + \lambda_i$\; 
	}
	\tcc{There is a permutation $\pi$ and scalars $\alpha_i$ such that $\alpha_i\cdot L_i =  \ell_{\pi(i)}$}
	Find all consecutive triplets and recover the permutation $\pi$ \label{alg:findtriplet} \;
	\tcc{WLOG $\pi$ is the identity permutation}
	Find $\cbra{\vu_i}$ such that $L_i(\vu_j)=\delta_{i,j}$ \label{alg:2nddualstep}\;
	\tcc{We now recover the $\alpha_i$s}
	\eIf{$m$ is odd}{\label{alg:findalphastep}
		\For{$i=0$ \KwTo $m-1$}{
			Set $\tilde{\ell}_i = f(\vu_i)\cdot L_i$\;
			}
		}
	{Set $\beta_0=\alpha_0=1$ and $\tilde{\ell}_0=L_0$\;
		\For{$i=1$ \KwTo $m-1$}{
			Set $\beta_i = \bra{f(\vu_{i-1}+\vu_i)-2}/\beta_{i-1}$ and $\tilde{\ell}_i = \beta_i \cdot L_i$\;
		}\label{alg:endstep}
	}
	\Return{$\tilde{\ell}_0,\ldots,\tilde{\ell}_{m-1}$}\;
\caption{reconstruction algorithm for $\ccont{\f}$}\label{alg:ccont}
\end{algorithm}

\paragraph{Analysis of Algorithm~\ref{alg:ccont}:} 

\begin{claim}
	Step~\ref{alg:degstep} can be executed in polynomial-time.
\end{claim}

\begin{proof}
	Let $\G_1(y,z)$ be a $1$-independent map. Let $w$ be a new variable and consider $\G=w\cdot \G_1$. I.e., we multiply each coordinate of $\G_1$ with $w$. Observe that the degree of $w$ and of $z$ in $\bra{f\circ \G}$ is exactly $\deg(f)=m$.
	As in the proof of \cref{thmhitcont}, we see that the $m$-homogeneous component of $\bra{f\circ \G}$, when viewed as a polynomial in $w$, is $\bra{\prod_{i=0}^{m-1}\ell_i^{[1]}}\circ \G_1\neq 0$. As we know that $m\leq n$, using interpolation (over $w$) we get black-box access to $\bra{f\circ \G}^{[k]}$, for every $0\leq k\leq n$. We look for the first $k$, starting from $n$ and going down, such that $\bra{f\circ \G}^{[k]}\neq 0$. This can be done, for example, by interpolation (over $y,z$). 
\end{proof}

\begin{claim}
	Step~\ref{alg:factorstep} can be done with polynomially many queries to 
	a root-finding algorithm over $\f$ (assuming $|\f|\geq n^3$).
\end{claim}
We assume some knowledge with known factoring algorithms. For good a reference see \cite{books/daglib/0010544} (the lecture notes of Madhu Sudan are also a great resource on the subject \cite{Sudan-alg-notes}).
\begin{proof}[Proof sketch]
Observe that $f^{[m]}= \prod_{i=0}^{m-1}\ell_i^{[1]}$, and all its linear factors are linearly independent. Known factoring algorithms require that we reduce the polynomial that we wish to factor to a square-free, bivariate polynomial. This can be easily done using $2$-independent maps. Let $\G_2(\by,z_1,z_2)$ be a  $2$-independent map that is a linear form in $z_1$ and $z_2$ (e.g., $\svgen_2$). \cref{obs:kwise}(\ref{item:ind-lin}) shows that composing $f^{[m]}$ with  $\G_2(\by,\bz)$, keeps all factors linearly independent, when viewed as linear polynomials in $\bz$.  Each assignment to $\by$ gives a different polynomial whose factors are homogeneous linear functions in $z_1,z_2$. Observe that there is an assignment to $\by$ from the set $[n^3]^{|\by|}$, that maintains the property that the factors are linearly independent. Indeed, for every two factors we need the assignment to be a nonzero of the determinant of the coefficient-matrix of the two factors. There are ${m \choose 2}$ such determinant, each has degree $2(n-1)$ as a polynomial in $\by$ (hence the requirement for a field of size $n^3$). By going over all such assignments to $\by$, we are guaranteed to find one that maintains this property.

 Once we reduced to the square-free, bivariate case, factoring algorithms proceed by reducing to factoring of univariate polynomials. In our case the univariate completely splits as a product of linear factors, hence the univariate factorization step only need oracle access to a root-finding algorithm. 
\end{proof}

Observe that  we have found irreducible linear functions $L_i^{[1]}$, each is a scalar product of some $\ell_{\pi(i)}^{[1]}$, for some permutation $\pi$. Let $\cbra{\alpha_i}$ be such that $\alpha_i\cdot L_i^{[1]} = \ell_{\pi(i)}^{[1]}$. 

\begin{claim}
	For every $i$, the for-loop in Step~\ref{alg:freetermloop} returns $L_i$ such that $\alpha_i\cdot L_i = \ell_{\pi(i)}$.
\end{claim}

\begin{proof}
	For $i\in [m]_0$, denote $\cont_{m}(y_0,\ldots,y_{m-1}) = y_{\pi(i)}\cdot F_{i,1}(\by\setminus y_{\pi(i)})+F_{i,0}(\by\setminus y_{\pi(i)})$.
	Observe that $\deg(F_{i,1})=m-1$ (since it contains the product of all $y_j$ except $y_{\pi(i)}$) and that $\deg\bra{F_{i,0}}=m-2$. Indeed, 
	\begin{align*}
	F_{i,0}(\by) & = \restr{\cont_{m}(y_0,\ldots,y_{m-1})}{y_{\pi(i)}=0} \\
	&= \trace \bra{\begin{pmatrix}
			y_0 & 1 \\
			1 & 0
		\end{pmatrix}  
		\cdot \ldots \cdot  
		\begin{pmatrix}
			y_{{\pi(i)}-1} & 1 \\
			1 & 0
		\end{pmatrix}
		\cdot  
		\begin{pmatrix}
			0 & 1 \\
			1 & 0
		\end{pmatrix}\cdot  
		\begin{pmatrix}
			y_{{\pi(i)}+1} & 1 \\
			1 & 0
		\end{pmatrix}
		\cdot\ldots 			\cdot\begin{pmatrix}
			y_{m-1} & 1 \\
			1 & 0
		\end{pmatrix}  
	}\\ 
	&=  \trace \bra{\begin{pmatrix}
		y_0 & 1 \\
		1 & 0
	\end{pmatrix}  
	\cdot \ldots \cdot  
	\begin{pmatrix}
		y_{{\pi(i)}-2} & 1 \\
		1 & 0
	\end{pmatrix}
	\cdot    
	\begin{pmatrix}
		y_{{\pi(i)}-1}+y_{{\pi(i)}+1} & 1 \\
		1 & 0
	\end{pmatrix}\cdot  
\begin{pmatrix}
y_{{\pi(i)}+2} & 1 \\
1 & 0
\end{pmatrix}
	\cdot\ldots 			\cdot\begin{pmatrix}
		y_{m-1} & 1 \\
		1 & 0
	\end{pmatrix}  
}\\ &= \cont_{m-2}\bra{y_0,\ldots,y_{{\pi(i)}-2},y_{{\pi(i)}-1}+y_{{\pi(i)}+1},y_{{\pi(i)}+2},\ldots,y_{m-1}}\;.	\end{align*}
		We now note that 
	$$f'_i(\ell_0,\ldots,\ell_{m-1}) = \frac{\partial \ell_{\pi(i)}}{\partial {\vv_i}} \cdot F_{i,1}({\bm\ell}\setminus \ell_{\pi(i)}) = \alpha_i \cdot F_{i,1}({\bm\ell}\setminus \ell_{\pi(i)}) \;.$$ As $g_i= f - L_i^{[1]} \cdot f'_i$, we get that
	\begin{align*}
		g_i &= \bra{\ell_{\pi(i)} \cdot F_{i,1}({\bm\ell}\setminus \ell_{\pi(i)})+F_{i,0}({\bm\ell}\setminus \ell_{\pi(i)})} - L_i^{[1]} \cdot \bra{ \alpha_i \cdot F_{i,1}({\bm\ell}\setminus \ell_{\pi(i)})} \\ &=  \bra{\ell_{\pi(i)} - \alpha_i\cdot L_i^{[1]}}\cdot F_{i,1}({\bm\ell}\setminus \ell_{\pi(i)})+F_{i,0}({\bm\ell}\setminus \ell_{\pi(i)})\;.
	\end{align*}
	Thus, $\deg(g_i)=m-2$ if and only if $\ell_{\pi(i)} - \alpha_i\cdot L_i^{[1]}=0$. In other words,  $\deg(g_i)=m-2$ if and only if $\ell_{\pi(i)}$ is homogeneous and $L_i=L_i^{[1]}$.
	As $\ell_{\pi(i)}^{[1]}=\alpha_i \cdot L_i^{[1]}$, it holds that $\bra{\ell_{\pi(i)} - \alpha_i\cdot L_i^{[1]}}\in \f$. Therefore, if $\deg(g_i)=m-1$ we get that 
	\begin{equation}\label{eq:gim-1}
		g_i^{[m-1]}=\bra{\ell_{\pi(i)} - \alpha_i\cdot L_i^{[1]}}\cdot F_{i,1}({\bm\ell}^{[1]}\setminus \ell_i^{[1]})^{[m-1]} = \bra{\ell_{\pi(i)} - \alpha_i \cdot L_i^{[1]}}\cdot \prod_{j\neq {\pi(i)}}\ell_j^{[1]} \;.
	\end{equation}
	Hence,  
	\begin{align*}
		\lambda_i &= L_i^{[1]}(\vu)\cdot g_i^{[m-1]}(\vu)/f^{[m]}(\vu) =L_i^{[1]}(\vu)\cdot \bra{\ell_{\pi(i)} - \alpha_i \cdot L_i^{[1]}}\cdot \prod_{j\neq {\pi(i)}}\ell_j^{[1]}/\prod_{j}\ell_j^{[1]}\\& = \bra{L_i^{[1]}(\vu) \cdot \bra{\ell_{\pi(i)} - \alpha_i \cdot L_i^{[1]}}}/\ell_{\pi(i)}^{[1]}(\vu) = \bra{\ell_{\pi(i)} - \alpha_i \cdot L_i^{[1]}}/\alpha_i \;.
	\end{align*}	
	It follows that 
	$$\alpha_i \cdot L_i = \alpha_i \cdot \bra{L_i^{[1]} + \lambda_i} = \alpha_i \cdot L_i^{[1]} + \alpha_i \cdot \lambda_i = \alpha_i \cdot L_i^{[1]}+ \bra{\ell_{\pi(i)} - \alpha_i \cdot L_i^{[1]}} = \ell_{\pi(i)}$$
	as claimed. 
	
	An important point to notice is that we can check whether  $\deg(g_i)=m-1$ in the same manner in which we computed $\deg(f)$ (thanks to Equation~\eqref{eq:gim-1}).
\end{proof}

Note that Step~\ref{alg:findtriplet} can be executed using \cref{cor:triplet} and \cref{lem:2hitder}. Indeed, as  $\ell_{\pi(i)}=\alpha_i L_i$, it follows that $\cbra{\vv_{i}/\alpha_i}$ is a dual set for  $\cbra{\ell_{\pi(i)}^{[1]}}$. That is, $\ell_{\pi(i)}^{[1]}\bra{\vv_{j}/\alpha_j}=\delta_{i,j}$. Therefore, $\restr{\frac{\partial^2 f}{\partial \bra{\vv_{i}/\alpha_i} \partial \bra{\vv_{k}/\alpha_k}}}{\ell_{\pi(j)}=0}=0$ if and only if $\restr{\frac{\partial^2 f}{\partial \vv_{i} \partial {\vv_{k}}}}{L_{j}=0}=0$.
Hence, with the help of  \cref{lem:2hitder} and interpolation, we can find all consecutive triplets. Once we have that information, construction of 
$\pi$ (up to reversal, which does not change the resulting polynomial) 
is immediate.
Since we know $\pi$ we can assume WLOG that $\pi$ is the identity permutation.

Step~\ref{alg:2nddualstep} is possible as the $L_i$s are linearly independent. Note that  $\ell_{\pi(i)}^{[1]}\bra{\vu_{j}}=\delta_{i,j}\cdot \alpha_j$.

\begin{claim}
	The linear functions $\tilde{\ell}_i$ that were computed in Steps~\ref{alg:findalphastep}-\ref{alg:endstep}  satisfy  $\cont_m\bra{\tilde{\ell}_0,\ldots,\tilde{\ell}_{m-1}}=f$.
\end{claim}

\begin{proof}
	First, observe that $\ell_i(\vu_j)=\alpha_i \cdot \delta_{i,j}$.
	Assume first that $m$ is odd. Then
	\begin{align*}
	f(\vu_i) &= \trace \bra{\begin{pmatrix}
			\ell_{0}(\vu_i) & 1 \\
			1 & 0
		\end{pmatrix}  
		\cdot  
		\begin{pmatrix}
			\ell_{1}(\vu_i) & 1 \\
			1 & 0
		\end{pmatrix}
		\cdot\ldots 			\cdot\begin{pmatrix}
			\ell_{m-1}(\vu_i) & 1 \\
			1 & 0
		\end{pmatrix}  
	} \\ &=\trace\bra{\begin{pmatrix}
		0 & 1 \\
		1 & 0
\end{pmatrix}^i \cdot \begin{pmatrix}
\alpha_i & 1 \\
1 & 0
\end{pmatrix} \cdot \begin{pmatrix}
0 & 1 \\
1 & 0
\end{pmatrix}^{m-i-1} } = \\
&= \trace{\begin{pmatrix}
		\alpha_i & 1 \\
		1 & 0
\end{pmatrix}} = \alpha_i	\;.	
\end{align*}
In this case we get that $\tilde{\ell}_i = 	f(\vu_i)\cdot L_i = \alpha_i L_i = \ell_i$. In particular, we recovered the original $\ell_i$s.

Next, assume that $m$ is even. Observe that since $m$ is even we can replace each $\ell_{2i}$ with $\ell_{2i}/\alpha_0$ and each $\ell_{2i+1}$ with $\ell_{2i+1}\cdot \alpha_0$ and still get the same $f$ (recall Equation~\eqref{eq:rescale}).
Therefore, we may  assume WLOG that $\alpha_0=1$.


The first iteration gives
	\begin{align*}
	f(\vu_{0}+\vu_1) &= \trace \bra{\begin{pmatrix}
			\ell_{0}(\vu_{0}+\vu_1) & 1 \\
			1 & 0
		\end{pmatrix}  
		\cdot  
		\begin{pmatrix}
			\ell_{1}(\vu_{0}+\vu_1) & 1 \\
			1 & 0
		\end{pmatrix}
		\cdot\ldots 			\cdot\begin{pmatrix}
			\ell_{m-1}(\vu_{0}+\vu_1) & 1 \\
			1 & 0
		\end{pmatrix}  
	} \\ &=\trace\bra{\begin{pmatrix}
			1 & 1 \\
			1 & 0
		\end{pmatrix} \cdot \begin{pmatrix}
			\alpha_1 & 1 \\
			1 & 0
		\end{pmatrix} \cdot \begin{pmatrix}
			0 & 1 \\
			1 & 0
		\end{pmatrix}^{m-2} } = \trace{\begin{pmatrix}
			\alpha_1+1 & 1 \\
			\alpha_1 & 1
	\end{pmatrix}} = \alpha_1+2	\;.	
\end{align*}
Hence, $\beta_1=\bra{f(\vu_{0}+\vu_1)-2}/\alpha_0 = \alpha_1/1=\alpha_1$, and therefore, $\tilde{\ell}_1=\ell_1$.
We proceed to show by induction that for every $i$, $\beta_i=\alpha_i$.
	\begin{align*}
	f(\vu_{i}+\vu_{i+1}) &= \trace \bra{\begin{pmatrix}
			0 & 1 \\
			1 & 0
		\end{pmatrix}  
		\cdot  \ldots \cdot 
		\begin{pmatrix}
			\ell_{i}(\vu_{i}+\vu_{i+1}) & 1 \\
			1 & 0
		\end{pmatrix}\cdot 
	\begin{pmatrix}
	\ell_{i+1}(\vu_{i}+\vu_{i+1}) & 1 \\
	1 & 0
\end{pmatrix}
		\cdot\ldots 			\cdot\begin{pmatrix}
			0 & 1 \\
			1 & 0
		\end{pmatrix}  
	} \\ &=\trace\bra{\begin{pmatrix}
			0 & 1 \\
			1 & 0
		\end{pmatrix} ^i \cdot \begin{pmatrix}
			\alpha_i & 1 \\
			1 & 0
		\end{pmatrix} \cdot \begin{pmatrix}
			\alpha_{i+1} & 1 \\
			1 & 0
		\end{pmatrix}\cdot \begin{pmatrix}
		0 & 1 \\
		1 & 0
	\end{pmatrix}^{m-i-2} }\\& = \trace{\begin{pmatrix}
			\alpha_i\cdot \alpha_{i+1}+1 & \alpha_i \\
			\alpha_{i+1} & 1
	\end{pmatrix}} = \alpha_i\cdot \alpha_{i+1}+2	\;,	
\end{align*}
and we conclude, from the induction hypothesis, that $\beta_{i+1}=\alpha_{i+1}$ and that $\tilde{\ell}_{i+1}=\ell_{i+1}$.
\end{proof}
Thus, \cref{alg:ccont} correctly outputs linear functions $\cbra{\tilde\ell_i}$ so that $\cont_m(\tilde\ell_0,\ldots,\tilde\ell_{m-1})=f$.

The claim regarding the running time is also obvious given the analysis above. We thus see that  \cref{thm:cont-recon} holds.

\begin{remark}\label{rem:rec-not-cont}
 As \cref{thm:tIndNotDense} shows that $t$-independent maps do not necessarily lead to robust hitting sets, our reconstruction algorithm is not continuous at $\vzero$ (recall the discussion in \cref{sec:discuss}): Intuitively, around $\vzero$, there is no way to break the tie between the different polynomials  $\cont_{m}^{(j,i,k)}(\xn)$ and decide which are the consecutive triplets.
\end{remark}


\section{Orbits of read-once formulas}
\label{sec:rof}

In this section we discuss the circuit classes $\anf{}{\f}$ and $\ROF^{\gl{}{\f}}$ (see \cref{def:ANF,def:ROF} below), which are dense in $\VPe$. We construct a hitting set for $\ROF^{\gl{}{\f}}$ and an interpolating set for $\anf{}{\f}$. Finally we observe that the randomized reconstruction algorithm of \cite{gupta2014random} works for every polynomial in $\roanf$. 

We start with basic definitions concerning ROFs and ROANFs and prove \cref{ROANFisDense}.

\begin{definition}\label{def:ROF}
	An \emph{arithmetic read-once formula} (\ROF{} for short) $\Phi$ over a field $\f$ in the variables $\xn=(x_1,\ldots,x_n)$ is a binary tree $T$ whose leaves are labeled with input variables and a pairs of field elements $(\alpha,\beta)\in\f^2$, and whose internal nodes are labeled with the arithmetic operations $\{+,\times\}$ and a field element $\alpha\in\f$. Each input variable can label at most one leaf.
	The computation is performed in the following way: A leaf labeled with the variable $x_i$ and with $(\alpha,\beta)$, computes the polynomial $\alpha x_i+\beta$. If a node $v$ is labeled with the operation $*\in\{+,\times\}$ and with $\alpha\in\f$, and its children compute the polynomials $\Phi_{v_1}$ and $\Phi_{v_2}$, then the polynomial computed at $v$ is $\Phi_v=\Phi_{v_1}*\Phi_{v_2}+\alpha$. A polynomial $f(\xn)$ is called a \emph{read-once polynomial} (\ROP{} for short) if $f(\xn)$ can be computed by a \ROF{}.
\end{definition}


\begin{observation}
	Read-once polynomials are always multilinear polynomials.
\end{observation}

We next define formulas in alternating normal form, as was first defined in \cite{gupta2014random}.

\begin{definition}[Section 3.2 in \cite{gupta2014random}]\label{def:ANF}
	We say that an arithmetic formula $\Phi$, over $\f$, is in \emph{alternating normal form} ($\Phi$ is called an \emph{ANF} for short) if:
	\begin{enumerate}
		\item The underlying tree of $\Phi$ is a complete rooted binary tree (the root node is called the output node). In particular, $\text{\normalfont size}(\Phi)=2^{\text{\normalfont depth}(\Phi)+1}-1$, where $\text{\normalfont size}(\Phi)$ is the number of nodes in the tree of $\Phi$ and $\text{\normalfont depth}(\Phi)$ is the maximum distance of a leaf node from the output node of $\Phi$.
		\item The internal nodes consist of alternating layers of $+$ and $\times$ gates. In particular, the label of an internal node at distance $d$ from the closest leaf node is $+$ if $d$ is even and $\times$ otherwise. So if the root node is a $+$ node, its children are all $\times$ nodes, its grandchildren are all $+$ etc.
		\item The leaves of the tree are labeled with linear functions. That is, each leaf is labeled with $\ell(\xn)=a_0+\sum_{i=1}^na_ix_i$, where each $a_i\in\f$ is a scalar.
	\end{enumerate}
	The \emph{product depth} $\Delta$ of $\Phi$ is the number of layers of product gates. The number of leaves of $\Phi$ is therefore always $4^\Delta$ if the top gate is $+$, and $\frac{1}{2}\cdot 4^\Delta$ if the top gate is $\times $.
\end{definition}

The class $\anf{}{\f}$ mentioned in \cref{sec:res-rof} is defined in terms of the following canonical read-once ANF formula (ROANF for short):

\begin{definition}[Notation from Fact 3.4 of \cite{gupta2014random}]\label{defCroanf}
	We denote the canonical ROANF polynomial, of product depth $\Delta$ on $4^\Delta$ variables, as $\croanf(\xn)$. It is defined recursively as follows:
	\begin{align*}
		\croanf[0](\xn)&=x_1 \\  \croanf[\Delta+1](\xn) &=\croanf \bra{\xn^{(1)}}\croanf\bra{\xn^{(2)}} +\croanf\bra{\xn^{(3)}}\croanf\bra{\xn^{(4)}} \;,
	\end{align*}
	where $\xn^{(i)}$ is the $4^\Delta$-tuple of variables $\{x_{(i-1)\cdot 4^\Delta+1},\ldots,x_{i\cdot 4^\Delta}\}$.
\end{definition}

For example, $\croanf[1]\bra{\xn}=x_1 x_2 + x_3 x_4$.

Observe that any polynomial in $\croanf^{\gla{n}{\f}}$ is an ANF according to \cref{def:ANF}, but not vice versa.

%


%
%
%
%
%

Next we give some basic definitions concerning the underlying tree of a ROF, or of a ROANF.

\begin{definition}
	Let $\Phi$ be a \ROF{} and  $v_i,v_j$  nodes of $\Phi$. The \emph{first common gate} of $v_i,v_j$ (denoted $\text{\normalfont fcg}(v_i,v_j)$) is the first gate in $\Phi$ common to all the paths from $v_i$ and $v_j$ to the root of the formula.
\end{definition}


\begin{definition}\label{def:sib}
	Let $T$ be the computation tree of some $\ROP$ polynomial $g\in\f[\xn]$. For a node $v\in T$ that is not the root, we denote by $\text{\normalfont sib}(v)\in T$ the unique sibling of $v$ in $T$. When clear from context, $\text{\normalfont sib}(v)\in\f[\xn]$ denotes the polynomial computed at node $\text{\normalfont sib}(v)$.
\end{definition}

We may characterize $\text{mon}(\croanf(\xn))$ by the first common gates of pairs of variables appearing in the monomials:
 
\begin{observation}\label{obs:monRoanf}
	$\xn^{\ve}\in\text{\normalfont mon}(\croanf(\xn))$ if and only if $\xn^{\ve}$ is multilinear of degree $2^\Delta$, and for every $x_i\neq x_j\in\text{\normalfont var}(\xn^{\ve})$ it holds that $\text{\normalfont fcg}(x_i,x_j)$ is a product gate.
\end{observation}

\begin{observation}\label{obs:derivCroanf}
	Let $n=4^\Delta$. Let $T$ be the computation tree of $\croanf(\xn)$ (from \cref{defCroanf} above). Fix some variable $x_i\in\xn$ and let $\{v_1,\ldots,v_\Delta\}\subseteq T$ be the addition gates on the path from $x_i$ to the root of $T$, where $v_\Delta$ is the root. Denote with $v_0\in T$  the leaf labeled $x_i$. Then, recalling \cref{def:sib},
	$$\partiald{\croanf}{x_i}=\prod_{k=0}^{\Delta-1}\text{\normalfont sib}(v_k)=\prod_{k=0}^{\Delta-1}\croanf[k](\text{\normalfont var}(\text{\normalfont sib}(v_k)))\;. $$
\end{observation}


\begin{corollary}
	For any set of variables $S\subseteq\xn$, $\partiald{\croanf}{S}$ is either zero, or a product of variable-disjoint ROANFs.
\end{corollary}

\begin{corollary}\label{cor:derivRoanfIndep}
	For any  $\vzero\neq \vu \in \f^{4^\Delta}$, $\partiald{\croanf}{\vu}$ is non-zero.
\end{corollary}
\begin{proof}
	Denote $\vu=(u_1,\ldots,u_n)$. By \cref{obs:derivCroanf}, every monomial of $\partiald{\croanf}{x_i}$ is divisible by $\text{sib}(x_i)$ and is not divisible by $x_i$. Furthermore, for every $j\neq i$, any monomial of $\partiald{\croanf}{x_j}$ that contains $\text{sib}(x_i)$, must also contain $x_i$. Thus, in any linear combination $\partiald{\croanf}{\vu}=\sum_{i=1}^n u_i\partiald{\croanf}{x_i}$, no cancellations can occur as the monomial sets in the summed polynomials are disjoint.
\end{proof}




We first give the simple proof of \cref{ROANFisDense}, that separates $\anf{}{\f}$, $\ROF^{\gl{}{\f}}$ and $\VPe$, and that shows that their closures are equal.

\begin{proof}[Proof of \cref{ROANFisDense}]
	From the definition it is obvious that $\anf{}{\f} \subseteq  \ROF^{\gl{}{\f}}$. 
	It is also clear that the classes are different as the degree of every polynomial in  $\anf{}{\f}$ is always a power of $2$, which is not necessarily the case for  polynomials in $\ROF^{\gl{}{\f}}$. 
	As polynomials in $\ROF^{\gl{}{\f}}$ are multilinear with respect to some basis, it is also clear that $\ROF^{\gl{}{\f}}\subsetneq \VPe$, as the example $f(x)=x^2$ shows. It is also not hard to demonstrate a multilinear polynomial in $\VPe$ that is not in  $\ROF^{\gl{}{\f}}$.
	The next claim follows example 3.8 of \cite{v010a018}.
	\begin{claim}
		$f(\xn)=x_1x_2+x_2x_3+x_3x_1\notin \ROF^{\gl{}{\f}}$.
	\end{claim}
	\begin{proof}
		Assume for a contradiction that there is some ROF formula containing $f$ in its orbit. As $f$ is irreducible, the top gate of $\Phi$ is an addition gate. As there cannot be any cancellations in $\Phi$, the children of the root must compute homogeneous degree $2$ polynomials. It is not hard to see that this means that the polynomial computed cannot be written as a ROF in only three linear functions, as one child of the root must compute a linear function.
	\end{proof}
	To show that the closures are equal, we note that 
	Proposition 3.2 of \cite{gupta2014random} states that any polynomial that is computed by a size $s$ formula, can be computed by  an ANF formula  of size $O(s^4)$. As the leaves of an ANF formula are labeled with linear functions, we can approximate these linear functions with linearly independent linear functions and thus conclude that 
	$\VPe\subseteq \overline{\anf{}{\f}}$. The claim about the closures immediately follows.
\end{proof}

\subsection{A hitting set generator for orbits of read-once formulas}

In this section we prove \cref{thm:PITROPINV} that gives a hitting set for  $\ROF^{\gla{n}{\F}}$. Our proof follows the proof of \cite{shpilka2015read}, who constructed such a generator for ROFs. We note that Minahan and Volkovich significantly improved upon the result of  \cite{shpilka2015read}, namely, they achieved a polynomial-sized hitting set for ROFs. However, we do not know how to adapt their approach to orbits of ROFs and instead use the method of  \cite{shpilka2015read} that is based on taking partial derivatives, an operation that works well when composing the ROF with a $k$-independent map (recall \cref{lem:indDerivLinear}). 
We now turn to proving \cref{thm:PITROPINV}.



\begin{proof}[Proof of \cref{thm:PITROPINV}]

The proof of the theorem is by induction on the number of variables in the  underlying ROF, which we denote by $m$.
In fact, we claim something stronger: 

Let $\Phi$ be a ROF on $m\leq 2^t$ many variables that computes a non-constant polynomial. Then, for  
$f\in \Phi^{\gla{n}{\f}}$ and any \ropinvgen[(t+1)]{} $\G$, over $\f$, $f\circ \G$ is a non-constant polynomial.

For $m\leq 2$ the claim follows from \cref{obs:kwise}.

Let $f\in \F[x_1,\ldots,x_n]$ be in the orbit of some ROF, on $m$ many variables,  $\Phi(w_1,\ldots,w_{m})$. Let $t$ be the smallest integer such that $m\leq 2^t$. By definition, for some linearly independent $n$-variate linear functions $\ell_1,\ldots,\ell_{m}$, $f(\xn)=\Phi\bra{\ell_1(\xn),\ldots,\ell_{m}(\xn)}$
	(where we abuse notation and identify $\Phi$ with the polynomial that it computes). Let $\{\vv_i\}$ be a dual set to $\cbra{\ell_i}$.
 
As in the proof of Lemma 5.1 of \cite{shpilka2015read}, we split the proof into cases depending on the top gate of $\Phi$. 
Let $\G_1,\G_{t}$ be a \ropinvgen[1]{} and a \ropinvgen[t]{}, respectively, such that $\G=\G_1+\G_{t}$.

\textbf{Case $\Phi=\Phi_1+\Phi_2+\alpha$:} As $\Phi_1$ and $\Phi_2$ are variable disjoint, we can assume, WLOG that $|\text{var}\bra{\Phi_1}|\leq m/2 \leq 2^{t-1}$. Assume further, WLOG, that $\frac{\partial \Phi_1}{\partial w_1}\neq 0$. As $\Phi_2$ does not depend on $w_1$, we get from \cref{lem:dualder}  that  $\partiald{f}{{\vv_1}}=\partiald{\Phi_1}{ w_1}\bra{\ell_1,\ldots,\ell_{m}} \neq 0$. By our induction hypothesis,  $\bra{\partiald{f}{{\vv_1}}}\circ \G_{t}=\bra{\partiald{f_1}{{\vv_1}}}\circ \G_{t}$ 
is a non-constant polynomial.  \cref{lem:indDerivLinear} implies that $f\circ\G=f\circ(\G_{1}+\G_t)\neq 0$, and it is clearly not a constant polynomial.

\textbf{Case $\Phi=\Phi_1\times \Phi_2 + \alpha$:} As we can assume that both $\Phi_1$ and $\Phi_2$ are non-constant (there is always such formula computing $\Phi(\bw)$ if it is not the constant polynomial), they both contain less than $m$ variables.
Denote $f_i = \Phi_i\bra{\ell_1,\ldots,\ell_m}$, so that $f=f_1\cdot f_2 + \alpha$.
The induction hypothesis implies that $f_1 \circ \G_{t+1}$ and $f_2 \circ \G_{t+1}$ are both non-constant. Hence, $f\circ \G_{t+1} = \bra{f_1 \circ \G_{t+1}}\cdot \bra{f_2 \circ \G_{t+1}}+\alpha$ is also non-constant, as we wanted to prove.   
\end{proof}

As before, \cref{cor:PITROPINV} follows immediately from \cref{thm:PITROPINV} and \cref{obsHitSetGen}.

\subsection{An interpolating set generator for $\anf{}{\f}$}

In this section, we construct an interpolating set generator for $\anf{}{\f}$, thus proving \cref{thm:pitSumOfRoanfThm}.
We restate the theorem to ease the reading.

\pitSumOfRoanfThm*

The first step in the proof is a reduction to the case where $f_1$ and $f_2$ are ``almost the same''. Recall that by \cref{fact:anf-sym}, $f_1$ and $f_2$ can be equal and still compute different linear functions at their bottom layer. The next lemma (roughly) shows that composing $\croanf(\xn)$ with an $O(\Delta)$-independent map, preserves equivalence of different ANFs while not introducing any new equivalences. 


\begin{restatable}{lemma}{pitSumOfRoanf}\label{lem:pitSumOfRoanf}
	Let $f_1=\croanf[\Delta_1](A_1\xn+\vb_1),f_2=\croanf[\Delta_2](A_2\xn+\vb_2)\in \anfn{}{\f}$ and $f=f_1-f_2$. For $i=1,2$, denote by $h_i\triangleq x_0^{\deg(f_i)}f_i(\frac{x_1}{x_0},\ldots,\frac{x_n}{x_0})$ the homogenization of $f_i$, and let $\tilde A_i$ be an extension of $A_i$ such that $\tilde A_i\in\gl{n+1}{\f}$ and $h_i=\croanf[\Delta_i](\tilde A_i\xn)$. 
	Set $k= 2\max\{\Delta_1,\Delta_2\}+7$ and let $\G$ be any uniform \ropinvgen[k]{}. If $f\neq 0$ then at least one of the following holds:
	\begin{enumerate}
		\item $f\circ \G\neq 0$. \label{case:equalG}
		\item $\Delta_1=\Delta_2$, and  there is a $1-1$ map between the quadratic  forms of  $h_2(\tilde A_1^{-1}\xn)$ and those of $\croanf[\Delta_1](\xn)$, such that any two quadratics that were matched have the same monomials, \emph{possibly with different coefficients}.\footnote{Thus, composition with $\G$ does not exactly preserve equivalence.} Furthermore, the map between the quadratics is a $\TR[4^{\Delta_1-1}][\f]$ symmetry (see \cref{defTR}). \label{case:different}
		\end{enumerate}
\end{restatable}

Observe that if $\cbra{\ell_{i,j}}$ are linear functions such that  $f_i=\croanf[\Delta_i]\bra{\ell_{i,1}(\xn),\ldots,\ell_{i,4^{\Delta_i}}}$, then the condition ``the monomials appearing in the quadratic forms of $h_2\bra{(\tilde A_1)^{-1} \xn}$ are identical to the monomials of the quadratic forms of $\croanf[\Delta_1](\xn)$, up to $\TR[4^{\Delta_1}][\f]$ symmetry'' is equivalent to saying that there exists a permutation $\pi\in \TR[4^{\Delta_1-1}][\f]$, matching quadratics in $f_2$ to those of $f_1$, such that when we represent the $i$th quadratic $q^{(2)}_i$ of $f_2$ according to the linear functions $\cbra{\ell_{1,1},\ldots,\ell_{1,4^{\Delta_1}}}$, then $q^{(2)}_i$ has the same set of $\cbra{\ell_{1,1},\ldots,\ell_{1,4^{\Delta_1}}}$-monomials as $q^{(1)}_{\pi(i)}$, the $\pi(i)$th quadratic in $f_1$. In general, whenever we say ``up to $\TR[4^{\Delta_1-1}][\f]$ symmetry'' we mean that there exists a permutation $\pi \in \TR[4^{\Delta_1-1}][\f]$ such that the statement holds  when we apply $\pi$ to the quadratics computed at the bottom layers.


Once we have this in mind we can see that the only ``bad'' case  is when, for every $i$, $\ell_{2,i}=\alpha_i\cdot \ell_{1,i}$, for scalars $\alpha_i\in\f$ (possibly after applying some  $\TR[4^{\Delta_1-1}][\f]$ symmetry). Thus, the proof of \cref{thm:pitSumOfRoanfThm} would follow from the next lemma. 

\begin{restatable}{lemma}{pitRoanfSame}\label{lem:pitRoanfSame}
	Let $\ell_1(\xn),\ldots,\ell_n(\xn)$ be linearly independent linear forms, and let $\alpha_1,\ldots,\alpha_n\in\f$ be non-zero constants. Let $f=\croanf(\ell_1,\ldots,\ell_n)$ and $g=\croanf(\alpha_1\ell_1,\ldots,\alpha_n\ell_n)$, and let $\G$ be a \ropinvgen[(2\Delta+2)]{}. It holds that if 
	$f-g\neq 0$ then $(f-g)\circ \G\neq 0$.
\end{restatable}

We first give the formal proof of the theorem and then prove the main lemmas.

\begin{proof}[Proof of \cref{thm:pitSumOfRoanfThm}]
	Let $h_1,h_2$ be the homogenizations of $f_1,f_2$ as in the premise of \cref{lem:pitSumOfRoanf}. Assume Case~\ref{case:different} of \cref{lem:pitSumOfRoanf} holds, as otherwise we are done. Then, for $n=4^{\Delta_1}$, this assumption implies that  for some linearly independent linear forms $\ell_1,\ldots,\ell_n$ and non-zero constants $\alpha_1,\ldots,\alpha_n\in\f$,   $h_1=\croanf[\Delta_1](\ell_1,\ldots,\ell_n)$ and $h_2=\croanf[\Delta_1](\alpha_1\ell_1,\ldots,\alpha_n\ell_n)$. By \cref{lem:pitRoanfSame}, if $f\neq 0$ then $(h_1-h_2)\circ\G\neq 0$; and by the following lemma (\cref{lem:uniIndGenHitsHom}), we may conclude $f\circ\G\neq 0$.
\end{proof}

\begin{restatable}{lemma}{uniIndGenHitsHom}\label{lem:uniIndGenHitsHom}
	Let $\xn=(x_1,\ldots,x_n)$ and $f\in\f[\xn]$ be a polynomial of degree $d$. Let $g(x_0,\xn)=x_0^df(\frac{x_1}{x_0},\ldots,\frac{x_n}{x_0})$ be the homogenization of $f$, and let $\G:\f^t\to\f^{n+1}$ be a polynomial map such that the coordinates of $\G$ are homogeneous polynomials of identical degree. Let $H:\f^t\to\f^n$ be the restriction of $\G$  to the coordinates  in $[n]$ (i.e., we ignore the $0$th coordinate). If $g\circ \G\neq 0$ then $f\circ H\neq 0$.
\end{restatable}

\begin{proof}
	Write $g(x_0,\xn)=\sum_{i=0}^d x_0^if^{[d-i]}(\xn)$, and denote by $\G_0$ the $0$th coordinate of $\G$ (such that $\G=(\G_0,H)$). We get:
	$$ g\circ \G=\sum_{i=0}^d (\G_0)^i\cdot(f^{[d-i]}\circ H). $$
	Fix $i\in [d+1]_0$ to be the minimal index such that $f^{[d-i]}\circ H\neq 0$. Such an index must exist, because $g\circ \G\neq 0$. As all coordinates of $\G$ are homogeneous and of identical degree, for any $i<i'\in[d]$ such that $f^{[d-i']}\circ H$ is non-zero, we must have $\deg(f^{[d-i]}\circ H)>\deg(f^{[d-i']}\circ H)$. Thus, nothing can cancel $f^{[d-i]}\circ H$ in $f\circ H$, proving $f\circ H\neq 0$.
\end{proof}

\subsubsection{Proof of \cref{lem:pitSumOfRoanf}}

	The high-level strategy for proving \cref{lem:pitSumOfRoanf} is as follows: first, we show that if Case~\ref{case:different} of the lemma is false, then there are  $\vv,\vu\in\f^n$ such that $\partiald{^2f}{\vv\partial\vu}=\partiald{^2f_1}{\vv\partial\vu}\neq 0$. This is proven in \cref{lem:separateRoanfSum}, based on the structural result of \cref{lem:roanfMonInc}. After that, we prove that \ropinvgen[(k-2)]{}s hit $\partiald{^2f_1}{\vv\partial\vu}$, in \cref{lem:pitDerivRoanf}.

To prove \cref{lem:pitSumOfRoanf}, we first set out to prove that inclusion of monomial sets is enough to deduce that Case~\ref{case:different} of \cref{lem:pitSumOfRoanf} holds:

\begin{lemma}\label{lem:roanfMonInc}
	Let $g(\xn)=\croanf(A\xn+\vb)$ for some $(A,\vb)\in\gla{n}{\f}$. Let $q_1,\ldots,q_{4^{\Delta-1}}$ denote the quadratic forms of $\croanf$ such that $g=\croanf[\Delta-1](q_1(A\xn+\vb),\ldots,q_{4^{\Delta-1}}(A\xn+\vb))$. If $\text{\normalfont mon}(g)\subseteq\text{\normalfont mon}(\croanf(\xn))$, then $\vb=0$ and $\text{\normalfont mon}(q_i(A\xn))=\text{\normalfont mon}(q_i(\xn))$, up to $\TR[4^{\Delta-1}][\f]$ symmetry. In particular, $\text{\normalfont mon}(g)=\text{\normalfont mon}(\croanf(\xn))$.
\end{lemma}

\begin{proof}
	The proof is by induction on $\Delta$. 
	
	For $\Delta=1$, we know $\text{mon}(g)\subseteq\{x_1x_2,x_3x_4\}$. $\croanf[1](\xn)$ is irreducible, so $\text{mon}(g)\neq\{x_1x_2\}$ or $\{x_3x_4\}$, and $g$ is non-constant so $\text{mon}(g)=\{x_1x_2,x_3x_4\}$. Now, let $\ell_1(\xn),\ldots,\ell_4(\xn)$ denote linearly independent linear functions such that $g=\ell_1(\xn)\ell_2(\xn)+\ell_3(\xn)\ell_4(\xn)$, and denote $\alpha_i\triangleq \ell_i(0)$. The $1$-homogeneous part of $g$ is given by:
	$$ g^{[1]}=\alpha_1\ell_2^{[1]}(\xn)+\alpha_2\ell_1^{[1]}(\xn)+\alpha_3\ell_4^{[1]}(\xn)+\alpha_4\ell_3^{[1]}(\xn). $$
	
	As $g$ is $2$-homogeneous, $g^{[1]}=0$. As the $\ell^{[1]}_i$s are linearly independent, this implies $\alpha_1=\ldots=\alpha_4=0$, and therefore $\vb=0$, proving the base case.
	
	Assume $\Delta>1$ and denote $\croanf(\xn)=F_1F_2+F_3F_4$, where $F_1,\ldots,F_4$ are the grandchildren of the root of $\croanf$. In particular, each $F_i$ is an 
	$\croanf[\Delta-1](\xn)$ formula (on one quarter of the variables). 
	We note that $g$ is $2^\Delta$ homogeneous because $\text{mon}(g)\subseteq\text{mon}(\croanf)$, so $g=\croanf(A\xn)$ (because $\croanf(A\xn+\vb)^{[2^\Delta]}=\croanf(A\xn)$). Denote $g=g_1g_2+g_3g_4$ where $g_i(\xn)=F_i(A\xn)$.
	
	First, note that $\text{var}(g)=\text{var}(\croanf(\xn))$: we already know $\text{var}(g)\subseteq\text{var}(\croanf(\xn))$, and $g$ must depend on at least $4^\Delta$ variables, or the $4^\Delta$ linear functions on the leaves cannot be linearly independent.
	
	Next, observe that $g_1g_2$ and $g_3g_4$ must be variable disjoint: if $x_i\in\text{var}(g_1g_2)\cap\text{var}(g_3g_4)$, then $\left(\partiald{g}{x_i}\right)(A^{-1}\xn)$ is a sum of non-constant, variable-disjoint, multilinear polynomials, and $\left(\partiald{g}{x_i}\right)(\xn)$ is therefore irreducible (recall \cref{irreducibleMultilinear}). However, if we denote by $x_j$ the sibling of $x_i$ in $\croanf(\xn)$, the fact that $\text{mon}(g)\subseteq\text{mon}(\croanf(\xn))$ implies that every monomial of $\partiald{g}{x_i}(\xn)$ is divisible by $x_j$. As $\Delta>1$, we have $\deg\bra{\partiald{g}{x_i}}\geq 3$, and therefore $\partiald{g}{x_i}$ must be reducible, in contradiction. Thus, $\text{var}(g_1g_2)\cap\text{var}(g_3g_4)=\emptyset$, and in particular $\text{mon}(g_1g_2),\text{mon}(g_3g_4)\subseteq\text{mon}(\croanf)$.
	
	Next, assume, WLOG, there exist some monomial $\xn^{\ve}\in\text{mon}(F_1F_2)$ such that $\xn^{\ve}\in\text{mon}(g_1g_2)$. If $g_1g_2$ contains a monomial of $F_3F_4$, then $g_1g_2$ can be partitioned into a sum of two variable-disjoint, non-constant, multilinear polynomials; which would contradict reducibility of $g_1g_2$. Thus, $\text{mon}(g_1g_2)\subseteq\text{mon}(F_1F_2)$. 
	As we showed that $\text{var}(g)=\text{var}(\croanf(\xn))$, the conditions on the monomials implies that there must exist some monomial of $F_3F_4$ in $g$, so we may conclude $\text{mon}(g_3g_4)\subseteq\text{mon}(F_3F_4)$, and in addition, $\text{var}(g_1g_2)=\text{var}(F_1F_2)$ and $\text{var}(g_3g_4)=\text{var}(F_3F_4)$.
	
	To apply induction, it remains to prove that $\text{mon}(g_i)\subseteq\text{mon}(F_i)$ for $i\in[4]$ (up to $\TR[][\f]$); focus on $g_1g_2$ and WLOG assume $\text{var}(g_1)\cap\text{var}(F_1)\neq\emptyset$.
	
	As all monomials of $g_1g_2$ are multilinear, $\text{var}(g_1)\cap\text{var}(g_2)=\emptyset$. As $\Delta>1$, we may denote by $p_1,p_2,p_3,p_4$ the variable-disjoint polynomials such that $F_1=p_1+p_2$ and $F_2=p_3+p_4$:
	$$ F_1F_2=(p_1+p_2)(p_3+p_4)=p_1p_3+p_1p_4+p_2p_3+p_2p_4 \;.$$
	
	We now show that $g_1$ cannot contain variables from both $F_1$ and $F_2$. Assume there exist monomials $\xn^{{\ve_1}},\xn^{{\ve_2}}\in\text{mon}(g_1)$ such that $\xn^{{\ve_1}}$ contains variables from $\text{var}(F_1)$ and $\xn^{{\ve_2}}$ contains variables from $\text{var}(F_2)$ ($\xn^{{\ve_1}}$ and $\xn^{{\ve_2}}$ may be the same monomial). 
	WLOG assume $\text{var}(\xn^{{\ve_1}})\cap\text{var}(p_1)\neq\emptyset$, and likewise $\text{var}(\xn^{{\ve_2}})\cap\text{var}(p_3)\neq\emptyset$. Let $\xn^{\vc}\in\text{mon}(g_2)$, and let $x_i|\xn^{\vc}$. If $x_i\in\text{var}(p_2)$, then $\xn^{{\ve_1}}\cdot \xn^{\vc}\in\text{mon}(g_1g_2)$ is a monomial involving variables from both $p_1$ and $p_2$, in contradiction; by a symmetric argument, we cannot have $x_i\in\text{var}(p_4)$. Thus, all monomials of $g_2$ may involve only variables of $p_1$ and $p_3$, i.e., $\text{var}(g_2)\subseteq\text{var}(p_1)\cupdot\text{var}(p_3)$. Therefore, the only way to get monomials involving variables of $p_2$ or $p_4$ is via monomials of $g_1$, so $g_1$ must contain monomials $\xn^{{\ve_1}'},\xn^{{\ve_2}'}$ containing variables of $p_2$ and $p_4$, respectively (here we use the fact that $\text{var}(g_1g_2)=\text{var}(F_1F_2)$). As before, we get $\text{var}(g_2)\subseteq\text{var}(p_2)\cupdot\text{var}(p_4)$, in contradiction.
	
	We can therefore conclude that $\text{var}(g_1)\subseteq \text{var}(F_1)$. Using $\text{var}(g_1g_2)=\text{var}(F_1F_2)$, we deduce $\text{var}(g_2)\cap\text{var}(F_2)\neq\emptyset$, and repeating the argument of the previous paragraph we conclude $\text{var}(g_2)\subseteq\text{var}(F_2)$, which implies $\text{var}(g_i)=\text{var}(F_i)$ for $i=1,2$.
	
	As $\text{mon}(g_1g_2)\subseteq\text{mon}(F_1F_2)$, we may conclude $\text{mon}(g_i)\subseteq\text{mon}(F_i)$ (for $i=1,2$):
	$$ \text{mon}(F_i)=\{\restr{\xn^{\ve}}{(\xn\setminus\text{var}(F_i))=1}:\xn^{\ve}\in\text{mon}(F_1F_2)\}\supseteq\{\restr{\xn^{\ve}}{(\xn\setminus\text{var}(g_i))=1}:\xn^{\ve}\in\text{mon}(g_1g_2)\}=\text{mon}(g_i). $$
	
	Finally, we may apply the induction hypothesis and conclude $\vb=\vzero$ and $\text{mon}(q_i(A\xn))=\text{mon}(q_i(\xn))$, up to $\TR[4^{\Delta-1}][\f]$ symmetry. I.e., there is a permutation $\pi \in \TR[4^{\Delta-1}][\f]$ such that  $\text{mon}(q_i(A\xn))=\text{mon}(q_{\pi(i)}(\xn))$ ($\TR[4^{\Delta-1}][\f]$ symmetry enters every time we use ``WLOG'' in the proof).
\end{proof}

The next step is showing that, if Case~\ref{case:different} of \cref{lem:pitSumOfRoanf} does not hold, then we may choose a pair of vectors by which to take a derivative of $f=f_1-f_2$ such that $\partiald{^2f_1}{\vv_1\vv_2}=0$ and $\partiald{^2f_2}{\vv_1\vv_2}\neq 0$. This is formalized in \cref{lem:separateRoanfSum} below, and is proved by applying \cref{lem:roanfMonInc}.

%

\begin{lemma}\label{lem:separateRoanfSum}
	Let $f=\croanf(A_1\xn)$ and $g=\croanf(A_2\xn)$, for some $A_1,A_2\in\gl{n}{\f}$. Denote $\tilde g\triangleq g(A_1^{-1}\xn)$. If $\text{\normalfont mon}(\tilde g)\neq\text{\normalfont mon}(\croanf(\xn))$, then there exist $\vv,\vu\in\f^{n}$ such that $\partiald{^2f}{\vv\partial\vu}=0$ and $\partiald{^2g}{\vv\partial\vu}\neq 0$.
\end{lemma}
\begin{proof}
	Let $\ell_1(\xn),\ldots,\ell_n(\xn)$ be linearly independent linear forms such that $f=\croanf(\ell_1(\xn),\ldots,\ell_{4^\Delta}(\xn))$, and let $\cbra{\vv_1,\ldots,\vv_{4^\Delta}}$ be a dual set.
	
	By \cref{lem:roanfMonInc}, the fact that $\text{mon}(\tilde g)\neq\text{mon}(\croanf(\xn))$ implies $\text{mon}(\tilde g)\not\subseteq\text{mon}(\croanf(\xn))$. Fix some monomial $\xn^{\ve}\in\text{mon}(\tilde g)\setminus\text{mon}(\croanf(\xn))$, and choose $\vv,\vu$ as follows:
	\begin{itemize}
		\item If $\xn^{\ve}$ is not a multilinear monomial, let $x_i$ be such that $x_i^2|\xn^{\ve}$. Set $\vv=\vu\triangleq \vv_i$. In this case, we get from \cref{lem:dualder} that $\partiald{^2f}{\vv\partial\vu}=\partiald{^2\croanf}{x_i^2}(\ell_1,\ldots,\ell_{4^\Delta})=0$, as $\croanf$ is multilinear. Clearly $\partiald{^2 g}{\vv\partial\vu}\neq 0$.
		
		\item If $\xn^{\ve}$ is multilinear, then let  $x_i,x_j\in\text{var}(\xn^{\ve})$ be such that $\text{fcg}(x_i,x_j)$ is an addition gate (all monomials of $\tilde g$ are of degree exactly $2^\Delta$, so \cref{obs:monRoanf} implies the existence of such a pair of variables). Set $\vv\triangleq \vv_i,\vu\triangleq \vv_j$. \cref{lem:dualder} again implies that $\partiald{^2f}{\vv\partial \vu}=\partiald{^2\croanf}{x_i\partial x_j}=0$, because $\text{fcg}(x_i,x_j)$ is an addition gate in $\croanf$. As before, it is clear that $\partiald{^2 g}{\vv\partial\vu}\neq 0$. \qedhere
	\end{itemize}
\end{proof}

Looking back at \cref{lem:pitSumOfRoanf}, \cref{lem:separateRoanfSum} allows us to separate $f_1$ from $f_2$, provided Case~\ref{case:different} of \cref{lem:pitSumOfRoanf} does not hold. We still need to provide a hitting set for $\partiald{^2f_1}{\vv\partial\vu}$, where $\vv,\vu$ are arbitrary, and satisfy $\partiald{^2f_1}{\vv\partial\vu}\neq 0$. To do so, we 
reduce $\partiald{^2f_1}{\vv\partial\vu}$ to a single, non-zero product of variable-disjoint ROPs composed with affine transformations (\cref{lem:pitDerivRoanf}). For simplicity, we first reduce to a product of ROPs in the standard basis in \cref{lem:hitSecondDeriv}, and subsequently extend the result to affine  orbits in \cref{lem:pitDerivRoanf}.

\begin{lemma}\label{lem:hitSecondDeriv}
	Let $\Delta\geq 2$, and let $f(\xn)=\sum_{i,j}\alpha_{i,j}\partiald{^2\croanf(\xn)}{x_i\partial x_j}$ be some non-zero linear combination of second derivatives of $\croanf(\xn)$. Then, there exist variables $x_i,x_j$, sets $D,Z\subseteq\xn$ such that $|D|\leq 2$ and $|Z|=2$, and a constant $\beta_{i,j}$ such that 
	$$\restr{\left(\partiald{^{|D|}f}{D}\right)}{Z=\vzero}=\beta{i,j}\restr{\left(\partiald{^{2+|D|}\croanf(\xn)}{x_i\partial x_j\partial D}\right)}{Z=\vzero}\neq 0.$$
\end{lemma}
\begin{proof} \fussy
	First, assume there exist some $i,j$ such that $\alpha_{i,j}\partiald{^2\croanf(\xn)}{x_i\partial x_j}\neq 0$ and $x_i\neq\text{sib}(x_j)$. Set $D= \{\text{sib}(x_i),\text{sib}(x_j)\}$. By \cref{obs:derivCroanf}, $\partiald{^4\croanf(\xn)}{x_i\partial x_j\partial D}\neq 0$ and is a product of variable-disjoint ROPs that do not depend on $x_i$ nor on $x_j$.
	
	Consider any pair $\{i',j'\}\neq\{i,j\}$ and set $h=\partiald{^2\croanf(\xn)}{x_{i'}\partial x_{j'}}$. Note that if $\partiald{^2h}{D}\neq 0$ then $\partiald{^2h}{D}$ is divisible by $x_i$ or by $x_j$ (or both, if  $\{i,j\}\cap \{i',j'\}=\emptyset$). If we set $Z\triangleq\{x_i,x_j\}$, then $\restr{\left(\partiald{^2h}{D}\right)}{Z=\vzero}=0$. This is true for any $\{i',j'\}\neq\{i,j\}$, and as $\partiald{^4\croanf(\xn)}{x_i\partial x_j\partial D}$ does not depend on $x_i$ nor on $x_j$ we get
	$$\restr{\left(\partiald{^2f}{D}\right)}{Z=\vzero}=\restr{\left(\partiald{^4\croanf}{x_i\partial x_j\partial D}\right)}{Z=\vzero}\neq 0\;.$$
	
	Next, assume all non-zero summands of $f$, $\alpha_{i,j}\partiald{^2\croanf(\xn)}{x_i\partial x_j}$, satisfy $x_i=\text{sib}(x_j)$. Note that if $x_ix_j+x_{i'}x_{j'}$ is a quadratic form of $\croanf(\xn)$, then $\partiald{^2\croanf(\xn)}{x_i\partial x_j}=\partiald{^2\croanf(\xn)}{x_{i'}\partial x_{j'}}$ (\cref{obs:derivCroanf}). Therefore, 
	$$ f=\sum_{\substack{x_ix_j+x_{i'}x_{j'}\;is\;a\\quadratic\;of\;\croanf}}(\alpha_{i,j}+\alpha_{i',j'})\partiald{^2\croanf(\xn)}{x_i\partial x_j} \;.$$
	
	Fix some $i,j,i',j'$ such that $q_1=x_ix_j+x_{i'}x_{j'}$ is a quadratic of $\croanf(\xn)$, and $(\alpha_{i,j}+\alpha_{i',j'})\partiald{^2\croanf(\xn)}{x_i\partial x_j}\neq 0$. As $\Delta\geq 2$, $q_1$ has a sibling quadratic form; denote it by $q_2\triangleq \text{sib}(q_1)=x_kx_\ell+x_{k'}x_{\ell'}$ and set $D\triangleq\{x_k\}$. Note that by \cref{obs:derivCroanf}, $(\alpha_{i,j}+\alpha_{i',j'})\partiald{^3\croanf(\xn)}{x_i\partial x_j\partial D}\neq 0$,   does not depend on $x_i,x_j,x_{i'},x_{j'}$, and is a product of variable-disjoint ROPs.
	
	Set $Z=\{x_i,x_{i'}\}$. Consider any pair $\{s,t\}$ such that $\{s,t\}\notin\{\{i,j\},\{i',j'\}\}$ and $x_s=\text{sib}(x_t)$. Set $h=\partiald{^2\croanf}{x_t\partial x_s}$.
	If 
	$\partiald{h}{x_k}\neq 0$ then it is divisible by the quadratic form $q_1=x_ix_j+x_{i'}x_{j'}$ (by \cref{obs:derivCroanf}), and thus $\restr{\left(\partiald{h}{D}\right)}{Z=\vzero}=0$. Hence,
	\begin{equation*}
		\restr{\left(\partiald{f}{D}\right)}{Z=\vzero}=\restr{\left((\alpha_{i,j}+\alpha_{i',j'})\partiald{^3\croanf(\xn)}{x_i\partial x_j\partial D}\right)}{Z=\vzero} \neq 0 \;. \qedhere
	\end{equation*} 
\end{proof}

\begin{lemma}\label{lem:pitDerivRoanf}
	Let $\Delta\geq 2$, let $f=\croanf(A\xn+\vb)$ for some $(A,\vb)\in\gla{n}{\f}$, and let $\vw,\vu\in\f^n$. Then, for any \ropinvgen[(2\Delta+5)]{} $\G$, if $\partiald{^2f}{\vw\partial\vu}\neq 0$ then $\partiald{^2f}{\vw\partial\vu}\circ\G\neq 0$.
\end{lemma}
\begin{proof}
	Let $\ell_1(\xn),\ldots,\ell_{4^\Delta}(\xn)$ be linearly independent linear functions such that $f=\croanf(\ell_1,\ldots,\ell_{4^\Delta})$. 
	Let $\cbra{\vv_1,\ldots,\vv_{4^\Delta}}$ be a dual set. There exist constants $\alpha_{i,j}$ such that:
	$$ 0\neq  \partiald{^2f}{\vw\partial\vu}(\xn)=\sum_{i,j}\alpha_{i,j}\partiald{^2\croanf}{x_i\partial x_j}(A\xn+\vb)\;. $$
	
	Denote $g(\xn)\triangleq\partiald{^2f}{\vw\partial\vu}(A^{-1}\xn-A^{-1}\vb)=\sum_{i,j}\alpha_{i,j}\partiald{^2\croanf}{x_i\partial x_j}(\xn)$, and let $x_{i_0},x_{j_0}$, $D=\{x_k,x_\ell\}$, $Z=\{x_r,x_m\}$ and $\beta_{i_0 j_0}$ be as promised by \cref{lem:hitSecondDeriv}. Thus,\footnote{Note that by  \cref{lem:hitSecondDeriv} we may have $|D|=1$, but we may add some other variable $x_\ell$  to simplify the notation.}
	\begin{equation}\label{eq:2nd-der}
	\restr{\left(\partiald{^2 g}{D}(\xn)\right)}{Z=\vzero}=\restr{\left(\beta_{i_0j_0}\partiald{^4\croanf}{x_{i_0}\partial x_{j_0}\partial D}(\xn)\right)}{x_r=x_m=0}\neq 0 \;.	
	\end{equation} 
 	From \cref{lem:dualder} and Equation~\eqref{eq:2nd-der}	we deduce that
	$$ \restr{\left(\partiald{^4 f}{\vw\partial\vu\partial\vv_k\partial\vv_\ell}(\xn)\right)}{\ell_r=\ell_m=0}=\restr{\left(\partiald{^2 g}{D}(A\xn+\vb)\right)}{\ell_r=\ell_m=0}=\restr{\left(\beta_{i_0 j_0}\partiald{^2}{\vv_k\partial\vv_\ell}\left(\partiald{^2\croanf}{x_i\partial x_j}(A\xn+\vb)\right)\right)}{\ell_r=\ell_m=0}\neq 0 \;.$$
	
	Let $\G = \G_1+\G_2+ \G_{2\Delta+1}$ be a $(2\Delta+5)$-independent map where $\G_1,\G_2$ are \ropinvgen[2]{}s,  $\G_{2\Delta+1}$ is a \ropinvgen[(2\Delta+1)]{}, and $\G_1,\G_2,\G_{2\Delta+1}$ are variable-disjoint.  
	As $\restr{\left(\partiald{^4 f}{\vw\partial\vu\partial\vv_k\partial\vv_\ell}(\xn)\right)}{\ell_r=\ell_m=0}$ is a non-zero product of ROPs composed with an affine transformation, where the underlying ROPs depend on at most $4^\Delta$ variables, we get from \cref{thm:PITROPINV} that $\restr{\left(\partiald{^4 f}{\vw\partial\vu\partial\vv_k\partial\vv_\ell}(\xn)\right)}{\ell_r=\ell_m=0}\circ\G_{2\Delta+1}\neq 0$.  \cref{lem:indProjectZero} implies that $\partiald{^4 f}{\vw\partial\vu\partial\vv_k\partial\vv_\ell}(\G_{2\Delta+1}+\G_2)\neq 0$. Finally, from \cref{lem:indDerivLinear} it follows that $\partiald{^2f}{\vw\partial\vu}(\G_{2\Delta+1}+\G_2+\G_1)\neq 0$, as required.
\end{proof}

We are now ready to prove \cref{lem:pitSumOfRoanf}. 

%

\begin{proof}[Proof of \cref{lem:pitSumOfRoanf}]
	First, assume $\Delta_1\neq \Delta_2$. WLOG assume $\Delta_1>\Delta_2$. Let $\ell_1,\ldots,\ell_{4^{\Delta_1}}$ be linearly independent linear functions such that $f_1=\croanf[\Delta_1](\ell_1,\ldots,\ell_{4^{\Delta_1}})$. There must exist some $i$ such that $\ell_i$ is \emph{not} spanned by the linear functions at the leaves of $f_2$. Fix some vector $\vv$ such that $\ell^{[1]}(\vv)=0$ for every linear function $\ell$ labeling a leaf of $f_2$, and such that $\ell_i^{[1]}(\vv)=1$. By \cref{lem:dualder,cor:derivRoanfIndep}, $\partiald{f_2}{\vv}=0$ and $\partiald{f_1}{\vv}\neq 0$; thus, $0\neq \partiald{f}{\vv}=\partiald{f_1}{\vv}$. From \cref{lem:pitDerivRoanf} it follows that any \ropinvgen[(2\Delta_1+5)]{} $\G'$ satisfies $\partiald{f}{\vv}\circ\G'\neq 0$; and therefore,  using \cref{lem:indDerivLinear}, we get  $f\circ \G\neq 0$, so Case~\ref{case:equalG} of the lemma holds.
	
	Next, assume $\Delta_1=\Delta_2$ and denote $h\triangleq h_1-h_2$ (recall that $h_i$ is the homogenization of $f_i$). As $\G$ is uniform, \cref{lem:uniIndGenHitsHom} implies that it suffices to prove that either $h\circ \G\neq 0$ (where we extend $\G$ to $n+1$ coordinates such that $\G$ is still a uniform \ropinvgen[k]{}) or that Case~\ref{case:different} of the lemma holds. 

	Assume that $h\circ\G= 0$. \cref{lem:separateRoanfSum,lem:pitDerivRoanf} imply that $\croanf(\xn)$ and $h_2\bra{\tilde{A_1}^{-1}(\xn)}$ have the same set of monomials. From \cref{lem:roanfMonInc} we conclude that Case~\ref{case:different}  holds.
	
\end{proof}

\subsubsection{Proof of \cref{lem:pitRoanfSame}}

Finally, we conclude the proof of \cref{thm:pitSumOfRoanfThm} by proving  \cref{lem:pitRoanfSame} that gives a hitting set for the difference of two polynomials in \anfn{\Delta}{\f} that, up to constant factors, have the same linear functions on the leaves. 

\begin{proof}[Proof of \cref{lem:pitRoanfSame}]
	First, if $f=\alpha g$ for some $\alpha\in\f$, then $f-g\in\anfn{}{\f}$ and the lemma follows from \cref{thm:PITROPINV}.  
	We therefore assume that $f$ is not a multiple of $g$, and denote that by $f\not\propto g$. 
	
	For any node $u$ in the complete binary tree of depth $2\Delta$, denote by $u_f$ the polynomial computed at node $u$ in $\croanf(\ell_1,\ldots,\ell_n)$, and by $u_g$ the polynomial computed at node $u$ in $\croanf(\alpha_1\ell_1,\ldots,\alpha_n\ell_n)$. Fix a node $u$ satisfying $u_f(\xn)\not\propto u_g(\xn)$, such that $u$ is a deepest node with that property. In particular, each child of $u_f$ is a multiple of the corresponding child of $u_g$. Note that, as $f\not\propto g$, such a node $u$ must exist; and by the premise of the lemma, $u_f$ and $u_g$ are not leaves. In addition, $u_f$ and $u_g$ must be addition gates, otherwise we may choose a child $u'$ of $u$ such that $u'_f(\xn)\not\propto u'_g(\xn)$.
	
	Let $\cbra{\vv_1,\ldots,\vv_n}$ be a dual set to $\cbra{\ell_1,\ldots,\ell_n}$. Denote $u_f=f_1f_2+f_3f_4$ and $u_g=g_1g_2+g_3g_4$, where the $f_i$s are the grandchildren of $u_f$ and the $g_i$s are the grandchildren of $u_g$. By choice of $u$, there exist constants $\alpha,\beta\in\f$ such that $f_1f_2=\alpha\cdot g_1g_2$ and $f_3f_4=\beta\cdot g_3g_4$, and $\alpha\neq \beta$ (otherwise $u_f=\alpha\cdot  u_g$). WLOG, assume $f_1,g_1$ are ancestors of the leaf labeled $\ell_1$ (or $\alpha_1\ell_1$), and $f_3,g_3$ are ancestors of the leaf labeled $\ell_3$ (or $\alpha_3\ell_3$). By \cref{obs:derivCroanf}, there exist polynomials $F(\xn),G(\xn)$ such that:
	\begin{eqnarray*}
		\partiald{f}{\vv_1}&=F(\xn)f_2(\xn)\partiald{f_1}{\vv_1}(\xn)  \; ,  \quad
		&\partiald{f}{\vv_3}=F(\xn)f_4(\xn)\partiald{f_3}{\vv_3}(\xn)\;,\\
		\partiald{g}{\vv_1}&=G(\xn)g_2(\xn)\partiald{g_1}{\vv_1}(\xn)  \quad \text{and} \quad
		&\partiald{g}{\vv_3}=G(\xn)g_4(\xn)\partiald{g_3}{\vv_3}(\xn)\;.
	\end{eqnarray*}
	Observe that 
	\begin{equation}\label{eq:anf-v1}
		 \partiald{(f-g)}{\vv_1}=F(\xn)f_2(\xn)\partiald{f_1}{\vv_1}(\xn)-G(\xn)g_2(\xn)\partiald{g_1}{\vv_1}(\xn)=(\alpha\cdot F(\xn)-G(\xn))g_2(\xn)\partiald{g_1}{\vv_1}(\xn) \;,
	\end{equation}
	and 
	\begin{equation}\label{eq:anf-v3}
		\partiald{(f-g)}{\vv_3}=F(\xn)f_4(\xn)\partiald{f_3}{\vv_3}(\xn)-G(\xn)g_4(\xn)\partiald{g_3}{\vv_3}(\xn)=(\beta\cdot F(\xn)-G(\xn))g_4(\xn)\partiald{g_3}{\vv_3}(\xn) \;.
	\end{equation}
	Let $\G_1,\G_{2\Delta+1}$ be a \ropinvgen[1]{} and a \ropinvgen[(2\Delta+1)]{}, respectively, such that $\G=\G_1+\G_{2\Delta+1}$. 
	\cref{thm:PITROPINV} and \cref{obs:derivCroanf} imply that $\left(g_2(\xn)\partiald{g_1}{\vv_1}(\xn)\right)\circ\G_{2\Delta+1}\neq 0$, so if $\alpha\cdot F(\G_{2\Delta+1})\neq G(\G_{2\Delta+1})$ then we get from Equation~\eqref{eq:anf-v1} that $\partiald{(f-g)}{\vv_1}\circ\G_{2\Delta+1}\neq 0$ and thus $(f-g)\circ\G\neq 0$ (using \cref{lem:indDerivLinear}). On the other hand, if $\alpha \cdot F(\G_{2\Delta+1})=G(\G_{2\Delta+1})$, then, since $\alpha\neq \beta$, a similar argument, relying on Equation~\eqref{eq:anf-v3}, shows that
	$\partiald{(f-g)}{\vv_3}\circ\G_{2\Delta+1}\neq 0$ and thus $(f-g)\circ\G\neq 0$, as claimed.
%
%
\end{proof}

\subsection{Reconstruction for \roanf}

In this section, we argue that the reconstruction algorithm of Gupta et al. \cite{gupta2014random}, when given oracle access to a polynomial $f\in\roanf$, w.h.p. successfully reconstructs an $\roanf$ formula computing $f$.
We do so by explaining why the different steps of their algorithm succeed w.h.p. on any input $f\in\roanf$. 
To ease the reading we give their algorithm (\texttt{AFR}) and its main subroutine (\texttt{LDR}) in the appendix (Algorithms~\ref{alg:afr} and \ref{alg:ldr}).
We remind  that their result, with minor changes, can be adapted to any large enough field, see \cref{rem:gupta2014-char}.

%
%

Before quoting the original result, we define the distribution on ANF formulas used in \cite{gupta2014random}. To this end, we define the \emph{universal} ANF:

\begin{definition}
Let $\Delta,n\in\n$. Let $\xn=(x_1,\ldots,x_n)$ and $\by=\{y_{i,j}:i=1,\ldots,4^\Delta,j=1,\ldots,n+1\}$ be formal variables. The \emph{universal $\Delta,n$ ANF}, denoted $\cU_{\Delta,n}(\xn,\by)$, is an ANF formula of product depth $\Delta$ in which leaf $i$ is labeled $\sum_{j=1}^n x_j y_{i,j}+y_{i,n+1}$.
\end{definition}

Trivially, for any ANF formula $f(\xn)$ of product depth $\Delta$ on $n$ variables, there exists an assignment $\vv\in\C^{(n+1)\cdot 4^\Delta}$ to the $\by$ variables of $\cU_{\Delta,n}(\xn,\by)$ such that $f(\xn)=\cU_{\Delta,n}(\xn,\vv)$. Given the number of variables $n$, the size $s=2\cdot 4^\Delta-1$ of the ANF we wish to sample, and a finite set of field elements $S\subseteq\C$, we define the distribution $\mathcal D(n,s,S)$ on ANF formulas by uniformly sampling an assignment $\vv$ from $S^{4^\Delta(n+1)}$. This is the distribution used in the main result of \cite{gupta2014random}:

\begin{theorem}[Theorem 1.1 of \cite{gupta2014random}]
Let $\f$ be a field of characteristic $0$ and $S$ be a finite subset of $\f$. Assume there is a black box holding an ANF formula $\Phi$ of size $s$ sampled from $\mathcal D(n,s,S)$, and $\Phi$ computes a polynomial $f\in\f[x_1,\ldots,x_n]$. There is a randomized algorithm that, given this black box, either outputs an ANF formula $\Phi'$ of size $\leq s$ computing $f$, or outputs \texttt{Fail}. The algorithm succeeds for a $(1-\frac{n^2\cdot s^{O(1)}}{|S|})$ fraction of the ANF formulas from $\mathcal D(n,s,S)$. Moreover, the running time of the algorithm is at most $(ns)^{O(1)}$.
\end{theorem}

We note that, although it is not mentioned in their main theorem, the output formula is unique up to \TS{}-equivalence, and this fact is stated when needed in intermediate results of \cite{gupta2014random} (recall \cref{fact:anf-sym}).
We prove \cref{thm:reconstructROANF} by going over the different steps of Algorithm~\ref{alg:afr}. We do not repeat all the arguments and claims of \cite{gupta2014random}, but rather give high level explanations, referring to theorems, algorithms and tools of \cite{gupta2014random}.

\begin{proof}[Sketch of proof of \cref{thm:reconstructROANF}]




We shall use the following notation in the proof. We wish to reconstruct  $f\in\roanf[n]$ that is computed by the ANF formula $\Phi$. We define the homogenization of $f$, $f^h$, as usual: $f^h(x_0,\ldots,x_n)=x_0^{\deg(f)}\cdot f\bra{x_1/x_0,\ldots,x_n/x_0}$.
Denote by $\bm A$ an $(n+1)\times (n+1)$ matrix of formal variables $a_{i,j}$. 
For $i\neq j\in\{r+1,r+2,\ldots,n\}$ we denote by $\bm{A^{i,j}_r}$ the matrix $\bm A$ where all columns except those indexed by $\{0,1,2,\ldots,r\}\cup\{i,j\}$ are set to zero (generic projection matrix to the variables $x_0,x_1,\ldots,x_r,x_i,x_j$). 
We denote by $A\in\C^{n\times n}$ an assignment to $\bm A$, and likewise $A^{i,j}_r$ would be an assignment to the $n\cdot(r+3)$ variables of $\bm{A^{i,j}_r}$. 
Note that $\croanf(\bm{A^{i,j}_r}\xn)$ is a \emph{universal} homogeneous $(r+3)$-variate ANF (in $\{x_0,x_1,\ldots,x_r,x_i,x_j\}$) in the sense that for every $(r+3)$-variate homogeneous ANF $f(x_0,x_1,\ldots,x_r,x_i,x_j)$, of depth $2\Delta$, there exists an assignment $A^{i,j}_r$ such that $f(\xn)=\croanf(A^{i,j}_r\xn)$. Finally, following \cite{gupta2014random}, we denote  $\sigma_{A^{i,j}_r}(f) \triangleq f^h(A^{i,j}_r \xn)$ (where now we think of $\xn$ as $\xn=(x_0,\ldots,x_n)$).

Looking at  Algorithm~\ref{alg:afr}, it is clear that except for 
Step \texttt{AFR}\ref{alg:afr3}, the  rest of the algorithm works without any assumptions on the input ANF. Hence, the proof of correctness boils down to proving that Step \texttt{AFR}\ref{alg:afr3} works w.h.p.; and more importantly, proving that the \texttt{LDR} algorithm (Algorithm~\ref{alg:ldr}, the subroutine invoked in Step \texttt{AFR}\ref{alg:afr3}) succeeds w.h.p. on random projections of \emph{any} $\roanf[n]$ instance.
Specifically, we need to prove that for \emph{any} $f\in\roanf[n]$, step \texttt{AFR}\ref{alg:afr3} succeeds with probability $\geq 1-\frac{|\Phi|^{O(1)}}{|T|}$ on a random linear projection  to $r+3=128$ variables (see \cref{rem:t}) of the \emph{homogenization} of $f$,  $f^h$ (where the coefficients of the projection are sampled from $T\subseteq\C$).



Gupta et al. define 
two conditions on internal nodes of an ANF $\cU_{\Delta,n}(\xn,v)$: \emph{formulaic independence} (FI, see \cref{defFI}) and \emph{pairwise singular independence} (PSI, see \cref{defPSI}). These conditions are defined in terms of dimensions of certain algebraic varieties $V_1,\ldots,V_k$.
In Lemmas 5.10, 5.11,
5.16 and 5.26 of their paper, they show that if every node of $\Phi$ satisfies FI, then the \texttt{LDR} algorithm  correctly reconstructs the polynomial computed at each node of $\Phi$ (up to an appropriate group of symmetries). Moreover, 
part (2) of their Lemma 5.16 shows that when a node $u$ of $\Phi$ satisfies FI and PSI, 
then the polynomials computed at the grandchildren of $u$ are computed up to 
$\TS{}$ equivalence.
Overall, this means that all the quadratic forms are computed correctly up to 
$\TS{}$-equivalence.

Thus, if the projected polynomials  $\sigma_{A^{i,j}_r}(f)$ that we compute in Step \texttt{AFR}\ref{alg:afr3} satisfy FI and PSI, 
then the algorithm will correctly reconstruct our $\roanf[n]$ formula.



To prove that (w.h.p.) $\sigma_{A^{i,j}_r}(f)$ satisfies FI and PSI, Gupta et al. prove that these conditions are captured by a set of polynomial equations. Intuitively, this is not a surprising result as  FI and PSI are algebraic conditions.

\begin{observation}\label{obs:goodVariety}\sloppy
For every $i,j\in\{r+1,r+2,\ldots,n\}$ there exists a set of nonzero polynomials ${p_1,\ldots,p_k\in\C[\bm{A^{i,j}_r}]}$ with the property that $\croanf(A^{i,j}_r\xn)$ satisfies FI and PSI if $A^{i,j}_r$ is not a point on the variety $V\left(p_1(\bm{A^{i,j}_r}),\ldots,p_k(\bm{A^{i,j}_r})\right)\triangleq\cbra{A^{i,j}_r \mid p_1({A^{i,j}_r})=\ldots= p_k({A^{i,j}_r})=0}$. Furthermore, the degree of each $p_i$ is $2^{O(\Delta)}$, which is polynomial in the size of the formula.
\end{observation}

This observation is not stated as is in \cite{gupta2014random} but it can be immediately deduced from the proofs of Corollaries 5.31 and 5.32 of \cite{gupta2014random}.

Thus, we wish to show that a random ${A^{i,j}_r}$ does not belong to the variety defined in \cref{obs:goodVariety}. For this we follow the same approach as Gupta et al. We prove that \emph{there exist} 
good projections $A^{i,j}_r$ that do not belong to the variety, and then using \szlem{} we  conclude that such a random projection is not on the variety.

\begin{claim}\label{clmLDRsucceeds}
Let $r\geq 125$ and $n\geq r$. For any $n$-variate $f\in\roanf[n]$, computed by the ANF formula $\Phi$, and any $i,j\in\{r+1,r+2,\ldots,n\}$, there exists some projection $A^{i,j}_r$ such that $\sigma_{A^{i,j}_r}(f)$ satisfies FI and PSI at \emph{every internal node} of $\Phi$.
\end{claim}

\begin{proof} 
	To prove the existence of a ``good'' projection for an arbitrary $f\in\roanf[n]$, we use an explicit ANF $g$, on $128$ variables, that can be described as a projection of \emph{any} $f\in\roanf[n]$ (more accurately, of $f^h$). 
	The definition of $g$ comes from the proof of Lemma 5.30 of \cite{gupta2014random}:
	\begin{align}
		\forall i,j\in[4]:\;\;\;g_{i,j}(\xn)\triangleq&(x_{32(i-1)+8(j-1)}^e+x^e_{32(i-1)+8(j-1)+1})\cdot(x^e_{32(i-1)+8(j-1)+2}+x^e_{32(i-1)+8(j-1)+3})\nonumber\\
		&+(x^e_{32(i-1)+8(j-1)+4}+x^e_{32(i-1)+8(j-1)+5})\cdot(x^e_{32(i-1)+8(j-1)+6}+x^e_{32(i-1)+8(j-1)+7})\nonumber\\
		\forall i\in[4]:\;\;\;g_i(\xn)\triangleq&g_{i,1}(\xn)g_{i,2}(\xn)+g_{i,3}(\xn)g_{i,4}(\xn)\label{giEqDef}\\
		g(\xn)\triangleq&g_1(\xn)g_2(\xn)+g_3(\xn)g_4(\xn)\label{gEqDef}.
	\end{align}
	
	The exponent $e\in\n$ is chosen such that the degree of $g$ is $2^\Delta$ for the given $\Delta$, i.e. $e=2^{\Delta-3}$. Gupta et al. prove that $g$ satisfies PSI in Lemma 5.30. In Lemma 5.29, the FI condition is proven to hold for a slightly different polynomial (specifically, they prove $g_i$ as defined in equation \eqref{giEqDef} satisfies FI), but the proof for formulaic independence of $g$ itself works exactly the same (relies on variable-disjointness of $g_1,\ldots,g_4$), so we get:
	
	\begin{fact}\label{gIsGood}
		The polynomial $g$ defined in Equation \eqref{gEqDef} satisfies FI and PSI (and so does $g(x_{\pi(0)},\ldots,x_{\pi(127)})$, for any permutation $\pi$).
	\end{fact}

	Let $g(\xn)$ be as defined in equation \eqref{gEqDef} above. Our goal here is, given an unknown $f\in\roanf[n]$ and indices $i,j\in[n]$, to prove there exists some projection $A^{i,j}_r$ such that $\sigma_{A^{i,j}_r}(f)=g(\xn)$ (possibly up to a permutation of the variables); as we only care about  projections up to permutations of the variables, we can WLOG assume $i=r+1,j=r+2$. The correctness of Algorithm~\ref{alg:ldr} is proven for a number of variables $\geq 128$ and $g$ is a $128$-variate polynomial, so for sake of simplicity we may assume $r=125$ such that projections of $f^h$ have the same number of variables as $g$.
	
	For an ANF $\Psi$ computing $g$ such that each leaf is labeled by a single variable from $\{x_1,\ldots,x_{128}\}$ (times some constant), denote by $\tilde\Psi$ a new formula constructed as follows: for every $i\in[4^\Delta]$, if leaf number $i$ in $\Psi$ is labeled $\alpha_i\cdot x_j$, relabel it to $\alpha_i \cdot x_j+\ell_i(\xn)$, where $\ell_i$ is some linear form depending on the variables $x_{129},\ldots,x_n$. Choose the coefficients of the $\ell_i$s so that all the leaves of $\tilde\Psi$ are linearly independent (thus, $\tilde\Psi(\xn)\in\roanf[n]$). As $f^h$ and $\tilde\Psi$ are two polynomials in the $\GL[n+1]$-orbit of \croanf{}, there exists some $B\in\GL[n+1]$ such that $f^h(B\xn)=\tilde \Psi(\xn)$, and by construction $\restr{\tilde \Psi}{x_{129}=0,x_{130}=0,\ldots,x_n=0}(\xn)=\Psi(\xn)=g(\xn)$. By defining $A^{i,j}_r$ to be the matrix $B$ with columns $129,\ldots,n$ set to zero, we get $\sigma_{A^{i,j}_r}(f) =\restr{\tilde \Psi}{x_{129}=0,x_{130}=0,\ldots,x_n=0}(\xn)=g(\xn)$. Since $A^{i,j}_r$ is a projection, this is what we wanted to prove.
\end{proof}


Thus, by applying the \szlem{}, we can conclude that a random projection (sampled from a set $T\subseteq\C$) 
of the homogenization of any $f\in\roanf[n]$ satisfies FI and PSI with probability at least $\bra{1-\frac{|\Phi|^{O(1)}}{|T|}}$, thanks to the upper bound on the degree of the $p_i$s of \cref{obs:goodVariety}. For Step \texttt{AFR}\ref{alg:afr3} to work, we need all $n^2$ projections to yield ``good'' polynomials, and by a simple application of the union bound we  deduce that \texttt{AFR3} succeeds with probability at least $\bra{1-\frac{n^2\cdot |\Phi|^{O(1)}}{|T|}}$.

This completes the proof of \cref{thm:reconstructROANF}
\end{proof}

\begin{remark}
The original theorem of \cite{gupta2014random} uses two sets of field elements: the set $S$, used to sample random ANFs from the distribution $\mathcal D(n,s,S)$, and the set $T$, used to sample random projections $A^{i,j}_r$ of the input ANF. As their algorithm works for any $f\in \roanf[n]$, we do not need the set $S$. Thus,we only use $T$, and we add run-time dependence on $\log(|T|)$ so we can sample the uniform distribution on $T$.
\end{remark}

\section{Dense orbits for \dthree{} circuits}\label{sec:sps}

In this section we  prove our claims regarding dense orbits in $\dthree$. We start by proving \cref{thm:denseInDepthThree} regarding the relation between $\taff{\gla{}{\f}}$, $\spinv{\f}$ and $\dthree$.

\begin{proof}[Proof of \cref{thm:denseInDepthThree}]
	The claim regarding the closures follows immediately from the fact that every matrix can be approximated by invertible matrices and from the simple observation that for any $n$-variate polynomial $f(\xn)\in\Sigma^{[s]}\Pi^{[d]}\Sigma(\f)$, there exist $A\in\f^{n\times n},\vb\in\f^n$ such that $\sdmpol(A\xn+\vb)=f(\xn)$.
	
	To prove the separation we first note that the polynomial $f(\xn)=x_1^2$ is in $\dtwo$, but not in $\taff{\gla{}{\f}}$: if $f(\xn)\in \taff{\gla{}{\f}}$, then there exists $(A,\vb)\in\gla{n}{\f}$ such that $f(A\xn+\vb)=\sdmpol[s][d]$, for some $s$ and $d$ (as we compose with invertible affine maps). However,  $f(A \xn+\vb)=(\ell(\xn))^2$ for some non-constant linear function $\ell(\xn)$, which is obviously not a multilinear polynomial.
	The second separation will follow from the next simple claim.
	
	\begin{claim}\label{cla:sdminvHom}
		If $f\in\spinv{\f}$ is $d$-homogeneous, then it is in the $\gl{n}{\f}$ orbit of some $d$-homogeneous $\dtwo$ circuit (i.e. no affine translation is needed). 
	\end{claim}
	\begin{proof}
		Let $(A,\vb)\in\gla{n}{\f}$ and let $\Psi$ be a $\dtwo$ circuit such that $f(\xn)=\Psi(A\xn+\vb)$. 
		Observe that for every $i$ it holds that $\Psi(\xn)^{[i]}\neq 0$ if and only if $\Psi(A\xn)^{[i]}\neq 0$, since $A$ is invertible. In particular, if $\Psi(\xn)$ had a monomial of degree larger than $d$ then the degree of $f(\xn)=\Psi(A\xn+\vb)$ would have been larger than $d$ in contradiction. Thus, all gates in $\Psi$ have degree at most $d$. Similarly, we now see that $f(\xn)=\bra{\Psi(Ax+\vb)}^{[d]}=\bra{\Psi(\xn)}^{[d]}(A\xn)$. Thus, $\Psi^{[d]}$ is the claimed $\dtwo$ circuit.
	\end{proof}
	Let $\sigma_d(\xn)$ be the $n$th elementary symmetric polynomial. I.e. the sum over all degree-$d$ multilinear monomials in $n$-variables.  Theorem 0 of \cite{nisan1996lower}
	shows that any homogeneous $\dthree$ circuit computing $\sigma_d$ must have size $\Omega(\frac{n}{2d})^d$. 
	As any homogeneous polynomial in $\dtwo^{\gl{n}{\f}}$ can be computed by a homogeneous $\dthree$ circuit of the same complexity, 
	we get an exponential lower bound on the sparsity of any $\spinv{\f}$ circuit computing $\sigma_d$, over any field. 
	To get an upper bound on the $\dthree$ complexity, note that, over any field of size $|\f|\geq n+1$,  $\sigma_d$ has a $\dthree$ circuit of size $O\bra{n^2}$ (see \cite{shp18}), that is obtained by interpolating the polynomial $f(Y)=\prod_{i=1}^{n}(Y+x_i)$. 
\end{proof}


We devote the rest of this section to proving \cref{thm:reconstructSDM,thm:PITsumSDMinv,thm:PITSINV}.

\subsection{A hitting-set generator for $\dtwo^{\gla{}{\f}}$ circuits}

%
%
%

In this section, we prove \cref{thm:PITSINV}. The main idea is that given 
some $f\in \dtwo^{\gla{}{\f}}$, where $f(\xn)=g(A\xn+\vb)=g(\ell_1(\xn),\ldots,\ell_n(\xn))$ for an $s$-sparse polynomial $g$, composing $f$ with a \ropinvgen[1]{} allows us to ``halve'' the number of monomials appearing in the underlying $\dtwo$ circuit $g(x)$. 
Depending on the structure of $g$, this can be done by either taking a derivative of $f$ at the direction of an appropriately chosen dual vector, or by restricting $f$ to a linear subspace in which some $\ell_i(\xn)= 0$ and other linear functions remain linearly independent. By \cref{lem:indDerivLinear,lem:indProjectZero}, both tasks can be simulated using a $1$-independent generator.

As a reminder, we restate \cref{thm:PITSINV} before giving its proof.


\PITSINV*
\begin{proof}
	By induction on $t$. For $t=0$, $0\neq f(\xn)$ is either a non-zero constant, or a product of non-zero linear functions. A non-zero linear function composed with a \ropinvgen[1]{} $\G$ is non-zero because the $n$ entries of $G$ are linearly independent (\cref{obs:kwise}(\ref{item:coordsGenInd})), so $f\circ \G\neq 0$.
	
	Let $t>0$ and let $\G_1(\bm{y_1},z_1)$ and $\G_t(\bm{y_2},z_2,\ldots,z_{t+1})$ be a \ropinvgen[1]{} and a \ropinvgen[t]{}, respectively, such that $\G=\G_1+\G_t$. Let $\ell_1,\ldots,\ell_n$ be linear functions such that the $i$th coordinate of $A\xn$ is $\ell_i(\xn)$, and let $\vb=(b_1,...,b_n)$.
	
	First, we note that WLOG we can assume that no variable $x_i$ divides $g$; otherwise we can take some $\tilde g\in\f[\xn]$ such that $g(\xn)=x_i^k\tilde g(\xn)$, $x_i$ does not divide $\tilde g$ and both $g$ and $\tilde g$ have the same sparsity. By the base case (sparsity $1$), $(\ell_i(\xn)+b_i)^k\circ \G\neq 0$, so $f\circ \G\neq 0$ if and only if $ \left(\tilde g(A\xn+\vb)\right)\circ \G\neq 0$. 
	
	Now that we know $g(\xn)$ is not divisible by any variable, we consider two cases:
	
	\textbf{Case 1:} There exists a variable $x_i\in\text{var}(g)$ that appears in $\leq 2^{t-1}$ monomials of $g(\xn)$. Choose $\vv\in\f^n$ such that $\ell_i(\vv)=1$, and for all $j\neq i$, $\ell_j(\vv)=0$. By \cref{lem:dualder}, $\partiald{f}{\vv}(\xn)=\bra{\partiald{g}{x_i}}\bra{A\xn+\vb}$. By choice of $x_i$, $\partiald{g}{x_i}$ is non-zero and of sparsity $\leq 2^{t-1}$, so by induction: $\bra{\partiald{f}{\vv}}(\G_t)\neq 0$.  \cref{lem:indDerivLinear} implies that $f\circ \G=f\circ (\G_1+\G_t)\neq 0$.
	
	\textbf{Case 2:} Every variable $x_i \in \text{var}(g)$ appears in at least $2^{t-1}$ monomials of $g$. Assume, WLOG, that  $x_1\in\text{var}(g)$, and define $\tilde g(\xn)\triangleq g(0,x_2,x_3,\ldots,x_n)$. As $x_1$ does not divide $g$, $\tilde g\neq 0$ and is of sparsity $\leq 2^{t-1}$. By \cref{lem:indProjectZero}, there exist linearly independent linear functions $\tilde\ell_2,\ldots,\tilde\ell_n$, an assignment $\valpha\in\f^{|\bm{y_1}|}$ and some linear function $L(\xn)$ such that $f\bra{\xn+\G_1\bra{\valpha,L(\xn)}}=\tilde g\bra{\ell_2(\xn),\ldots,\ell_n(\xn)}\neq 0$. As $\tilde g$ is non-zero and has sparsity $\leq 2^{t-1}$, we get from the induction hypothesis that $f\bra{\xn+\G_1\bra{\valpha,L(\xn)}}\circ \G_t\neq 0$, and therefore $f\bra{\xn+\G_1\bra{\bm{y_1},z_1}}\circ \G_t\neq 0$. Hence, $f\circ \G=f\circ\bra{\G_t+\G_1}=f\bra{\xn+\G_1\bra{\bm{y_1},z_1}}\circ \G_t\neq 0$.
\end{proof}

\cref{cor:hs-sparse} follows immediately from \cref{thm:PITSINV} and \cref{obsHitSetGen}.

\subsection{An interpolating set generator for $\taff{\gla{}{\f}}$}


To construct an interpolating set generator for $\taff{\gla{n}{\f}}\triangleq \taff{\gla{}{\f}}\cap\f[x_1,\ldots,x_n]$ we need a generator that hits the difference of two polynomials of $\taff{\gla{n}{\f}}$. As this class is closed under multiplication by scalars, such a generator hits every nonzero sum of two $\taff{\gla{n}{\f}}$ polynomials. The main idea can be described as follows: the  tensor $\sdmpol$ on variables $\{x_{1,1},\ldots,x_{s,d}\}$ has the property that  for any two variables in distinct product gates, $x_{i,j}$ and $x_{i',j'}$ ($i\neq i'$), it holds that $\partiald{^2\sdmpol}{x_{i,j}\partial x_{i',j'}}= 0$. 
We prove that for a sum of distinct $\taff{\gla{n}{\f}}$ polynomials, there is always a pair of ``dual'' vectors such that if we take a derivative in their direction then one of the $\taff{\gla{n}{\f}}$ polynomials of the sum vanishes. Once we prove this, all that is left is to hit the remaining polynomial (or actually, its derivative).

If $f\in\taff{\gla{n}{\f}}$ and $\vv\in\f^n$ is arbitrary, then $\partiald{f}{\vv}$ need not be in $\taff{\gla{n}{\f}}$. We thus begin by constructing a hitting set generator for  directional derivatives of $\taff{\gla{n}{\f}}$ polynomials.

\begin{lemma}\label{PITderivSDMINV}
	Let $f\in\taff{\gla{n}{\f}}$,  $k\in\n$ and $\vv_1,\ldots,\vv_k\in\f^n$. Then, for any \ropinvgen[(k+2)]{} $\G$:
	$$ \partiald{^kf}{\vv_1\partial\vv_2\cdots\partial\vv_k}\neq 0\Rightarrow \partiald{^kf}{\vv_1\partial\vv_2\cdots\partial\vv_k}\circ \G\neq 0 \;.$$
\end{lemma}
\begin{proof}
	Let $\G_1^{(1)},\G_1^{(2)},\G_k$ be a pair of \ropinvgen[1]{}s and a \ropinvgen[k]{}, respectively, such that $\G=\G_1^{(1)}+\G_1^{(2)}+\G_k$. 
	Let $\{\ell_{1,1},\ldots,\ell_{s,d}\}$ be linearly independent linear functions such that $f(\xn)=\sum_{i=1}^{s}\prod_{j=1}^{d}\ell_{i,j}$. Let $\cbra{\vu_{i,j}}$ be a dual set to $\cbra{\ell^{[1]}_{i,j}}$. I.e., $\ell^{[1]}_{i,j}(\vu_{i',j'})=\delta_{i,i'}\cdot \delta_{j,j'}$.
	
	\fussy
	Set $g(\xn)\triangleq\partiald{^kf}{\vv_1\partial\vv_2\cdots\partial\vv_k}(\xn)$. 
	For every $i$, let $Q_i(w_{i,1},\ldots,w_{i,d})$ be a polynomial satisfying 
	$Q_i(\ell_{i,1}(\xn),\ldots,\ell_{i,d}(\xn))=\partiald{^k \bra{\prod_{j=1}^{d}\ell_{i,j}(\xn)}}{\vv_1\partial\vv_2\cdots\partial\vv_k}$.
	In particular, $g(\xn)=\sum_{i=1}^sQ_i(\ell_{i,1},\ldots,\ell_{i,d})$.
	Fix some $i\in[s]$ such that $Q_i$ is non-constant (if no such $i$ exists, then $g$ is a non-zero constant and thus $g\circ \G\neq 0$). Assume, WLOG, that $Q_i$ depends non-trivially on $w_{i,1}$ and consider the derivative  in direction $\vu_{i,1}$.
	From \cref{lem:dualder} We get
	$$ \partiald{Q_i(\ell_{i,1}(\xn),\ldots,\ell_{i,d}(\xn))}{\vu_{i,1}}=\partiald{Q_i}{w_{i,1}}\bra{\ell_{i,1}(\xn),\ldots,\ell_{i,d}(\xn)}\neq 0 \;, $$
	and for $i'\neq i$
	$$ \partiald{Q_{i'}(\ell_{i,1}(\xn),\ldots,\ell_{i,d}(\xn))}{\vu_{i,1}}=\partiald{Q_{i'}}{w_{i,1}}(\ell_{i,1}(\xn),\ldots,\ell_{i,d}(\xn))=0\;. $$
	Thus
	$$ \partiald{g}{\vu_{i,1}}=\partiald{Q_i}{w_{i,1}}(\ell_{i,1}(\xn),\ldots,\ell_{i,d}(\xn)) \neq 0 \;. $$
	As $Q_i(\ell_{i,1}(\xn),\ldots,\ell_{i,d}(\xn))$ is a $k$th order directional derivative of the product $\ell_{i,1}(\xn)\cdots\ell_{i,d}(\xn)$ we have that
	$$ Q_i(\ell_{i,1}(\xn),\ldots,\ell_{i,d}(\xn))=\sum_{\substack{S\subseteq[d]\\|S|=k}}\alpha_S\left(\prod_{j\in [d]\setminus S}\ell_{i,j}(\xn)\right) \;,$$
	for some constants $\alpha_S\in\f$. Thus, 
	$$ \partiald{g}{\vu_{i,1}}=\sum_{\substack{S\subseteq\{2,\ldots,d\}\\|S|=k}}\alpha_S
	\left(\prod_{j\in\{2,\ldots,d\}\setminus S}\ell_{i,j}(\xn)\right)\;. $$
	Assume, WLOG, that for $T=\{2,\ldots,k+1\}$, $\alpha_T\neq 0$.
	Observe that except for 
	the term $\alpha_T\left(\prod_{j\in\{2,\ldots,d\}\setminus T}\ell_{i,j}(\xn)\right)$, every other term is divisible by one of the functions $\ell_{i,j}$, for $j\in T$. Let $V=\{\vv \mid \ell_{i,j}(\vv)=0, \; \forall j\in T\}$. It follows that  $\restr{\partiald{g}{\vu_{i,1}}}{V} = \restr{\alpha_T\left(\prod_{j\in\{2,\ldots,d\}\setminus T}\ell_{i,j}(\xn)\right)}{V}\neq 0$.
	\cref{lem:indProjectZero} implies that there exist linear functions $L_1(\xn),\ldots,L_k(\xn)$ and an assignment $\vbeta$ such that for $\vL=\bra{L_1,\ldots,L_k}$: $$\partiald{g}{\vu_{i,1}}\bra{\xn+\G_k(\vbeta,\vL(\xn))} = \alpha_T\left(\prod_{j\in\{2,\ldots,d\}\setminus T}\ell_{i,j}\bra{\xn+\G_k(\vbeta,\vL(\xn))}\right) \neq 0\;.$$ 
	As the right term is a product of linear functions, we get from \cref{obs:kwise}(\ref{item:coordsGenInd})  that $$\partiald{g}{\vu_{i,1}}\bra{\xn+\G_k(\vbeta,\vL(\xn))}\circ \G_1^{(2)}\neq 0\,.$$
	Therefore, $\partiald{g}{\vu_{i,1}}\circ (\G_1^{(2)}+\G_k)\neq 0$. The claim now follows from \cref{lem:indDerivLinear}. 
\end{proof}

It is not hard to see that the proof above implies the following hitting set generator for $\taff{\gla{}{\f}}$:

\begin{corollary}\label{cor:PITSDMINV}
	If $0\neq f\in\taff{\gla{n}{\f}}$, then for any \ropinvgen[2]{} $\G$: $f\circ \G\neq 0$.
\end{corollary}

%

We are now prepared to a construct hitting set generator for $\taff{\gla{n}{\f}}+\taff{\gla{n}{\f}}$. We recall the statement of  \cref{thm:PITsumSDMinv}.

\PITsumSDMinv*
\begin{proof}
	Let $\G_6$ be a uniform $6$-independent polynomial map and 
	let $\cbra{\ell_{i,j,k}}$ be linear functions such that $f_i = \sdmpol[s_i][d_i]\bra{\ell_{i,1,1},\ldots,\ell_{i,s_i,d_i}}$.
	
	We first prove that  if $f\circ \G_6=0$ then $d_1=d_2$. 
	Assume for a contradiction that $d_1>d_2$. Observe that $f_1^{[d_1]}= \sdmpol[s_1][d_1]\bra{\ell_{1,1,1}^{[1]},\ldots,\ell_{1,s_1,d_1}^{[1]}}$ (recall that $\ell^{[1]}$ is the degree $1$ homogeneous part of $\ell$). As the $\ell_{1,i,j}^{[1]}$s are linearly independent, it follows that $f_1^{[d_1]}\neq 0$.  \cref{cor:PITSDMINV} implies that $f_1^{[d_1]}\circ \G_6\neq 0$, and as $\G_6$ is uniform, we get that $\deg\bra{f_1\circ \G_6}=d_1\cdot \deg\bra{\G_6}$. On the other hand,  $\deg\bra{f_2\circ \G_6}\leq d_2\cdot \deg\bra{\G_6}< \deg\bra{f_1\circ \G_6}$. It follows that $f\circ \G_6\neq 0$, in contradiction. From now on we denote $d=d_1=d_2$.
	
	Next, we note that we can  assume that $f$ is homogeneous. Let $\tell_{i,j,k}=x_0\cdot \ell_{i,j,k}(\xn/x_0)$ be the homogenization of $\ell_{i,j,k}$.
	Observe that the homogenization of $f$ is  $\tilde{f}(x_0,\xn)\triangleq x_0^d f(\xn/x_0)= \sdmpol[s_1][d]\bra{\tell_{1,1,1},\ldots,\tell_{1,s_1,d_1}}+ \sdmpol[s_2][d]\bra{\tell_{2,1,1},\ldots,\tell_{2,s_2,d_2}}$, which is a homogeneous polynomial in $\taff{\gl{n+1}{\f}}+\taff{\gl{n+1}{\f}}$.
	By \cref{lem:uniIndGenHitsHom},  it is enough to prove that $\tilde{f}\circ \G'_6 \neq 0$, where $\G'_6$ is a uniform $6$-independent map into $\f^{n+1}$. Hence, to simplify notation and WLOG, we assume from now on  that $f$ is homogeneous and that $\ell_{i,j,k}=\ell_{i,j,k}^{[1]}$. Next, we handle the case  $s_1\neq s_2$. 

	

	Assume, WLOG, that  $s_1>s_2$.  
	As the $s_1\cdot d$ linear functions $\cbra{\ell_{1,i,j}}_{i,j}$ are linearly independent, there must exist a linear form, WLOG, $\ell_{1,1,1}$, such that $\ell_{1,1,1}\not\in \text{span}\bra{\cbra{\ell_{2,i,j}}_{i,j}}$.  
	As before, fix a vector $\vv$ such that $\ell_{1,1,1}(\vv)=1$ and $\ell_{2,i,j}(\vv)=0$ for all $i,j\in[s_2]\times [d]$.  \cref{lem:dualder} implies that $\partiald{f_2}{\vv}= 0$. On the other hand, from linear independence we get that $\frac{\partial \bra{\prod_{j=1}^{d}\ell_{1,1,j}} }{\partial \vv}\neq 0$ and,  the same argument also gives $\partiald{f_1}{\vv}\neq  0$. Thus $\partiald{f}{\vv}\neq 0$.
	From  \cref{lem:indDerivLinear,PITderivSDMINV} we conclude that any uniform \ropinvgen[4]{} hits $f$. Observe that the proof above also shows that it must be the case that $\text{span}\bra{\cbra{\ell_{1,i,j}}_{i,j}}=\text{span}\bra{\cbra{\ell_{2,i,j}}_{i,j}}$, or else any uniform \ropinvgen[4]{} hits $f$.
	
	From this point on, we assume that $s_1=s_2=s$ and that $\text{span}\bra{\cbra{\ell_{1,i,j}}_{i,j}}=\text{span}\bra{\cbra{\ell_{2,i,j}}_{i,j}}$. 
	
	As $\text{span}\bra{\cbra{\ell_{1,i,j}}_{i,j}}=\text{span}\bra{\cbra{\ell_{2,i,j}}_{i,j}}$, we can represent $f_2$ as a polynomial in $\cbra{\ell_{1,i,j}}_{i,j}$ (recall this notion from Section~\ref{sec:not}). We split the proof into two cases, depending on the $\cbra{\ell_{1,i,j}}_{i,j}$-monomials appearing in $f_2$:
	
	
	\begin{enumerate}
		\item The set of $\cbra{\ell_{1,i,j}}_{i,j}$-monomials appearing in $f_2$ is a subset of the $\cbra{\ell_{1,i,j}}_{i,j}$-monomials in $f_1$. I.e., $f_2(\xn)=\sum_{i=1}^{s}\alpha_i \cdot \prod_{j=1}^{d}\ell_{1,i,j}$. 
		This means that $f=\sum_{i=1}^{s}(1+\alpha_i) \cdot \prod_{j=1}^{d}\ell_{1,i,j} \in \sdmpol[s][d]^{\gl{n}{\f}}$, and the theorem follows from \cref{cor:PITSDMINV}.
		
		\item There exists an $\cbra{\ell_{1,i,j}}_{i,j}$-monomial $\prod_{i,j}\ell_{i,j}^{a_{i,j}}$ 
		in $f_2$ that is not an $\cbra{\ell_{1,i,j}}_{i,j}$-monomial of $f_1$. 
		Let $\cbra{\vv_{i,j}}$ be a dual set to $\cbra{\ell_{1,i,j}}$. We proceed to show we can choose two vectors 
		$\vu,\vw\in\{\vv_{1,1},\ldots,\vv_{s,d}\}$ 
		such that $\partiald{^2f_1}{\vu\partial\vw}=0$ and $\partiald{^2f_2}{\vu\partial\vw}\neq 0$. We again consider two cases:
		\begin{itemize}
			\item There exists some $a_{i,j}\geq 2$: 
			Let $\vu=\vw =\vv_{i,j}$. By \cref{lem:dualder}:
			$$ \partiald{^2f_1}{\vu\partial\vw}(\xn)=\partiald{\sdmpol}{^2x_{i,j}}(\ell_{1,1,1},\ldots,\ell_{1,s,d})=0 $$
			and
			$$ \partiald{^2f_2}{\vu\partial\vw}(\xn) \neq 0\;,$$
			as 
			the $\cbra{\ell_{1,i,j}}$-monomial $\prod_{i,j}\ell_{i,j}^{a_{i,j}}$ exists in $f_2$. 
			
			\item $a_{i,j}\leq 1$ for every $i,j$: In this case, since $f_2$ is homogeneous, there must be some $i\neq i'$ such that for some $j$ and $j'$, $a_{i,j},a_{i',j'}\neq 0$. 
			Now choose $\vu = \vv_{i,j}$ and $\vw = \vv_{i',j'}$. As before, it is easy to verify that 
			$$\partiald{^2f_1}{\vu\partial\vw}(\xn)=0 \quad \text{and} \quad \partiald{^2f_2}{\vu\partial\vw}(\xn)\neq 0\;.$$
			%
		\end{itemize}
		Thus, in either cases, there exist $\vu,\vw$ such that
		$$ \partiald{^2f}{\vu\partial\vw}=\partiald{^2f_2}{\vu\partial\vw}\neq 0. $$
		
		By \cref{PITderivSDMINV}, any \ropinvgen[4]{} hits $\partiald{^2f}{\vu\partial\vw}$; so by \cref{lem:indDerivLinear}, any uniform \ropinvgen[6]{} hits $f$.
	\end{enumerate}
\end{proof}

\subsection{Reconstruction of \sdminv{} circuits}\label{secReconD3}

In  \cite{kayal2019reconstruction}, Kayal and Saha gave a polynomial-time, randomized reconstruction algorithm that, given black-box access to a homogeneous \dthree{} circuits satisfying a \emph{non-degeneracy} condition (\cref{nonDegenKayal}), reconstructs the circuit with high probability.
To prove \cref{thm:reconstructSDM} all we have to do is show that any homogeneous polynomial $f\in \taff{\gla{}{\f}}$
satisfies the non-degeneracy condition of \cref{nonDegenKayal}.

To explain the condition we first need to define the \emph{partial derivative} space of a polynomial:


\begin{definition}
	For an $n$-variate polynomial $f(\xn)\in\f[\xn]$, of degree $d$, and for any $k\in[d]$, the \emph{partial derivative space of order $k$} of $f$ ($\text{PD}_k$ space for short),  denoted $\partial^kf$, is the $\f$-span of all  partial derivatives of $f$ of  order $k$:
	$$ \partial^kf=\text{\normalfont span}_\f\left\{\partiald{^kf}{x_{i_1}\partial x_{i_2}\cdots\partial x_{i_k}}:i_1,\ldots,i_k\in[n]\right\} \;.$$
\end{definition}


\begin{definition}[Non-degeneracy condition \cite{kayal2019reconstruction}]\label{nonDegenKayal}
	Let $f(\xn)=f_1(\xn)+\ldots+f_s(\xn)$, where $f_i=\prod_{j=1}^{d}\ell_{i,j}$ for some linear forms $\ell_{i,j}$, be an $n$-variate $d$-homogeneous polynomial, which can be computed by a depth-$3$ circuit of top fan-in $s$. Fix $k\triangleq \ceil{\frac{\log(s)}{\log(\frac{n}{{\mathsf e}\cdot d})}}$, where ${\mathsf e}$ is the base of the natural logarithm. We say $f(\xn)$ is \emph{non-degenerate} if $\dim(\partial^kf)=s\cdot{d\choose k}$, and for every $i\in[s]$ there exist $2k+1$ linear forms $\ell_{i,r_1},\ldots,\ell_{i,r_{2k+1}}$ such that:
	$$ \dim\left(\partial^k\left(\sum_{j\in[s]\setminus\{i\}}f_j\right)\;\;\text{\normalfont mod}\;\;\text{\normalfont span}_\C\{\ell_{i,r_1},\ldots,\ell_{i,r_{2k+1}}\}\right)=(s-1)\cdot{d\choose k} $$
\end{definition}

\begin{theorem}[Theorem 1 of \cite{kayal2019reconstruction}]\label{ksThm1}
	Let $n,d,s\in\n$, $n\geq(3d)^2$ and $s\leq(\frac{n}{3d})^{\frac{d}{3}}$. Let $\f$ be a field of characteristic zero or greater than $ds^2$.\footnote{This requirement appears before the statement of their theorem.}  There is a randomized, $\text{poly}(n,d,s)=\text{poly(n,s)}$ time algorithm which takes as input black-box access to an $n$-variate $d$-homogeneous polynomial $f$ that can be computed by a non-degenerate (\cref{nonDegenKayal}) $\dthree$ circuit of top fan-in $s$, and outputs a non-degenerate, $n$-variate, $d$-homogeneous $\dthree$ circuit of top fan-in $s$ computing $f$.
\end{theorem}

For our proof we will need the following simple fact.

\begin{fact}\label{claimSpan}
	Let $f(\xn)$ be a polynomial of degree $d$ and $(A,\vb)\in\gla{n}{\f}$. Then, for any $k\in[d]$:
	$$ \partial^kf(A\xn+\vb)=\left\{g(A\xn+\vb):g\in\partial^k f(\xn)\right\} \;.$$
	
\end{fact}


\begin{proof}[Proof of \cref{thm:reconstructSDM}]
	
	As given a non-homogeneous $\taff{\gla{n}{\f}}$ circuit we can easily get query access to its homogenization, $f^h=x_0^df(\frac{x_1}{x_0},\ldots,\frac{x_n}{x_0})$, which is a homogeneous polynomial in $\taff{\gl{n+1}{\f}}$, we can assume WLOG that the black-box polynomial  is homogeneous. It should also be  clear that a polynomial satisfies  the condition in \cref{nonDegenKayal} if and only if its homogenization does.

	It is clear that $\dim\left(\partial^k\sdmpol\right)=s{d\choose k}$, and since composing with an invertible linear transformation does not affect the dimension of the PD$_k$ space (\cref{claimSpan}), it follows that  $\dim\left(\partial^kf\right)=s{d\choose k}$ for \emph{any} $d$-homogeneous, $s$-sparse $f\in \taff{\gl{n}{\f}}$. 
	It is also clear that $\sdmpol$ satisfies the second condition and that this condition too is invariant under invertible linear transformations.

	We still need to argue that the output of the algorithm of \cref{ksThm1} is a $\taff{\gl{}{\f}}$ circuit. \cref{ksThm1}  guarantees that the output circuit $\Phi=\sum_1^s\prod_1^d\ell_{i,j}$ is a non-degenerate $d$-homogeneous, $\dthree$ circuit computing $f$. We claim the linear forms $\ell_{i,j}$ on the leaves are linearly independent, and conclude that it is indeed a  $\taff{\gl{}{\f}}$ circuit.
	Indeed, as $f(\xn)$ is $\gl{n}{\f}$-equivalent to $\sdmpol(\xn)$ and $\partial^{d-1}\sdmpol(\xn)=\text{span}_\f\{x_{1,1},\ldots,x_{s,d}\}$, it follows that $\partial^{d-1}\Phi$ has dimension $s\cdot d$. The space $\partial^{d-1}\Phi$ is contained in $\text{span}_\f\{\ell_{1,1},\ldots,\ell_{s,d}\}$, so by dimension argument the set $\{\ell_{1,1},\ldots,\ell_{s,d}\}$ must be linearly independent.
	
	Finally, we note that by \cref{sdmSymmetry} the representation that was found is unique up to $\TPS[s,d][\f]$-equivalence. 
	
	This concludes the proof of \cref{thm:reconstructSDM}.
\end{proof}





%
%
%

%

\bibliographystyle{alpha}
\bibliography{main}

\appendix

\section{The reconstruction algorithm of \cite{gupta2014random}}


For Algorithm~\ref{alg:afr} we introduce the following notation. Given integers $0<r<i<j\leq n$ we denote $\xn_{r,i,j}\triangleq (x_0,\ldots,x_r,0,\ldots,0,x_i,0,\ldots,0,x_j,0,\ldots,0)$ a vector of variables of length $n+1$. To be consistent with  the notation of \cite{gupta2014random} we also use the following notation: given an $(n+1)\times (n+1)$ matrix $A$ and a  polynomial $f(x_1,\ldots,x_n)$ we denote $\sigma_{A}(f)\triangleq f^h(A\xn)$, where $f^h$ is the homogenization of $f$. Finally, we define the rank of a homogeneous quadratic polynomial $q$ to be the minimal $k$ such that for some linear forms $\cbra{\ell_i}_{i=1}^{k}$, $q=\ell_1^2+\ldots+\ell_t^2 - \ell_{t+1}^2 - \ldots - \ell_k^2$.

\IncMargin{2em}
\begin{algorithm}
	\SetKwData{Left}{left}\SetKwData{This}{this}\SetKwData{Up}{up}	\SetKwFunction{Union}{Union}\SetKwFunction{FindCompress}{FindCompress}
	\SetKwInOut{Input}{input}\SetKwInOut{Output}{output}
	\SetNlSty{texbf}{\texttt{\bf AFR}}{}
	\Input{Black-box access to an $n$-variate  polynomial $f\in\f[\xn]$ of degree at most $d=2^\Delta$}
	\Output{Either a set of $4^\Delta$ linear functions $\ell_1,\ldots,\ell_{4^\Delta}$ such that $f=\croanf(\ell_1,\ldots,\ell_{4^\Delta})$ or \texttt{Fail}}
	\BlankLine
	If $\Delta=0$ then $f$ is a linear function.  Compute $f$ via interpolation and return the linear function;
	
	\textbf{Homogenization}. Homogenize $f$ (i.e. obtain query access to $f^h$)\label{alg:afr2}; 
	
	\textbf{Reduction to \texttt{LDR}}. Pick $(n+1)$ vectors $\va_0,\ldots,\va_n$, each of whose coordinates are chosen uniformly at random from a large enough subset $T\subseteq\f$. Let $r=127$ and $m=4^{\Delta-1}$. For $r<i<j\leq n$, let $A^{i,j}_r$ be the $(n+1)\times(n+1)$ matrix whose $k$th column  $(k\in[n+1]_0)$ is $\delta_{ijk}\cdot \va_k$ where $\delta_{ijk}$ is $1$ if $k\in\{0,1,\ldots,r\}\cup\{i,j\}$ and $0$ otherwise. For each $A^{i,j}_r$ invoke the LDR algorithm on $\sigma_{A^{i,j}_r}(f)$ (which is an  $r+3$-variate polynomial) to obtain an $m$-tuple, $Q_{i,j}=(q_{i,j,1},\ldots,q_{i,j,m})$, of quadratic polynomials  satisfying
	{\begin{itemize}
			\item $\text{rank}(q_{i,j,l})\leq 4$ for each $\cbra{i,j}\in{\{r+1,r+2,\ldots,n\}\choose 2}$ and $l\in[m]$, and
			\item $\sigma_{A^{i,j}_r}(f)=\croanf[{\Delta-1}](Q_{i,j})$
	\end{itemize}}
	\label{alg:afr3}
	
	\textbf{Patchwork}. Invoke the algorithm of Lemma 6.6 of \cite{gupta2014random} on input $((\va_0,\ldots,\va_n),(q_{i,j})_{r<i<j\leq n})$ and obtain an $m$-tuple of quadratic forms $Q=(q_1,q_2,\ldots,q_m)$\label{alg:afr4};
	
	For each $i\in[m]$, find linear forms $\ell_{i,1},\ell_{i,2},\ell_{i,3},\ell_{i,4}$ such that
	$q_i=\ell_{i,1}\cdot\ell_{i,2}+\ell_{i,3}\cdot\ell_{i,4}$\label{alg:afr5};
	
	\Return $(\ell_{1,1},\ldots,\ell_{1,4},\ell_{2,1},\ldots,\ell_{m,3},\ell_{m,4})$\label{alg:afr6};

	\caption{ANF Formula Reconstruction \texttt{AFR($f(\xn),\Delta$)} (Algorithm 6.9 of \cite{gupta2014random})}\label{alg:afr}
\end{algorithm}
\DecMargin{2em}

\begin{remark}\label{rem:t}
	We note that in Algorithm 5.1 of \cite{gupta2014random} (Algorithm~\ref{alg:ldr}) they treat $f$ as an $(r+1)$-variate polynomial. However the $r$ in their Algorithm 6.9 (Algorithm~\ref{alg:afr}) is not the same $r$ as in Algorithm~\ref{alg:ldr}, specifically, $r_{\texttt{AFR}}=r_{\texttt{LDR}}+2$. Hence,  to avoid confusion, we decided to denote the number of variables in  Algorithm~\ref{alg:ldr}  with $r+3$.	
\end{remark}

\IncMargin{2em}
\begin{algorithm}
	\SetKwData{Left}{left}\SetKwData{This}{this}\SetKwData{Up}{up}	\SetKwFunction{Union}{Union}\SetKwFunction{FindCompress}{FindCompress}
	\SetKwInOut{Input}{input}\SetKwInOut{Output}{output}
	\SetNlSty{texbf}{\texttt{\bf LDR}}{}

	\Input{An $r+3$-variate homogeneous polynomial $f\in\f[\bm Y]$ of degree $d=2^\Delta$ given as a list of coefficients}
	\Output{Either a tuple of $m=4^{\Delta-1}$ quadratic forms $(q_1,\ldots,q_m)$, each of rank $4$, such that $f=\croanf[\Delta-1](q_1,\ldots,q_m)$, or \texttt{Fail}}
	\BlankLine
	If $\Delta=1$ then return $f$ itself;
	
	Let $\text{\normalfont Sing}(f)$ be the ideal generated by the first order derivatives of $f$ - i.e., the ideal
	$$\gen{\partiald{f}{Y_0},\partiald{f}{Y_1},\ldots,\partiald{f}{Y_{r+2}}}.$$
	Use Proposition 4.8 of \cite{gupta2014random} to determine the dimension of $\text{\normalfont Sing}(f)$. If codimension of $\text{\normalfont Sing}(f)$ is not $4$, output \texttt{Fail}. Else, compute a set of generators $g_1,g_2, \ldots,g_l$ for the top dimensional component (of codimension $4$) of $\text{\normalfont Sing}(f)$ using the algorithm of Theorem 4.14 of \cite{gupta2014random}\label{alg:ldr2};
	
	Compute a basis $\{\tilde g_1,\ldots,\tilde g_{t}\}$ for the vector space $V\subseteq\f[\bm Y]$ consisting of all the homogeneous components of degree $\frac{d}{2}$ of each $g_i$ above. If $t=\dim(V)\neq 4$, output \texttt{Fail}\label{alg:ldr3};
	
	By solving an appropriate system of polynomial equations in $4$ unknowns, compute another basis $\{h_1,h_2,h_3,h_4\}$ of $V$ such that the singularities of each $h_i$ has a component of codimension $4$\label{alg:ldr4};
	
	By going over all permutations $\pi:[4]\to[4]$, find one such that $f$ is an $\f$-linear combination of $h_{\pi(1)}\cdot h_{\pi(2)}$ and $h_{\pi(3)}\cdot h_{\pi(4)}$. Compute $\alpha,\beta$ such that
	$f=\alpha h_{\pi(1)}h_{\pi(2)}+\beta h_{\pi(3)}h_{\pi(4)}$.
	Let $\tilde h_1=\alpha h_{\pi(1)}$, $\tilde h_2=h_{\pi(2)}$, $\tilde h_3=\beta h_{\pi(3)}$, $\tilde h_4=h_{\pi(4)}$\label{alg:ldr5};

	For each $i\in[4]$, make a recursive call to \texttt{LDR}$(\tilde h_i,\Delta-1)$ and obtain $Q_i= (q_{i,1},q_{i,2},\ldots,q_{i,4^{\Delta-2}})$ such that $\tilde h_i=\croanf[\Delta-2](q_{i,1},q_{i,2},\ldots,q_{i,4^{\Delta-2}})$ \label{alg:ldr6};
	
	 \Return $Q=Q_1\circ Q_2\circ Q_3\circ Q_4$, where `$\circ $' denotes list concatenation \label{alg:ldr7};
	
	\caption{Low-dimensional formula reconstruction \texttt{LDR}($f(\bm Y),\Delta$) (Algorithm 5.1 of \cite{gupta2014random})}\label{alg:ldr}
\end{algorithm}
\DecMargin{2em}

\newpage
\subsection{Definition of Formulaic Independence and Pairwise Singular Independence}



In \cite{gupta2014random} Gupta et al.  characterize ``bad'' inputs to their average-case, randomized algorithm in terms of points in a specific variety. As we only stated their algorithm over the complex numbers, we define varieties only over $\C$. However, all definitions can be easily extended to other fields as well.

	For any set of $n$-variate polynomials $\cF\subseteq\C[\xn]$, we define the \emph{zero set} of $\cF$ as:
	$$ V(\cF)\triangleq \{\va\in\C^n\;|\;\forall f\in \cF:f(\va)=0\}\;. $$
	Any set $V\subseteq\C^n$ that can be defined as a zero set $V=V(\cF)$ for some set of polynomials $\cF\subseteq\C[\xn]$ is called a \emph{variety}, or an \emph{algebraic set}.



The notions ``Formulaic Independence'' and ``Pairwise Singular Independence'' are defined in terms of dimensions of \emph{projective varieties}, as the polynomials in question are always homogeneous.

	Let $r\in\n$. The \emph{r-dimensional projective space} $\mathbb{P}^r$ is the space $\C^{r+1}\setminus\{\vzero\}$ with the equivalence relation $\sim$, where $\vv,\vu\in\C^{r+1}\setminus\{\vzero\}$ satisfy $\vv\sim\vu$ if and only if there exists some $\lambda\in\C$ such that $\lambda\vv=\vu$.

If $V=V(f_1,\ldots,f_k)$ is a variety where every $f_i$ is an $r+1$-variate homogeneous polynomial, and if $\vv\in\C^{r+1}$ satisfies $f_1(\vv)=\ldots=f_k(\vv)=0$, then for every $\lambda\in\C$: $f_1(\lambda \cdot\vv)=\ldots=f_k(\lambda \cdot\vv)=0$. Thus, the set $V\setminus\{\vzero\}$ can be viewed as a subset of $\mathbb{P}^r$. In this case we call $V$ a \emph{projective variety}, and define its dimension as follows:

\begin{definition}[Proposition 11.4 in \cite{harris2013algebraic}]
	The dimension of a projective variety $V\subseteq\mathbb{P}^r$, denoted $\dim(V)$, is the largest integer $k$ such that any linear space of dimension $\geq r-k$ intersects $V$ nontrivially.
\end{definition}

The definition of \emph{formulaic independence} involves the algebraic set of \emph{singularities} of a polynomial $f$, and the \emph{Jacobian matrix} of a tuple of polynomials:
%
	For a polynomial $f\in\C[\xn]$, the set of \emph{singularities} of $f$ is the set of points $\vv\in\C^n$ such that $f(\vv)=\bra{\partiald{f}{x_1}}(\vv)=\bra{\partiald{f}{x_2}}(\vv)=\ldots=\bra{\partiald{f}{x_n}}(\vv)=0$. In other words, 
	$$ \text{\normalfont Sing}(f)=V\bra{f,\partiald{f}{x_1},\ldots,\partiald{f}{x_n}}\;. $$

	Given a tuple of polynomials $\bm f=(f_1,\ldots,f_m)\in\C[\xn]^m$, the \emph{Jacobian} of $\bm f$ is the following matrix of partial derivatives of $f_1,\ldots,f_m$:
	$$ J(\bm f,\xn)=\begin{pmatrix}
		\partiald{f_1}{x_1}&\partiald{f_1}{x_2}&\cdots&\partiald{f_1}{x_n}\\
		\partiald{f_2}{x_1}&\partiald{f_2}{x_2}&\cdots&\partiald{f_2}{x_n}\\
		\vdots&\vdots&\ddots&\vdots\\
		\partiald{f_m}{x_1}&\partiald{f_m}{x_2}&\cdots&\partiald{f_m}{x_n}
	\end{pmatrix}\in\C[\xn]^{m\times n}\;. $$

\begin{definition}[Definition from Section 3.1 of \cite{gupta2014random}]\label{defMinors}
	Let $M(\xn)\in\C[\xn]^{s\times r}$ be a matrix whose entries are polynomials in $\xn$,  and let $t\in\n$. We denote by $\text{\normalfont Minors}(M(\xn),t)\subseteq\C[\xn]$ the set of determinants of all $t\times t$ submatrices of $M(\xn)$.
\end{definition}

\begin{definition}[Definition 5.2 of \cite{gupta2014random}]\label{defVJ}
	Let $\bm g=(g_1(\xn),\ldots,g_k(\xn))\in\C[\xn]$ be a $k$-tuple  of homogeneous polynomials. The algebraic set $V_J(g_1,\ldots,g_k)$ ($V_J(\bm g)$ for short) is defined to be 
	the set of common zeroes of polynomials in $\text{\normalfont Minors}(J(\bm g,\xn),k)$. In other words, $V_J(\bm g)$ consists of all points $\vv\in \mathbb{P}^r$ for which the rank
	of the Jacobian matrix $J(\bm g, \xn)$ is less than $k$.
\end{definition}

\begin{definition}[Formulaic Independence, Definition 5.3 of \cite{gupta2014random}]\label{defFI}
	Let $\xn=(x_0,x_1,\ldots,x_r)$ and let $f,f_1,f_2,f_3,f_4\in\C[\xn]$ such that $f=f_1\cdot f_2+f_3\cdot f_4$. Denote $\bm f\triangleq (f_1,f_2,f_3,f_4)$. We say that $f_1,f_2,f_3,f_4$ are \emph{formulaically independent} if
	$ \dim(V(\bm f))=r-4$ and $\dim(\text{\normalfont Sing}(f)\cap V_J(\bm f))<r-4$.
	We say that a homogeneous ANF formula $\Phi$ \emph{satisfies formulaic independence} at node $v$ if  $v$ is a $+$ gate, and the four polynomials computed at the grandchildren of $v$ are formulaically independent.
\end{definition}

To define \emph{pairwise singular independence}, we must first define the \emph{iterated Jacobian matrix}:

\begin{definition}[The Iterated Jacobian and the variety $V_I$, Definition 5.19 of \cite{gupta2014random}]\label{defIterJacobian}\fussy
	Let $\xn=(x_0,x_1,\ldots,x_r)$, and let $\bm{g_1},\ldots,\bm{g_k}\in(\C[\xn])^m$ be $m$-tuples of homogeneous, $(r+1)$-variate polynomials: $\bm{g_i}=g_{i,1},\ldots,g_{i,m}$. The \emph{iterated Jacobian} of $(\bm{g_1},\ldots,\bm{g_k})$, denoted $I(\bm{g_1},\ldots,\bm{g_k})$, is defined to be the following matrix: $I(\bm{g_1},\ldots,\bm{g_k})\in\C[\xn]^{{r+1\choose k}\times m^k}$ has its rows indexed by $k$-sized subsets of indices of variables $\{j_1,\ldots,j_k\}\in{[r+1]_0\choose k}$ and its columns indexed by tuples $(i_1,\ldots,i_k)\in[m]^k$. The $(\{j_1,\ldots,j_k\},(i_1,\ldots,i_k))$th entry of $I(\bm{g_1},\ldots,\bm{g_k},\xn)$ is the polynomial
	$$ \det\begin{pmatrix}
		\partiald{g_{1,i_1}}{x_{j_1}}&\partiald{g_{2,i_2}}{x_{j_1}}&\cdots&\partiald{g_{k,i_k}}{x_{j_1}}\\
		\partiald{g_{1,i_1}}{x_{j_2}}&\partiald{g_{2,i_2}}{x_{j_2}}&\cdots&\partiald{g_{k,i_k}}{x_{j_2}}\\
		\vdots&\vdots&\ddots&\vdots\\
		\partiald{g_{1,i_1}}{x_{j_k}}&\partiald{g_{2,i_2}}{x_{j_k}}&\cdots&\partiald{g_{k,i_k}}{x_{j_k}}
	\end{pmatrix} \;. $$
\fussy	The algebraic set $V_I(\bm{g_1},\ldots,\bm{g_k})$ is defined to be the common zeroes of the polynomials in $\text{\normalfont Minors}(I(\bm{g_1},\ldots,\bm{g_k}),\ell^k)$.
\end{definition}

\begin{definition}[Pairwise Singular Independence, Definition 5.20 of \cite{gupta2014random}]\label{defPSI}
	Let $\{f_{i,j}\}_{i,j=1}^4\subseteq\C[\xn]$ be sixteen homogeneous, $(r+1)$-variate polynomials of the same degree. For every $i\in[4]$, let $f_i=f_{i,1}\cdot f_{i,2}+f_{i,3}\cdot f_{i,4}$ and $\bm{f_i}=(f_{i,1},f_{i,2},f_{i,3},f_{i,4})$. For a set $S=\{i_1,\ldots,i_k\}\subseteq[4]$, denote:
	$ W_S\triangleq V_J(f_{i_1},\ldots,f_{i_k})\cap V_I(\bm{f_{i_1}},\ldots,\bm{f_{i_k}})$.
	We say that $(\bm{f_1},\bm{f_2},\bm{f_3},\bm{f_4})$ are \emph{pairwise singularly independent} if
	\begin{enumerate}
		\item for all $1\leq i<j\leq 4$:
		$ \dim(\text{\normalfont Sing}(f_i)\cap\text{\normalfont Sing}(f_j))\leq r-6,\qquad$ and
		\item for all $S\subseteq[4]$ such that $|S|\geq 2$:
		$ \dim(W_S)\leq r-6 $.
	\end{enumerate}
	We  say that  a homogeneous ANF formula $\Phi$ satisfies \emph{pairwise singular independence} at a node $v$ if the node $v$ is a $+$ gate, and
	 $(\bm{f_{v_1}},\bm{f_{v_2}},\bm{f_{v_3}},\bm{f_{v_4}})$ are pairwise singularly independent, where $v_1,v_2,v_3,v_4$ are nodes which are the grandchildren of $v$ and $\bm{f_{v_i}}$ is the $4$-tuple of polynomials computed at the grandchildren of the node $v_i$.
\end{definition}

\end{document}